\PassOptionsToPackage{dvipsnames,table}{xcolor}

\documentclass[sigconf,nonacm]{acmart}

\usepackage{graphicx}
\usepackage{subcaption}
\usepackage{multirow}
\usepackage{balance}
\usepackage{cleveref}
\usepackage{caption}
\usepackage{algorithm}
\usepackage[normalem]{ulem}
\usepackage[italicComments=true,rightComments=true,commentColor=gray]{algpseudocodex}
\makeatletter
\renewcommand{\ALG@beginalgorithmic}{\small}
\makeatother
\usepackage{mathtools}
\usepackage{amsmath}
\usepackage{booktabs}
\usepackage{array}
\usepackage{soul}
\usepackage{tikz}
\usepackage[framemethod=tikz]{mdframed}
\usepackage{amsthm}
\usepackage{xspace}
\definecolor{forestgreen}{RGB}{34,139,34}
\usepackage[dvipsnames,table]{xcolor}
\usepackage[most]{tcolorbox}
\usepackage{balance}
\usepackage{enumitem}
\usepackage{ifthen}

\newboolean{submit}
\setboolean{submit}{true} %

\newboolean{rag}
\setboolean{rag}{false} %

\usepackage{titlesec}
\titlespacing*{\section}{0pt}{1.0ex}{0.8ex}
\titlespacing*{\subsection}{0pt}{0.8ex}{0.6ex}
\titlespacing*{\subsubsection} {0pt}{0ex}{4pt}

\makeatletter
\AtBeginDocument{%
  \let\ACM@origsection\section
  \let\ACM@origsubsection\subsection
  \let\ACM@origsubsubsection\subsubsection
  \let\ACM@origparagraph\paragraph
}
\makeatother

\setlist[itemize]{noitemsep,nolistsep}

\lstdefinestyle{plaintextstyle}{
  breaklines=true,
  breakatwhitespace=false,
  basicstyle=\footnotesize\ttfamily,
  columns=flexible,
  keepspaces=true,
  showstringspaces=false,
  frame=single,
  framesep=2pt,
  numbers=left,
  numbersep=5pt,
  numberstyle=\tiny\color{gray},
  xleftmargin=10pt,  %
  framexleftmargin=10pt,
  xrightmargin=3pt,
  resetmargins=true, %
  aboveskip=10pt,    %
  belowskip=10pt,    %
  gobble=0           %
}

\lstdefinelanguage{text}{
  identifierstyle=,
  keywordstyle=,
  commentstyle=,
  stringstyle=,
  morekeywords={}
}

\newcommand{\topic}[1]{\vspace{.5pt} \noindent{\bf #1.}}

\DeclareCaptionFormat{empty}{#1}

\newcommand*\circled[1]{\tikz[baseline=(char.base)]{
            \node[shape=circle,draw,inner sep=0.5pt] (char) {\small #1};}}

\newcolumntype{L}{l}
\newcolumntype{C}{>{$}c<{$}}
\newcolumntype{R}{>{$}r<{$}}

\newcommand{\sys}{\textsc{SHED}\xspace}

\newcommand{\code}[1]{\texttt{\small #1}}
\newcommand{\ds}[1]{\textsf{\textsc{\small #1}}}

\newcommand{\best}[1]{\textcolor{green!50!black}{{\bf #1}}}
\newcommand{\diff}[1]{\textcolor{red!70!black}{{\bf \em  #1}}}

\newlength{\tinyskipamount}
\newcommand{\tinyskip}{\vspace{\tinyskipamount}}

\newcounter{definition}
\newenvironment{definition}[1][]{\refstepcounter{definition}\par\tinyskip\textsc{Definition~\thedefinition.\ #1}}{\tinyskip}

\newcounter{example}
\newenvironment{example}[1][]{\refstepcounter{example}\par\tinyskip\textsc{Example~\theexample.\ #1}}{$\square$\tinyskip}

\newcounter{theorem}
\newenvironment{theorem}[1][]{\refstepcounter{theorem}\par\tinyskip\textsc{Theorem~\thetheorem.\ #1}}{\tinyskip}

\newenvironment{cleanproof}[1][]{\par\textsc{Proof.\ #1}}{\qedsymbol\endproof}
\newenvironment{cleanproofsketch}[1][]{\par\textsc{Proof Sketch.\ #1}}{\qedsymbol\endproof}

\newenvironment{myproof}
{%
  \ifthenelse{\boolean{submit}}
  {\expandafter\comment}
  {\begin{cleanproof}}
}
{%
  \ifthenelse{\boolean{submit}}
  {\expandafter\endcomment}
  {\end{cleanproof}}
}

\newcommand{\hide}[1]{}
\newcommand{\hidec}[1]{\ifthenelse{\boolean{true}}{}{}}
\newcommand{\blank}[1]{\textcolor{red}{\textbf{\footnotesize}}}
\newcommand{\keep}[1]{#1}
\newcommand{\keepc}[1]{\textcolor{black}{#1}}
\newcommand{\ttt}[1]{\texttt{\small #1}}

\AtBeginDocument{%
  }

\makeatletter
\def\@ACM@resetaffil{\global\@ACM@instpresentfalse\global\@ACM@citypresentfalse}
\makeatother

\begin{document}

\title{Robust Hierarchical Structures for Agentic Document Analysis}

\author{Ruiying Ma}
\affiliation{%
  \institution{UC Berkeley}
}
\email{ruiyingm@berkeley.edu}

\author{Yiming Lin}
\affiliation{%
  \institution{UC Berkeley}
}
\email{yiminglin@berkeley.edu}

\author{Aditya G. Parameswaran}
\affiliation{%
  \institution{UC Berkeley}
}
\email{adityagp@berkeley.edu}

\renewcommand{\shortauthors}{Ma et al.}

\begin{abstract}
Large Language Models (LLMs) enable us
to better understand text
documents, including \hide{}PDFs and Word documents.
However, LLMs, as well as more modern
LLM agents, i.e., those with tool-calling abilities,
typically treat such documents as plain text,
ignoring the fact that
they are often organized hierarchically
into sections and subsections.
Extracting this structure, while difficult,
can improve efficiency and effectiveness
for agents (and humans)---since
only sections relevant to a given task need to be processed.
Unfortunately, prior work on structure extraction
provides no formal guarantees
on how well the inferred structure matches the true one.
Instead, we target a {\em robust} and {\em compact} variant
that is feasible to infer and useful in practice.
Robustness ensures that
the text under each subsection header
is a superset of
the text under the same header
in the true structure.
Compactness seeks to minimize this superset,
reducing agentic cost (or human cognitive load).
We propose \sys,
a two-stage workflow
for inferring a robust and compact structure.
The first stage is pluggable with an infinite family of approaches,
each guaranteeing robustness
for a specific document class.
We theoretically characterize the document space
using these classes and their hierarchical relationships.
Empirically, \sys improves F-1 scores (measuring the robustness–compactness trade-off) by
\hide{}
13\%--68\% over non-LLM baselines and 9\%--15\% over expensive LLM-based approaches.
Finally, we show how \sys-inferred structures are valuable for
agentic document analysis: agents using \sys
outperform baselines,
achieving \hide{}\keep{3\%--23\%} higher accuracy
while being up to \hide{}\keep{$10\times$} cheaper.

\end{abstract}

\maketitle

\section{Introduction}
\label{sec:intro}

\begin{figure*}[t]
\centering
    \begin{subfigure}[b]{1\linewidth}
         \centering
         \includegraphics[width=1\textwidth]{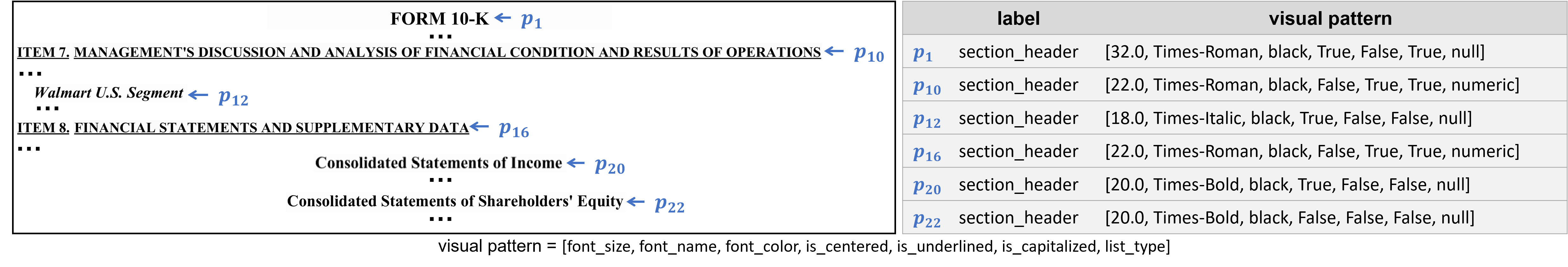}
         \vspace{-2em}
         \caption{The header phrases and their visual patterns.}
         \label{fig:10k}
    \end{subfigure}
    \begin{subfigure}{0.54\linewidth} %
         \centering
         \raisebox{0.3cm}
         {\includegraphics[width=1\textwidth]{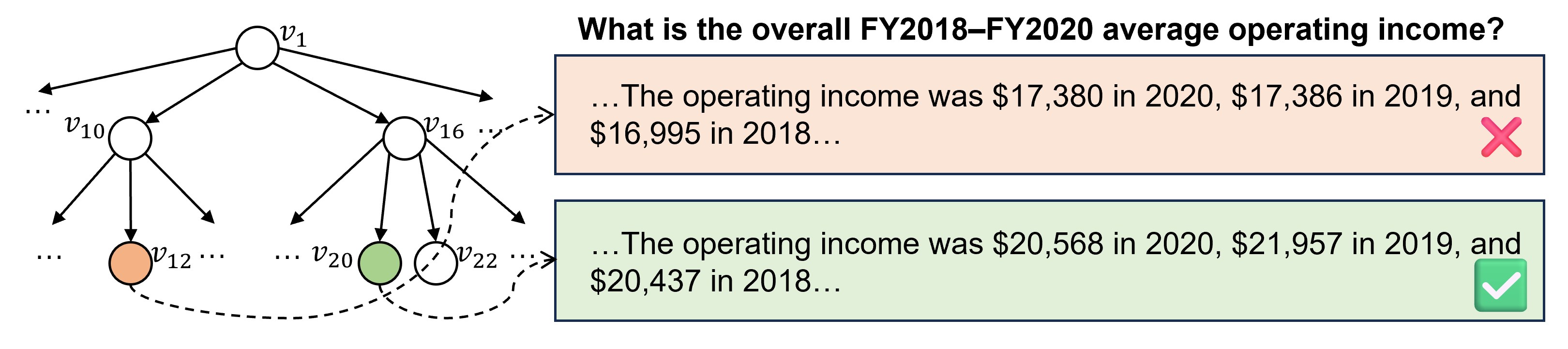}}
         \vspace{-2em}
         \caption{The true SHT.}
         \label{fig:true_sht}
    \end{subfigure}
    \begin{subfigure}{0.2\linewidth} %
         \centering
         \raisebox{0.6cm}
         {\includegraphics[width=0.9\textwidth]{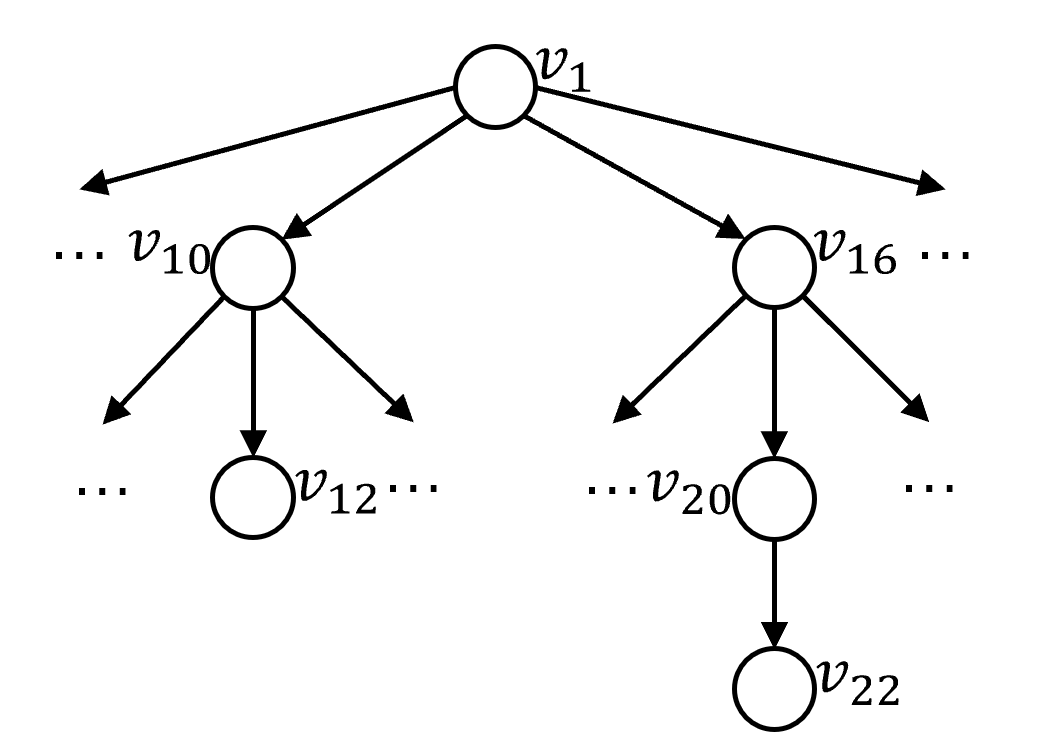}}
         \vspace{-2em}
         \caption{A robust SHT.}
         \label{fig:robust_sht}
    \end{subfigure}
    \begin{subfigure}{0.25\linewidth}
         \centering
         \includegraphics[width=1\textwidth]{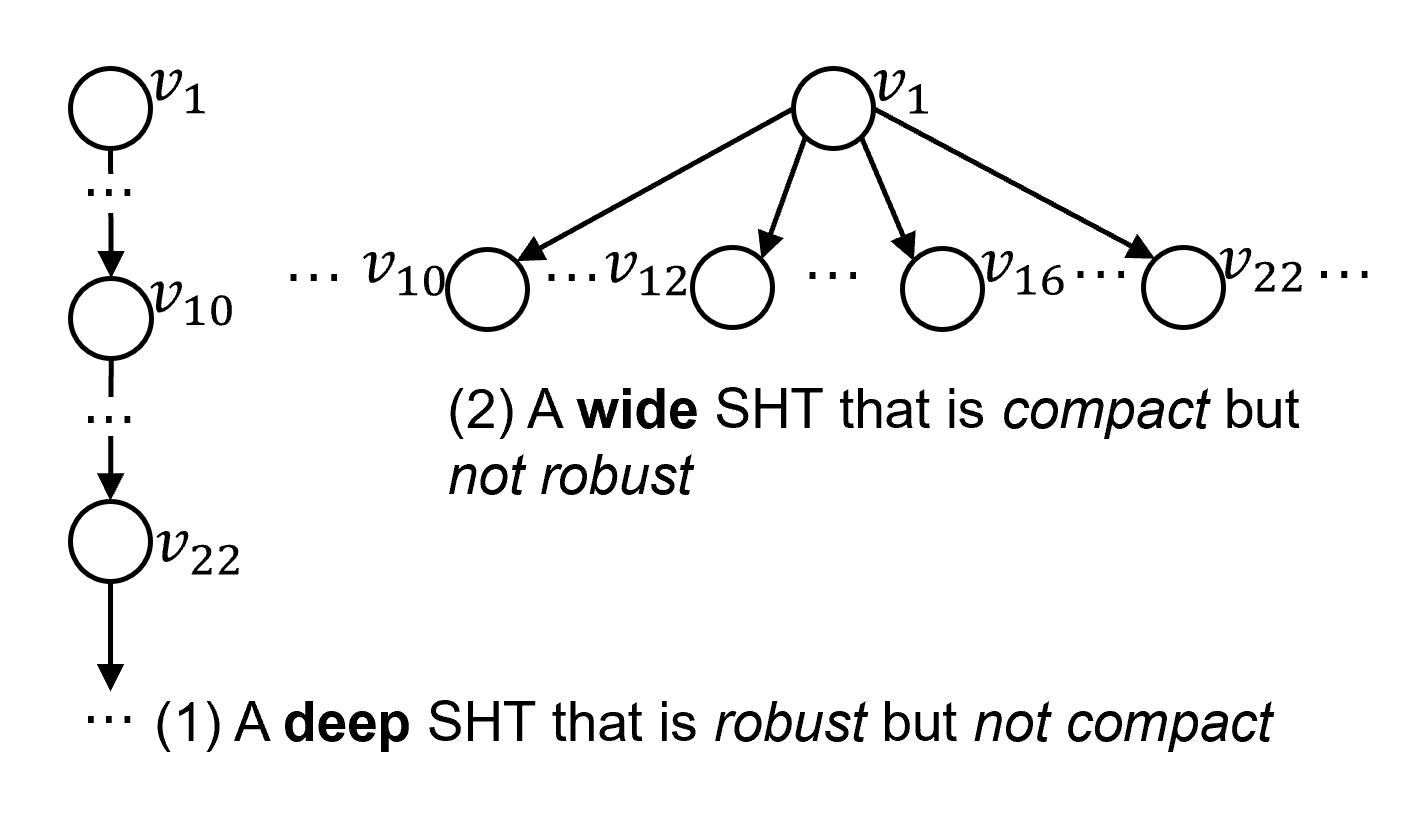}
         \vspace{-2em}
         \caption{Deep and wide SHTs.}
         \label{fig:deep_and_wide_shts}
    \end{subfigure}
    \caption{Hierarchical structures of a financial document {Walmart Inc. form 10-K}. The node $v_i$ in an SHT corresponds to the header phrase $p_i$. }
    \label{fig:filing}
\end{figure*}

\hide{}Thanks to Large Language Models (LLMs), we are seeing a resurgence of
interest in automatically processing
the vast volumes of unstructured data in our
organizations, including
PDF and Word documents~\cite{unstructured1, unstructured2}.
Modern {\em agentic} LLMs
with tool-calling~\cite{yao2022react},
and systems built on them,
such as Claude Code~\cite{claude-code}, Codex~\cite{codex}, and Claude Cowork~\cite{cowork},
are rapidly emerging as a dominant
paradigm for reading, reasoning over,
and querying long documents,
a setting we term {\em agentic document analysis}.
Yet processing long documents remains
challenging for agents:
reading them in full
incurs high token costs
and causes accuracy to degrade with context (i.e., input) length,
a phenomenon known as ``context rot'' or ``lost-in-the-middle''~\cite{lost-in-the-middle, context-rot}.
One solution is to simply provide more
tools to the agents, such as
regular expression-based search
within documents, e.g., via \code{grep}~\cite{sun2025docagent},
or embedding-based search~\cite{agentic-rag-survey},
which allow agents to retrieve and process portions of documents
rather than all of it at once.
However, such tools treat documents as a {\em bag of words},
ignoring their underlying hierarchical structure
in the form of sections and subsections.
Humans rely on this hierarchy---essentially a table of contents---to understand a document at a glance and navigate it efficiently; agents should too.
\keep{Moreover, this hierarchy is common across domains: e.g., as of 2025, 17M+ SEC EDGAR filings~\cite{edgar}, 530K+ ClinicalTrials
reports~\cite{clinicaltrials}, and 4M+ DOE OSTI reports~\cite{osti} have such structures.}
We therefore ask the question:
{\em how can we automatically and correctly identify
this hierarchical structure---and can agents (and humans)
make use of this structure for efficient and effective document analysis?}

\begin{example}[Impact of Hierarchical Structure.]
    Consider a Walmart Inc.~10-K document with 170 pages in Figure~\ref{fig:filing}.
    Its hierarchical structure
    can be represented as a tree shown
    in Figure~\ref{fig:true_sht}, where nodes represent headers,
    and edges represent semantic (containment) relationships.
    Following Lin et al.~\cite{zendb},
    we call this tree a {\em Semantic Hierarchical Tree (SHT)}.
    For example, $v_i$ in Figure~\ref{fig:true_sht}
    corresponds to the header phrase $p_i$ in Figure~\ref{fig:10k},
    while edge $v_{16} \rightarrow v_{20}$ indicates
    that the document portion
    headed by {\em Consolidated Statements of Income}
    (corresponding to $v_{20}$)
    is a subsection under the section titled {\em ITEM 8...}
    represented by $v_{16}$.
    Consider the
    question {\em “What is the overall FY2018–FY2020 average operating income?”}.
    The correct answer is in the subsection corresponding to $v_{20}$,
    with the header {\em ``Consolidated Statements of Income''},
    and content about {\em operating income}---retrieving and processing
    this entire subsection\hide{}
    is sufficient to answer this question.
    However, an agent without the document hierarchy
    may use \code{grep} to retrieve portions of the content
    underlying $v_{12}$ and $v_{20}$
    (the two boxes in Figure~\ref{fig:true_sht}), since
    the content is similar to the question, even though
    $v_{12}$ is a false positive,
    since it corresponds to the U.S.~segment, and not overall.
    Without knowing the structure---and specifically the headers---for
    $v_{12}$ and $v_{20}$, it is impossible for the agent
    to know which one is accurate.
    Even worse, in other cases, the document portion retrieved
    may not be understandable
    in isolation (unlike section/subsection-level units
    that are more logically complete),
    requiring the agent to make a guess based on incomplete information.
\end{example}

\noindent Like agents, human understanding of long documents hinges
on SHTs~\cite{human-reading-linear-text},
improving the accessibility of
scientific~\cite{hiqa},
legal~\cite{sht-for-legal},
and medical~\cite{sht-for-med} documents.
An SHT
not only aids humans in locating which
subsections they should review (top-down), but
also provides them with {\em context} for where they are in the document
when reviewing a given text portion (bottom-up)~\cite{cognitive-sht, cognitive-headings}.
For example, when reading the text containing {\em ``operating income''} under $v_{20}$, the headers are essential to understand what the text refers to.

\topic{Lack of guarantees in prior approaches}
Past work on SHT inference from documents \hide{}is based on LLMs~\cite{contextgem, llmaparse}, custom models~\cite{ml-sht-1, ml-sht-2}, or rules~\cite{ pymupdf, pdfminer},
of which LLM-based approaches are currently considered
state-of-the-art.
Unfortunately, {\em prior approaches
for inferring SHTs provide
no real guarantees} for how well such SHTs match
the ``true'' underlying SHT,
except in very specific settings that are rare in practice,
as we will discuss later~\cite{zendb}.
In particular,
LLM-based approaches similarly provide no guarantees, and also are expensive. In fact, as we will see later, they sometimes
hallucinate, introducing headers that don't exist in the underlying
document.
Overall, identifying the true SHT is unrealistic
due to heterogeneous layouts, semantic ambiguity,
and length~\cite{sht-inference-challenge}.

\topic{Robustness for SHTs}
Instead of inferring the true SHT,
we introduce a new objective of
{\em robustness}.
A robust SHT is one where,
for each (section or subsection)
header, we require the text ``under it'' to be
a superset of the corresponding text in the
true SHT.
Robustness is essential for
correctness, as discussed below:
\begin{example} [Robust Retrieval.]
    Figure~\ref{fig:robust_sht} shows a robust SHT;
    the document portion corresponding to $v_{22}$
    is now part of that of $v_{20}$,
    as the edge $v_{16} \rightarrow v_{22}$
    in the true SHT (\autoref{fig:true_sht})
    becomes a direct path connecting $v_{16}$ to $v_{22}$.
    When
    answering the question in Figure~\ref{fig:true_sht}
    using this robust SHT,
    suppose $v_{20}$ is retrieved
    and its document portion (where the correct answer resides)
    is returned to answer the question.
    Here, in addition to the portion
    providing the correct answer (i.e., one corresponding to $v_{20}$
    in the true SHT), the extra portion corresponding to $v_{22}$
    is also returned, which is permissible
    as long as all true portions
    (the ones corresponding to $v_{20}$) are included---as
    agents are
    able to filter out or ignore extraneous irrelevant content.
    Moreover, retrieving $v_{20}$ as a complete subsection, rather than keyword-matched fragments, ensures that the returned content is self-contained and interpretable. %
\end{example}

\topic{Compactness for SHTs}
While a robust SHT
guarantees that the returned context
for any given header (i.e., the text portions underneath)
is a superset of the true one,
we also want to {\em minimize the size of the returned context}
to reduce either processing cost for agents or human cognitive load,
depending on the setting---we call this property {\em compactness}.
Consider the deep SHT in Figure~\ref{fig:deep_and_wide_shts}-1,
which forms a chain of nodes with a single leaf.
This tree is robust but not compact
(e.g., retrieving the context of $v_{10}$
returns all the content from $p_{10}$
to the end of the document).
In contrast, the wide SHT in
Figure~\ref{fig:deep_and_wide_shts}-2,
where all nodes except $v_1$ are in the same layer,
is compact but not robust
(e.g., the document portion of $v_{16}$
no longer contains that of $v_{22}$,
and $v_{16}$ is no longer an ancestor of $v_{22}$).
{\em We therefore aim to infer SHTs {\bf\em that are guaranteed to be robust,
while being as compact as possible},
ensuring correctness while reducing cost}.
(This is analogous to ensuring 100\% recall,
while maximizing precision, when retrieving
document subsections.)

\topic{Inferring a robust and compact SHT via \sys}
Inferring an SHT
that guarantees robustness
and is compact is non-trivial.
While current tools can
identify headers accurately~\cite{vgt-paper, lightgbm}, i.e., section or subsection headers,
corresponding to nodes in the tree,
they fail to assemble these (header) nodes
into a robust tree.
We propose \sys\footnote{\textbf{S}emantic \textbf{H}ierarchical Structure \textbf{E}xtraction for \textbf{D}ocuments.}, a framework
for inferring
a robust and compact hierarchical structure for documents\keep{~with visually identifiable structures, using only intrinsic document elements (e.g., (sub)section headers) without any external ontology}.

Intuitively, each
header has specific human-interpretable visual patterns in
the form of numbering, font, size, and capitalization.
Headers with identical visual patterns
are typically of the same type,
and therefore belong in similar ``levels'' in the SHT,
e.g., section headers are often formatted similarly,
and would be at similar levels in the SHT.
Taking phrases labeled as headers as input nodes,
\sys employs a two-stage workflow:
it first groups nodes with identical visual patterns
into clusters and infers relationships among these clusters;
it subsequently derives edges among nodes
from the cluster relationships to assemble a compact SHT.

The first stage is {\em semantic depth (SD) inference},
where \sys assigns each cluster a {\em semantic depth}.
These semantic depths may differ from the depths in the true SHT,
but their relative ordering indicates hierarchical relationships:
the cluster of a node’s parent
must have a smaller semantic depth than
that of the node.
We introduce two approaches for SD inference:
{\em local-first},
which uses the header immediately prior
to a cluster’s first occurrence
to infer its semantic depth,
and {\em global-first},
which uses all headers prior to that occurrence.
The latter generalizes Lin et al.'s approach~\cite{zendb},
which only focuses on inferring the true SHT
rather than robust variants,
and only considers a narrow
class of {\em well-formatted} documents that,
as we will show in \Cref{sec:robustness-analysis},
are rare, occupying $<$25\% of real-world documents.
The second stage,
{\em SHT assembly},
constructs a compact SHT
by deriving edges from the relative ordering of the inferred semantic depths,
with tree depth
bounded by the number of clusters,
independent of the SD inference approach used in the first stage.
{\bf \em  \sys (with local-first) improves
robustness and compactness
by 13\%--68\% and 9\%--15\% in F-1 scores
over non-LLM baselines and expensive LLM-based approaches, respectively}.

\topic{Characterization of Robustness Classes}
\begin{figure}[t]
    \centering
    \includegraphics[width=1\linewidth]{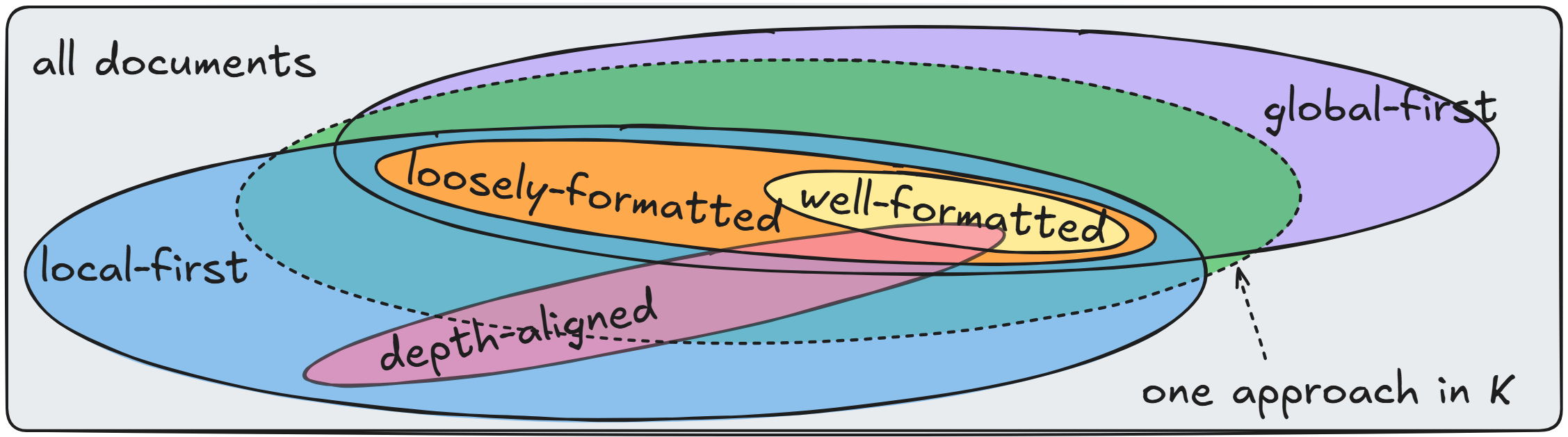}
    \caption{Our contribution: Hierarchy of robustness classes.}
    \label{fig:hierarchy-robustness}
\end{figure}
While our two-stage framework
results in high robustness
and compactness as we see in our experiments, we wanted
to additionally characterize the space of documents
for which we can guarantee robustness for each approach,
and also go beyond the two SD inference approaches,
local- and global-first.
To do so,
we establish
a necessary and sufficient condition
for robustness
for each approach,
a property we call {\em SD conformance}.
Leveraging SD conformance,
we construct a Venn-diagram of relationships
for various approaches, as shown in  \autoref{fig:hierarchy-robustness}---this diagram shows the sets of documents for which each approach guarantees robustness.
We prove that neither local- nor global-first dominate,
but there is a class of documents we call {\em depth-aligned},
for which local-first guarantees robustness, but global-first does not.
We further generalize beyond local- and global-first
to
{\em an infinite family of approaches $\mathcal{K}$,
each guaranteed to return a robust SHT}
for a specific class of documents
while optimizing for compactness.
We additionally define
a class of {\em loosely-formatted} documents
which lies at the intersection of all approaches in $\mathcal{K}$,
including local- and global-first,
and extends far beyond {\em well-formattedness} (by $2.57\times$ empirically).
To our knowledge, this is {\bf \em the first formal characterization
of various document classes and effectiveness of the corresponding approaches}.

\topic{Application to Agentic Document Analysis}
We demonstrate the practical usefulness
of robust and compact SHTs for {\em agentic document analysis}\keep{, whose core needs are what we target, motivating \sys's formulation, guarantees, and evaluations.
Specifically, agents need structure for accurate and cost-efficient analysis of long, complex documents, and \sys precisely targets such structures: robustness (the text under each header is a superset of the true text) ensures that the agent gets all relevant content, while compactness bounds context size to mitigate context rot and provide lower cost.}
By providing an agent with the hierarchical structure
inferred by \sys, the agent can retrieve
the appropriate section or subsection at a time,
{\bf \em outperforming agents without structures or
those with baseline structures by \hide{}\keep{3\%--23\%} in average accuracy,
while reducing total
costs by as much as \hide{}\keep{$10\times$} across four datasets}.

\vspace{2pt}
\noindent
We make the following contributions in this paper:
\begin{itemize}[leftmargin=*, topsep=0pt]
\itemsep0em
\item We formalize the novel concepts of robustness and compactness
of SHTs for document analytics
(\Cref{sec:defn})
and design \sys,
a two-stage workflow with a first stage
that is
pluggable with an infinite family of approaches (\Cref{sec:sht-construction}--\ref{sec:robustness-analysis}).
\item We theoretically characterize various classes of documents
for which \sys guarantees robustness
while optimizing compactness\hide{} (\Cref{sec:robustness-analysis}).
\item Across four datasets spanning various domains and layouts (\Cref{sec:eval_setup}),
for SHT inference, \sys outperforms non-LLM baselines by 13\%--68\%,
and significantly
more expensive LLM-based approaches by 9\%--15\%
in F-1 score, measuring the robustness-compactness trade-off
(\Cref{sec:sht_eval}).
\item We demonstrate the practical usefulness of \sys
for agentic document analysis,
where providing hierarchical structure to agents
improves average accuracy by \hide{}\keep{3\%--23\%}
while reducing total costs by up to \hide{}\keep{$10\times$}
over baselines across four datasets (\Cref{sec:doc_qa}).
\keep{\item We additionally show that \sys's scalability on changing number of documents or complexity is satisfactory. \sys still infers robust and compact SHTs, and has high QA accuracy at low cost, outperforming baselines by up to $4.5\times$ (\Cref{sec:scalability}).}
\end{itemize}

\section{Robust Semantic Hierarchical Trees}
\label{sec:defn}

We focus on documents
with
rich visual formatting
(e.g.,
PDF and Word files).
We first review relevant concepts related to documents, including semantic hierarchical trees, as defined in~\cite{zendb}, in \Cref{subsec:preliminaries}, and then define robust and compact variants in \Cref{sec:robustness-analysis}.

\subsection{Preliminaries}
\label{subsec:preliminaries}

\topic{Documents and Phrases}
Given a document $D$, we operate on a plain text serialized representation as a list of phrases, $\small D = \{p_1,p_2,...,p_n\}$, indexed based on where they appear in $D$. Each phrase $p_i$ is associated with a label $\small l_i \in \{\code{section\_header}, \code{text}, \code{table},\\ \code{image}, \ldots\}$.
This
labeling
can be performed at no cost
using
open-source document recognition tools~\cite{vgt-paper, lightgbm, mineru-3},
which typically extract phrases in a top-down, left-to-right manner,
consistent with the human reading order. Although different tools may define different label sets, labels such as \code{section\_header} (which broadly refers to the headers such as titles, section and subsection headers)
and \code{text} (representing content) are common, making it feasible to identify the hierarchical structure of a document as described below.
We say a phrase is a {\em header phrase} if its label is \code{section\_header}.

\topic{Visual Pattern}
Humans typically infer hierarchical structures from visual features of the header phrases such as font and layout position. For example, header phrases in larger, bold fonts that are centered may indicate section headers, whereas header phrases in smaller, non-bold fonts that are not centered may indicate subsection headers. We similarly rely on the {\em visual patterns} of phrases, defined below, to infer hierarchical structures.
For each phrase $p_i \in D$, let $vp(p_i)$ be the visual pattern associated with $p_i$, encoding various visual features.
$vp(p_i)$ is represented as a feature vector, with features such as $\tiny [\code{font\_size}, \code{font\_name}, \code{font\_color}, \code{is\_underlined}, \\\tiny \code{is\_centered}, \code{is\_capitalized}, \code{list\_type}, \ldots]$.
Here, {\tiny \code{font\_size}}, {\tiny \code{name}}, and {\tiny \code{color}} represent the style of the phrase. The remaining features capture aspects such as whether the phrase is underlined,  center-aligned, fully capitalized, or part of a numeric or alphabetical list.

\topic{Semantic Hierarchical Tree}
We define a Semantic Hierarchical Tree (SHT) for a document $D$, given phrases $p_i\in D$, their labels $l_i$, and visual patterns $vp(\cdot)$.

\begin{definition} [Semantic Hierarchical Tree.]~\label{def:sht}
    An SHT for a document is a directed rooted tree $T = (V, E)$, with a one-to-one correspondence between the nodes $V$ and the header phrases. Let $v_i$ denote the node corresponding to the header phrase $p_i$ with index $i$.
    For any two nodes $v_i, v_j \in V$, where $i < j$, $v_i$ must precede $v_j$ in the pre-order traversal of the tree.
\end{definition}

Note that the node set $V = \{v_{i_1}, v_{i_2}, \ldots, v_{i_n}\}$ corresponds to a subset of the phrases $P_i\subseteq D, P_i = \{p_{i_1}, p_{i_2}, \ldots, p_{i_n}\}$, namely, only those that are header phrases, and thus indices $i_1, i_2, \ldots, i_n$ are not necessarily consecutive.
Definition~\ref{def:sht} additionally implies that the indices of $v_i$ and its descendants are smaller than those of any right sibling of $v_i$ (if any). For example, ``Section~2'' and all its subsection headers should appear before ``Section~3''.
Further note that $V$ is fixed (comprising all header phrases $p_{i_1}, \ldots, p_{i_n}$), while $E$ is not, admitting multiple valid SHTs for the same document. We define the {\em true} SHT as the one that human users believe best reflects the document's hierarchical structure.

\begin{example}
Figure~\ref{fig:true_sht} shows the (partial) true SHT for Walmart Inc. form 10-K.
$v_{16}$ corresponds to $p_{16}$;
its child $v_{20}$ corresponds to $p_{20}$,
a subsection header of
$p_{16}$.
The pre-order traversal
preserves a strictly increasing sequence of node indices,
$v_1, \ldots, v_{10}, \ldots, v_{22}$.
\end{example}

\topic{Text Span}
The {\em text span} of a node $v_i$, denoted by $ts(v_i)$, refers to the text portions (i.e., phrases) associated with $v_i$ in $D$. Given an SHT $T$, $ts(v_i)$ denotes the {\em sequential} list of phrases starting from the header phrase corresponding to $v_i$ and ending immediately before the header phrase corresponding to the next non-descendant node $v_j$ in $T$, defined below, i.e.,  $[p_i, p_{i+1}, \ldots, p_{j-1}]$\hide{}\keepc{.} Formally, node $v_j$ is a non-descendant of $v_i$ in $T$ with the smallest index $j > i$. For example, the text span of $v_{20}$ in \autoref{fig:true_sht} is $[p_{20}, p_{21}]$, since its next non-descendant node is $v_{22}$.

\subsection{Robust and Compact SHTs}
\label{subsec:robust-sht}

Constructing the true SHT may be unrealistic due to semantic ambiguities in documents. A key insight of this paper is that {\em it is often sufficient to construct a robust SHT} for downstream tasks, rather than the true one, as defined below.

\begin{definition} [Robust SHT.]~\label{def:robust-sht}
    An SHT $T$ of a document $D$ is {\em robust} if $\forall~v_i\in V$, $ts'(v_i)\subseteq ts(v_i)$, where $ts'$ and $ts$ denote the text spans inferred from $T'$ (the true SHT) and $T$, respectively.
\end{definition}

A robust SHT $T$ guarantees that the text span $ts(v_i)$ of any node $v_i$ in $T$ is a {\em superset} of that of $v_i$ in the true SHT $T'$.
This robustness property provides correctness guarantees for \hide{}many tasks.
For example, when performing retrieval (e.g., as part of agentic document processing) over a robust SHT $T$, the retrieved context for $v_i$ will always include the true context for $v_i$ in $T'$. This property ensures correctness of question answering \hide{}as long as the retrieval method identifies the right header (e.g., section\hide{} header) where the answer lies, and the LLM gives the right answer when provided a superset of the text portion under this header (i.e., its text span).
\blank{} %

\begin{example}
    Figure~\ref{fig:robust_sht}
    presents
    a robust tree $T$ related to the true one $T'$ in Figure~\ref{fig:true_sht}, where $v_{22}$ is now a child of $v_{20}$.
   $ts(v_{20})$ contains $ts'(v_{20})$ and $ts'(v_{22})$,
    forming a superset of the true context $ts'(v_{20})$ that includes the answer to the question in Figure~\ref{fig:10k}.
\end{example}

A robust SHT also guarantees that any parent-child pairs in the true SHT remain as ancestor-descendant pairs,
a property we will use for SHT inference in the following sections.
\keep{The formal statement is given below and proved in Appendix~\ref{append:proof-ts-correct}:}\hidec{}
\begin{theorem} \label{thm:ts-correct}
    Consider an SHT $T$ of a document $D$. $T$ is robust
    {\em if and only if } for every pair of nodes $v_i, v_j$ where $v_j$ is a child of $v_i$ in the true SHT $T'$, $v_j$ is a descendant of $v_i$ in $T$.
\end{theorem}
\begin{myproof}
$(\Leftarrow)$
Let $T'$ be the true SHT. Any parent-child pair in $T'$ remains an ancestor-descendant pair in SHT $T$.
Suppose $v_{k}$ and $v_{k'}$ are the next non-descendant of $v_i$ in $T$ and $T'$, respectively.
For $v_{k'-1}$, either $v_{k'-1} = v_i$, or it is the last descendant of $v_i$ in $T'$.
For the former case, $k'-1 = i < k$.
For the latter case, since $v_{i}$ is \hide{}\keepc{an ancestor} of $v_{k'-1}$ in $T'$, it is also an ancestor of $v_{k'-1}$ in $T$. Therefore, $v_{k'-1}$ is a descendant of $v_i$ in $T$, indicating $k'-1 < k$.
Thus $ts'(v_i) = [p_i, \ldots, p_{k'-1}] \subseteq [p_i, \ldots, p_{k-1}] = ts(v_i)$. $T$ is robust.

$(\Rightarrow)$
For each parent-child pair $(v_i, v_j)$ in the true SHT $T'$, $p_j\in ts'(v_i)$. Since $T$ is robust, we have $ts'(v_i)\subseteq ts(v_i)$. Therefore, $p_j\in ts(v_i)$, indicating that $v_i$ is an ancestor of $v_j$ in $T$.
\end{myproof}

While we want $T$ to be robust, at the same time, we must optimize $T$ for efficiency by minimizing the text span of each node.
For example, \hide{}\keepc{when retrieving using $T$}, the retrieved context for node $v_i$ is its text span $ts(v_i)$ in $T$.
While $ts(v_i)$ must include the true text span $ts'(v_i)$, it should also contain as little excess content $ts(v_i)\backslash ts'(v_i)$ as possible.
Therefore, to minimize the text span of each node of $T$ while preserving robustness, we want to minimize $\mathsf{C} = \sum_{v_i\in V} ts(v_i)\backslash ts'(v_i)$; we refer to $\mathsf{C}$ as {\em compactness}.

We can now formally define our problem.

\begin{definition} [Robust and Compact SHT Inference.]~\label{def:problem}
    Given a document $D = \{p_1, \ldots, p_n\}$, where each phrase $p_i$ has a label $l_i$ and visual pattern $vp(p_i)$, infer an SHT $T$ from $D$ such that $T$ is robust,
    while minimizing compactness $\mathsf{C}=\sum_{v_i\in V} ts(v_i) \backslash ts'(v_i)$.
\end{definition}
\vspace*{-\tinyskipamount}

\begin{example}
    For the SHTs
    of
    the Walmart Inc. form 10-K,
    we prefer
    \autoref{fig:robust_sht}
    over the deep and wide SHTs in \autoref{fig:deep_and_wide_shts}.
    The wide SHT is not robust:
    $ts(v_{16})$ incorrectly excludes
    $ts'(v_{20})$.
    The deep SHT is not compact:
    $ts(v_{20})$ unnecessarily includes all content from $p_{20}$ to the end of
    $D$.
    In contrast,
    \autoref{fig:robust_sht} achieves robustness, while also being more compact.
    For robustness, $ts(v_{16})$ is a superset of its true text span, including
    $ts'(v_{20})$.
    For compactness, $ts(v_{20})$ includes
    only
    the small text span of $v_{22}$ in addition to
    $ts'(v_{20})$.
\end{example}
\vspace*{-\tinyskipamount}

\section{SHT Inference: Two Approaches}
\label{sec:sht-construction}

In this section, we present
our two-stage workflow, \sys,
for inferring a robust and compact SHT $T=(V,E)$ from a document $D$,
starting with an overview in \Cref{subsec:workflow}.
As mentioned earlier, the set of nodes $V$ corresponds to the set of header phrases. To assemble $V$ into an SHT $T$, we rely on the intuition that the semantic hierarchy represented by an SHT depends on the {\em visual patterns} of nodes $vp(\cdot)$.
Recall that these visual patterns are represented by feature vectors that encode various visual features (e.g., \code{font\_size}, \code{font\_name}) of the phrases corresponding to the nodes.

\subsection{Overall Workflow}
\label{subsec:workflow}

The overall workflow of SHT inference consists of two stages:

\noindent{\bf 1) Semantic depth (SD) inference}
(\Cref{subsec:infer-sd-local} and \Cref{subsec:infer-sd-global}).
The semantic relationships between nodes (i.e., the edges in the SHT) closely correspond to their visual patterns: certain visual patterns (e.g., those of section headers) tend to indicate parent nodes, while others (e.g., those of subsection headers) indicate children.
Therefore, we first cluster the nodes based on their visual patterns---so each visual pattern corresponds to a cluster---and then infer what we refer to as the {\em semantic depth} of each cluster.
Semantic depth does not directly specify the exact ``level'' of nodes with the given visual pattern in the SHT;
instead, it constrains semantic relationships among nodes with different visual patterns (i.e., edges in the SHT), as described below.

\noindent{\bf 2) SHT assembly}
(\Cref{subsec:sht-assembly}).
We then build the SHT $T$
by inferring edges based on semantic depths of the clusters.
Specifically, we use the {\em relative ordering} of semantic depths to determine the allowed semantic relationships: the parent of a node in the SHT must belong to a cluster with a {\em smaller} semantic depth.

\vspace{2pt}
\noindent As shown in \autoref{fig:workflow},
there are multiple approaches for SD inference $\code{Infer-SD}_1, \code{Infer-SD}_2, \ldots$;
the approaches in Sections~\ref{subsec:infer-sd-local} and~\ref{subsec:infer-sd-global} correspond to $\code{Infer-SD}_1$, which we refer to as local-first, and  $\code{Infer-SD}_2$, which we refer to as global-first, respectively.
Both local- and global-first approaches rely on the node with the smallest index in each cluster to determine the semantic depth of each cluster.
This node, which we call the {\em indexing node} of the cluster, corresponds to the first occurrence of the visual pattern associated with the cluster in the document.
The local-first approach infers the semantic depth of each cluster using {\em local} structural cues, examining only the node whose header phrase immediately precedes that of the cluster's indexing node.
In contrast, global-first uses {\em global} structural cues by considering all nodes whose header phrases precede that of the indexing node.
Within our two-stage workflow, both approaches produce \hide{}{correct and compact SHTs},
as we show in \Cref{subsec:correct-compact}.

The robustness of the SHTs inferred by local- and global-first approaches will be analyzed in Sections~\ref{subsec:sd-robustnes} and~\ref{subsec:comp-infer-sd}.
We further combine local- and global-first into an infinite family of approaches
in \Cref{subsec:family-infer-sd}.
Each approach in this family also ensures SHT robustness for a particular class of documents, as shown in \Cref{subsec:hierarchy-robustness}.
The global-first approach generalizes ZenDB's inference algorithm~\cite{zendb}, which was restricted to {\em well-formatted} documents. We show in \Cref{subsec:hierarchy-robustness} that well-formatted documents constitute only a small subset of real-world documents, and the global-first approach, as a generalization within our two-stage framework, infers robust SHTs for a substantially broader class of documents.

We begin by discussing our first SD inference approach.

\begin{figure}[t]
    \centering
    \includegraphics[width=1\linewidth]{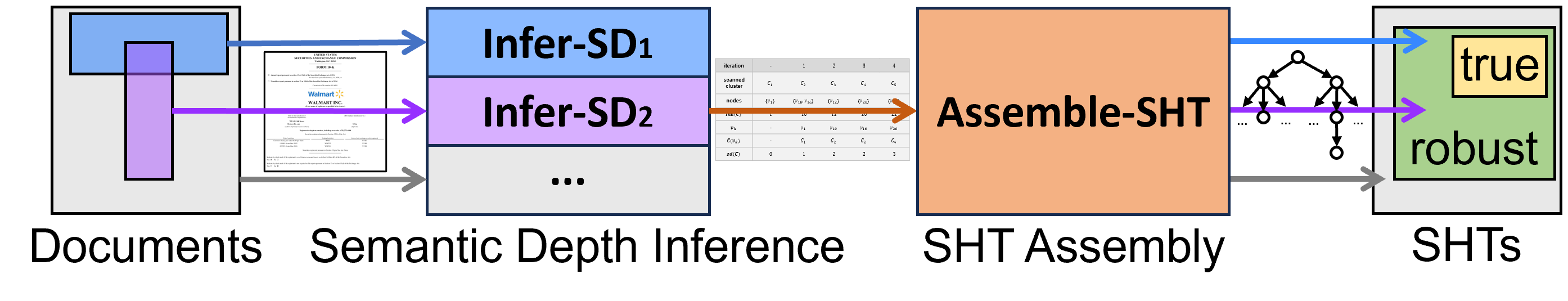}
    \caption{The overall workflow of \sys for SHT Inference.}
    \label{fig:workflow}
\end{figure}

\subsection{Semantic Depth Inference: Local-First}
\label{subsec:infer-sd-local}
In local-first SD inference,
we start by clustering the header nodes $V$, based on their visual patterns, into clusters $\mathcal{C} = \{C_1, \ldots, C_m\}$. Two nodes $v_i, v_j \in V$ belong to the same cluster if they share identical visual patterns, i.e., $vp(p_i) = vp(p_j)$.
For each cluster $C_j$, we will infer its semantic depth, denoted $sd(C_j)$. We assume that cluster $C_1$ contains only the root node $v_1$ and set its $sd(C_1) = 0$.

Next, we describe the inference of $sd(C)$ for each subsequent cluster.
We leverage the local-first property, i.e., inferring semantic depth using the node that immediately precedes the indexing node of each cluster, as detailed below.
We infer $sd(C_j)$ from the node $v_i$ whose header phrase $p_i$ is the {\em first} instance of the visual pattern corresponding to $C_j$ encountered in the document. Formally,
let $i=ind(C_j)$ denote the {\em index of a cluster}, defined as the smallest index among the indices of the nodes belonging to $C_j$.
We call $v_{ind(C)}$ the {\em indexing node} of $C$.
For example, if a cluster $C_2 = \{v_{10}, v_{16}\}$, its index is $ind(C_2) = 10$, and $v_{10}$ is its indexing node.
We use $C(v)$ to denote the cluster that contains node $v$.

\begin{figure}[t]
\centering
    \begin{subfigure}[b]{0.45\linewidth}
         \centering
         \includegraphics[width=1\textwidth]{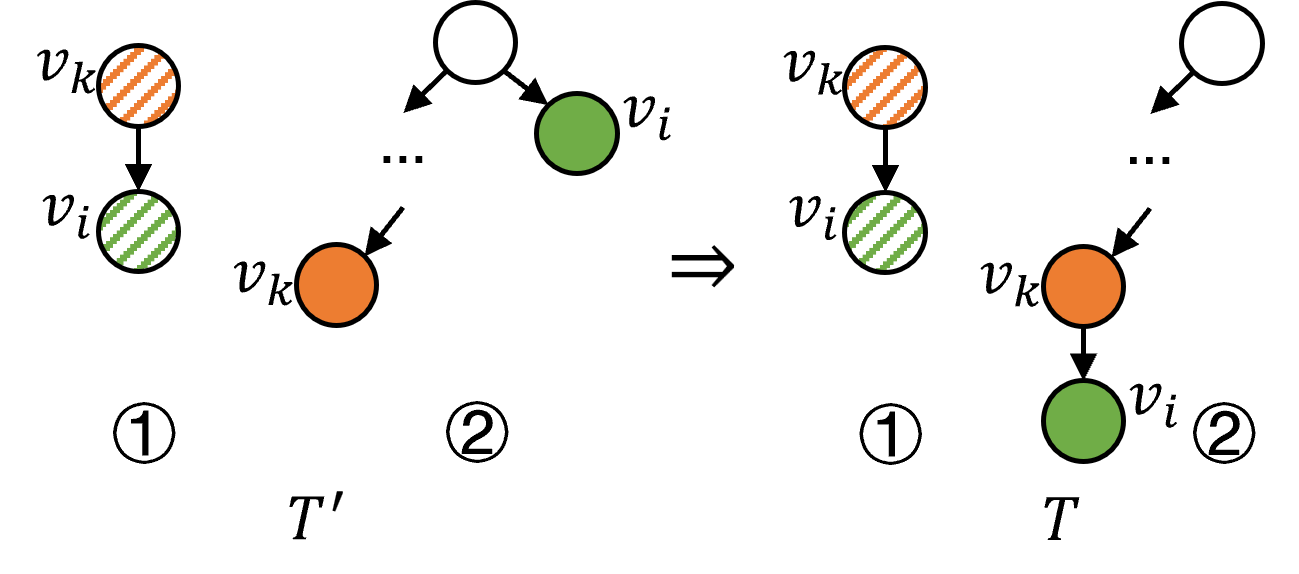}
         \vspace{-2em}
         \caption{Local-first.}
         \label{fig:intuition_infer_sd_1}
    \end{subfigure}
    \begin{subfigure}[b]{0.52\linewidth}
         \centering
         \includegraphics[width=1\textwidth]{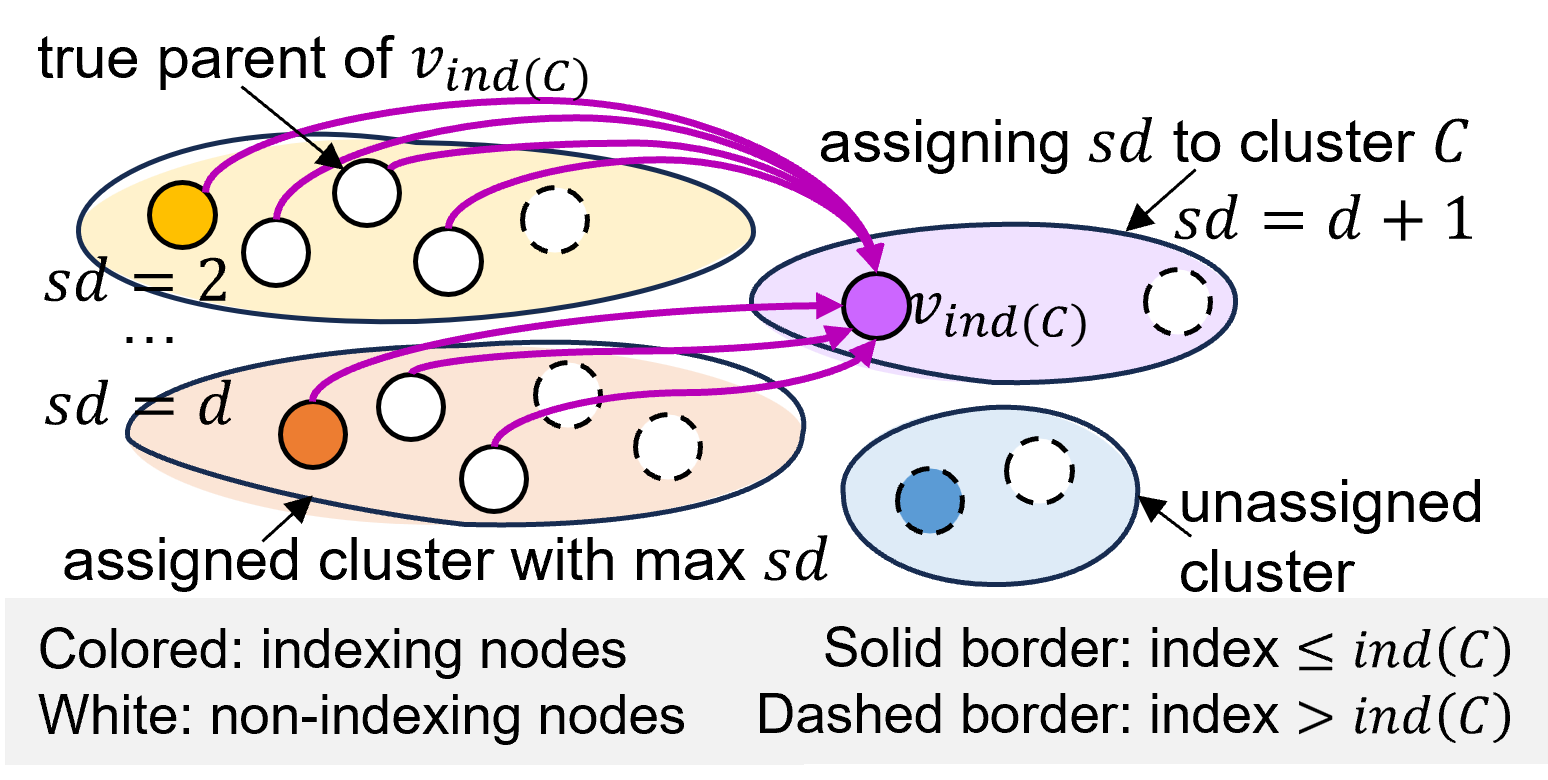}
         \vspace{-2em}
         \caption{Global-first.}
         \label{fig:intuition_infer_sd_2}
    \end{subfigure}
    \caption{Intuition of SD Inference.}
\end{figure}

\begin{figure*}[t]
\centering
    \begin{subfigure}[b]{0.1696\linewidth}
         \centering
         \includegraphics[width=1\textwidth]{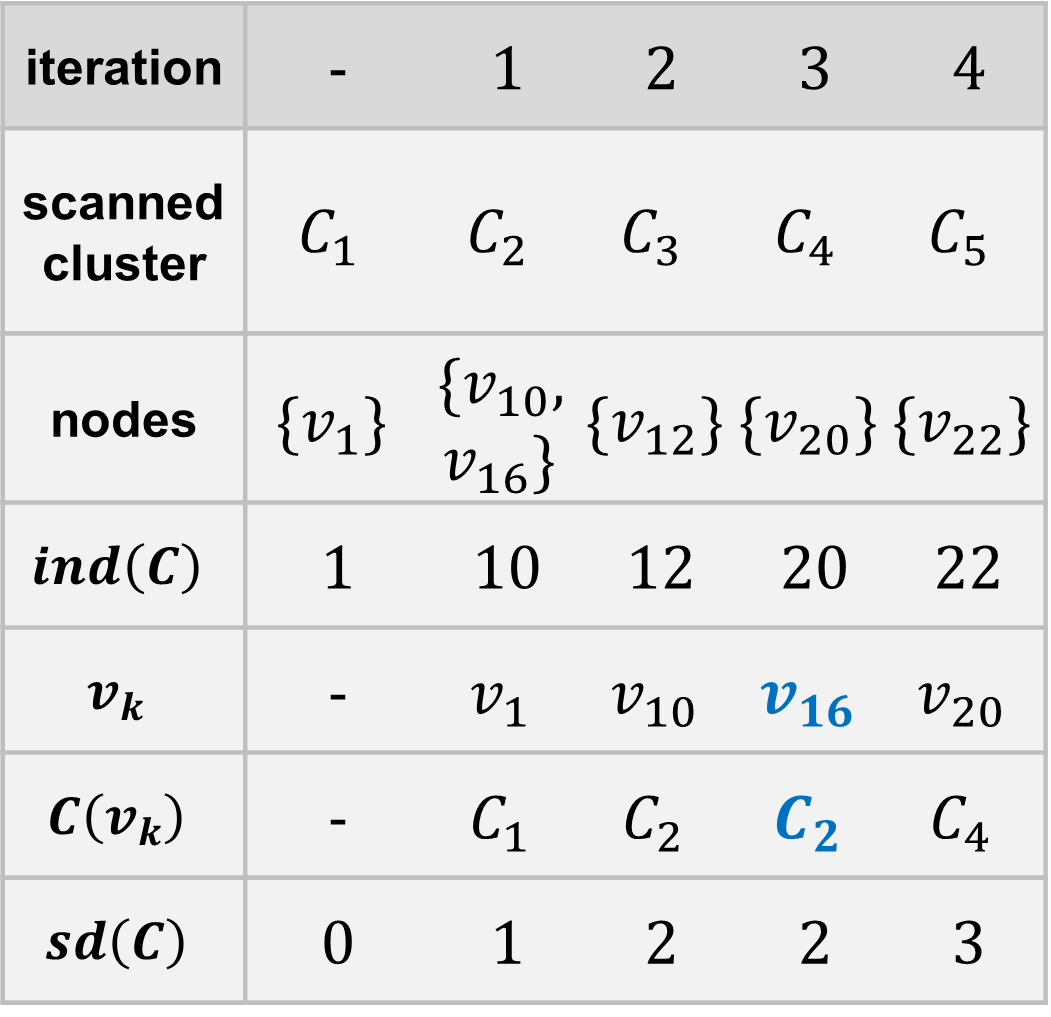}
         \vspace{-1.5em}
         \caption{Algorithm~\ref{alg:build_cg_1}}
         \label{fig:infer_sd_local}
    \end{subfigure}
    \begin{subfigure}[b]{0.1304\linewidth}
         \centering
         \includegraphics[width=1\textwidth]{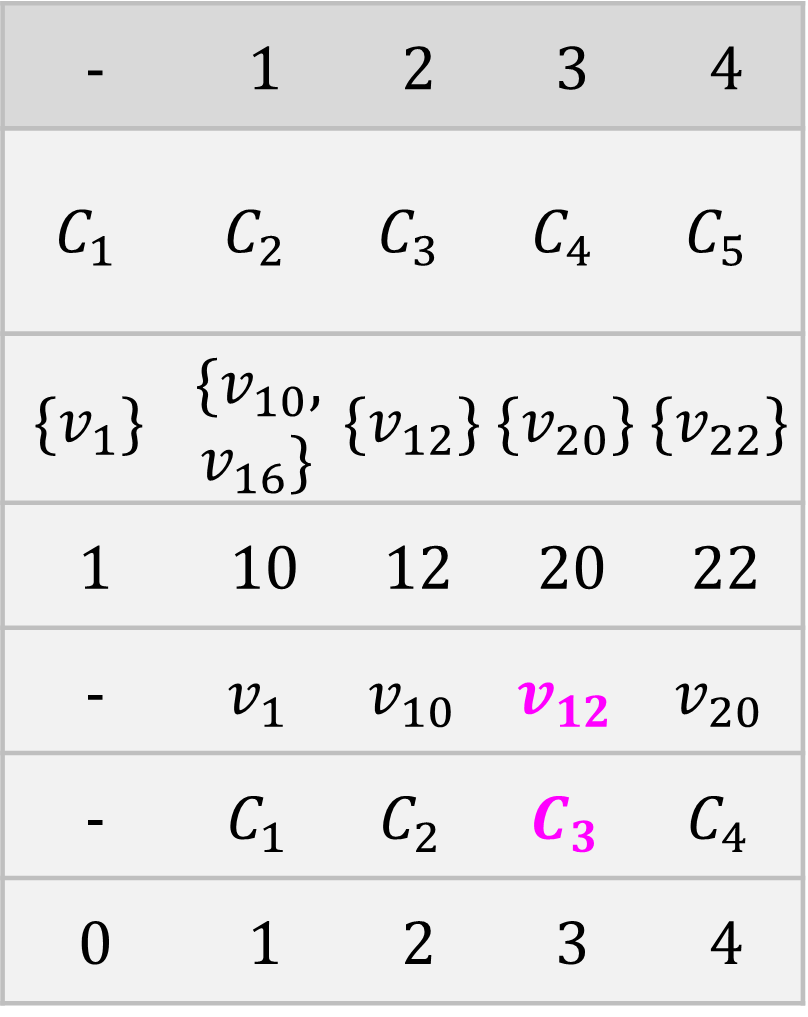}
         \vspace{-1.5em}
         \caption{Algorithm~\ref{alg:approxcg-2}}
         \label{fig:infer_sd_global}
    \end{subfigure}
    \begin{subfigure}{0.67\linewidth} %
         \centering
         \includegraphics[width=1\textwidth]{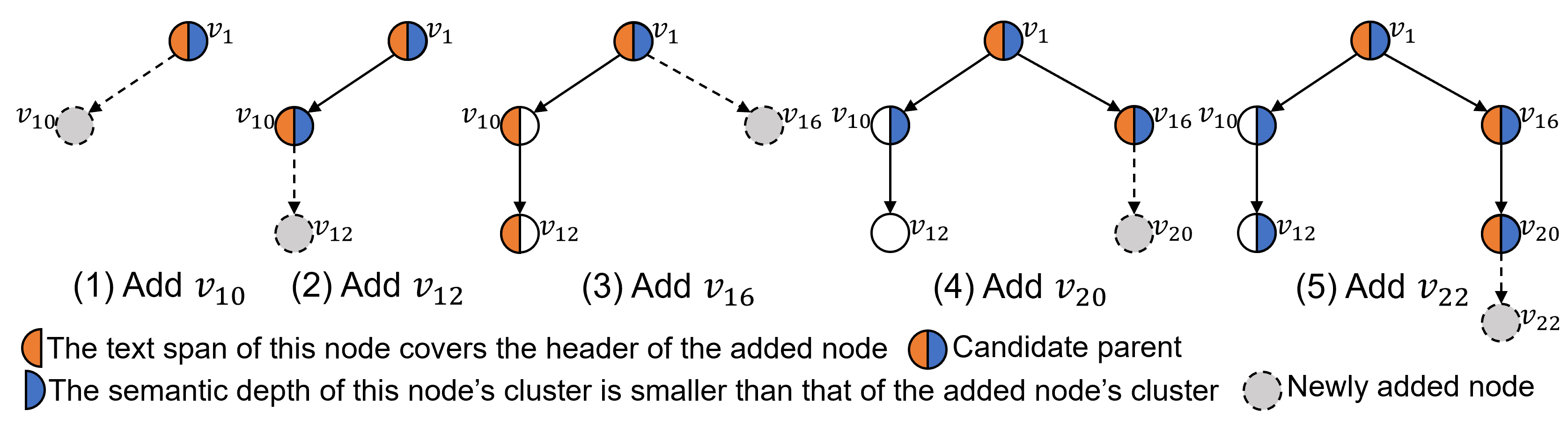}
         \vspace{-1.5em}
         \caption{Algorithm~\ref{alg:shtgen}}
         \label{fig:shtgen}
    \end{subfigure}
    \caption{An illustration of SHT inference.}
    \label{fig:sht_infer_example}
\end{figure*}

For the indexing node $v_i$ of a cluster $C_j$, suppose $v_k$ is the node with the largest index $k < i$. In other words, $v_k$ immediately precedes $v_i$ in the pre-order traversal of the true SHT $T'$, indicating one of the following two cases, as shown in \autoref{fig:intuition_infer_sd_1}: \circled{1} $v_k$ is the true parent of $v_i$, or \circled{2} $v_k$ is a descendant of the true parent of $v_i$ in $T'$---but we are unaware of which case we are in.
Recall that, to ensure robustness, we require the true parent of $v_i$ to remain an ancestor in the inferred SHT, as stated in Theorem~\ref{thm:ts-correct}.
\hide{}If we move $v_i$ to be a child of $v_k$, the true parent of $v_i$ in $T'$ remains an ancestor of $v_i$ in $T$, ensuring robustness.
Based on this intuition, we assign $C(v_i)$ a larger semantic depth than $C(v_k)$.
Note, however, that the latter case \circled{2} may lead to reduced compactness, but we accept this compactness cost to ensure robustness.

\begin{algorithm}[bt]
    \caption{$\code{Infer-SD}_1(V, \mathcal{C})$ // Local-First}
    \begin{algorithmic}[1]
        \State $sd(C_1)\leftarrow 0$
        \For{$C\in \mathcal{C}\backslash \{C_1\}$ in ascending order of $ind(C)$}
            \State $v_k\leftarrow \arg\max_{v_i\in V, i < ind(C)}: i$
            \State $sd(C)\leftarrow sd(C(v_k)) + 1$
        \EndFor
        \State \Return $sd$
        \Comment{a mapping from a cluster to its semantic depth}
    \end{algorithmic}
    \label{alg:build_cg_1}
\end{algorithm}

Algorithm~\ref{alg:build_cg_1} shows local-first SD inference.
Starting with $sd(C_1) = 0$ (Line 1), we scan the clusters in increasing order of their indices (Line 2-4).
For
cluster $C$, we find
$v_k$ with the largest index $k < {ind(C)}$ (Line 3).
We then set $sd(C)$ to be
$sd(C(v_k)) + 1$ to ensure robustness of the SHT inferred subsequently.
The scanning order of Algorithm~\ref{alg:build_cg_1} ensures that $sd(C(v_k))$ has already been inferred when computing $sd(C)$, since $ind(C(v_k)) \leq k < ind(C)$ \hidec{}
(Line 2).

\begin{example}
    \autoref{fig:infer_sd_local} illustrates Algorithm~\ref{alg:build_cg_1} using only the nodes explicitly shown in \autoref{fig:robust_sht}, with all other nodes hidden for simplicity.
    For example, to infer $sd$ for $C_4$,
    we identify the node
    $v_k = v_{16}$
    immediately prior to its indexing node $v_{20}$
    and assign $sd(C_4) = sd(C(v_{16})) +1 = sd(C_2) + 1 = 2$.
\label{ex:infer_h}
\end{example}
\vspace*{-\smallskipamount}

\subsection{Semantic Depth Inference: Global-First}
\label{subsec:infer-sd-global}

Next, we provide an alternate global-first approach for SD inference\footnote{As mentioned earlier, Algorithm~\ref{alg:approxcg-2} generalizes ZenDB's algorithm to work on non-well-formatted documents, and leverages the two-stage approach\keep{; see details in~Appendix~\ref{append:relation_zendb}.}\hidec{}
} (Algorithm~\ref{alg:approxcg-2}).
Like Algorithm~\ref{alg:build_cg_1}, Algorithm~\ref{alg:approxcg-2} assigns semantic depths to clusters based on their indices: clusters with smaller indices (i.e., visual patterns whose first occurrence appears earlier in the document) have smaller semantic depths.
Formally, after sorting the clusters
in
increasing order of their indices (Line 2), Algorithm~\ref{alg:approxcg-2} assigns each cluster $C$ a semantic depth that is one greater than the largest semantic depth
among all preceding clusters
(Line 3-4).

\begin{algorithm}[bt]
    \caption{$\code{Infer-SD}_{2}(V, \mathcal{C})$ // Global-First}
    \begin{algorithmic}[1]
    \State $sd(C_1)\leftarrow 0$
    \For{$C\in \mathcal{C}\backslash{\{C_1\}}$ in ascending order of $ind(C)$}
        \State $v_k\leftarrow \arg\max_{v_i\in V, i < ind(C)}: sd(C(v_i))$
        \State $sd(C) \leftarrow sd(C(v_k)) + 1$
    \EndFor
    \State \Return $sd$
    \end{algorithmic}
    \label{alg:approxcg-2}
\end{algorithm}

In effect, the semantic depth of each cluster $C$ is determined {\em globally} via its indexing node $v_{ind(C)}$:
Algorithm~\ref{alg:approxcg-2} considers {\em all} nodes with smaller indices than $v_{ind(C)}$.
We do not know which of these nodes is the true parent of $v_{ind(C)}$ (shown as purple arrows in \autoref{fig:intuition_infer_sd_2}), so we opt for the worst case and set $sd(C)$ to be one larger than the largest semantic depth among their clusters, ensuring that $sd(C)$ is greater than the semantic depth of the cluster of the true parent of $v_{ind(C)}$, as illustrated in \autoref{fig:intuition_infer_sd_2}.

\begin{example}
Using the same nodes and clusters as in \autoref{fig:infer_sd_local},
Algorithm~\ref{alg:approxcg-2} infers
$\small sd(C_i) = i - 1$ for each $\small i\in [1, 5]$,
with $\small sd(C_i) < sd(C_j)$ whenever $\small i < j$,
as shown in \autoref{fig:infer_sd_global}.
\end{example}
\vspace*{-\smallskipamount}

\subsection{SHT Assembly}
\label{subsec:sht-assembly}

Given nodes $V$ and inferred semantic depths $sd$ of their clusters $\mathcal{C}$ (using local- or global-first),
we now assemble $T$ by inferring the parent of each non-root node in $V$.

We infer the parents of nodes in increasing order of node indices, just as humans make sense of the structure of a document by reading the document sequentially.
Suppose we have a partial SHT $T$ containing all nodes with indices smaller than $i$, and we now infer the parent $v_k$ of node $v_i$. Let $v_j$ denote the true parent of $v_i$ (which is unknown).
Using the information available (i.e., $V, \mathcal{C}$, and $sd$), we can derive conditions that the true parent $v_j$ must satisfy, so that selecting $v_k$ ensures robustness while optimizing compactness $\mathsf{C} = \sum_{v\in V}ts(v)\backslash ts'(v)$. For simplicity, we assume that $T$ is the (partial) true SHT, i.e., $T$ matches the true SHT so far.
We define {\em candidate parents} of $v_i$ as the nodes that satisfy conditions {\bf 1)} and {\bf 2)} as stated below; this set must include the true parent $v_j$. We then select $v_k$ from this set to ensure robustness and optimize compactness, without knowing $v_j$.

\noindent{\bf 1) Coverage of Text Span:}
$p_i\in ts(v_j)$; i.e.,
the text span of $v_j$ inferred from $T$ must contain the header phrase $p_i$ of $v_i$ (e.g., a subsection header is contained within its parent section).
Since all nodes in $T$ have indices smaller than $i$, we have $j < i$, and $p_i\in ts(v_j)$ further implies that $v_j$ lies on the {\em rightmost path} of $T$, as described below, with its text span extending to the end of the document: $ts(v_j) = [p_j, \ldots, p_n]$, where $j < i < n$ and $p_n$ denotes the final phrase of the document.
The rightmost path of $T$ starts from the root and follows the rightmost child at each node down to a leaf. In \autoref{fig:shtgen}(5), $[v_1, v_{16}, v_{20}, v_{22}]$ is the rightmost path.

From {\bf 1)} alone, the candidate parents for $v_i$ are the nodes on the rightmost path of $T$.
We can therefore select $v_k$ as the candidate with the largest index (i.e., the rightmost leaf of $T$), inserting $v_i$ at the end of the rightmost path, guaranteeing that the true parent $v_j$ remains an ancestor of $v_i$, and ensuring robustness.

However, this choice of $v_k$ yields worst-case compactness: $ts(v_i)$ is added to the text span of every node on the rightmost path, contributing $r$ times to compactness $\mathsf{C}$, where $r$ is the number of candidate parents. To reduce $r$ and improve compactness, we introduce {\bf 2)} to further narrow the set of candidate parents.

\topic{2) Semantic Depth Constraint: $sd(C(v_j)) < sd(C(v_i))$}
Recall that semantic depths of clusters constrain semantic relationships between nodes: the semantic depth of a parent's cluster must be smaller than that of the child's cluster, i.e., $sd(C(v_j)) < sd(C(v_i))$.

Condition {\bf 2)} uses semantic depths $sd$ to further restrict the set of candidate parents to nodes on the rightmost path of $T$ whose clusters have smaller semantic depths than $C(v_i)$.
This set still contains the true parent $v_j$ while reducing the number of candidate parents $r$.
Selecting $v_k$ as the candidate parent with the largest index preserves robustness of $T$,
since both $v_k$ and $v_j$ lie on the rightmost path with $k \ge j$, implying that $v_k$ is equal to or a descendant of $v_j$.
This choice of $v_k$ also improves compactness of $T$, as the reduced set of candidate parents leads to a smaller $r$ than using {\bf 1)} alone.

\begin{algorithm}[bt]
    \caption{\code{Assemble-SHT}($V, \mathcal{C}, sd$)}
    \begin{algorithmic}[1]
        \State $T\leftarrow $ an empty tree
        \For{$v_i\in V$ in ascending order of their indices}
            \State $v_k \leftarrow \arg\max_{v_j\in T, p_i\in ts(v_j), sd(C(v_j)) < sd(C(v_i))}: j$
            \State Add node $v_i$ to $T$ as the rightmost child of $v_k$
        \EndFor
        \State \Return $T$
    \end{algorithmic}
    \label{alg:shtgen}
\end{algorithm}

Using the above two conditions, we can now present the algorithm to assemble an SHT $T$
in Algorithm~\ref{alg:shtgen}.
We consider the nodes in increasing order of their indices (e.g., from $v_1, v_{10}$ to $v_{22}$) (Line 2).
In each iteration, we insert $v_i$ into $T$ as the {\em rightmost child} of the selected parent $v_k$ (Line 4), so that the pre-order traversal of $T$ remains consistent with the nodes' indexing order and thus with the document's reading order.
As mentioned earlier, $v_k$ is selected as the candidate parent of $v_i$ with the largest index,
from among all $v_j$ in $T$ that satisfy:
{\bf 1)} $p_i\in ts(v_j)$, and {\bf 2)} $sd(C(v_j)) < sd(C(v_i))$.
Note that the text span $ts(v_j)$ of each node $v_j$ evolves as new nodes are inserted into $T$.
Therefore, in Line 3, $ts(v_j)$ denotes the text span of $v_j$ inferred from the partial tree $T$ at that point in the algorithm.

\begin{example}
    \autoref{fig:shtgen} illustrates Algorithm~\ref{alg:shtgen} using the nodes, clusters, and semantic depths from Figure~\ref{fig:infer_sd_local}.
    When inserting $v_{20}$ (step (4)),
    the candidate parents---nodes in $T$ with $\small sd<sd(C(v_{20})) = 2$ (blue)
    and text spans covering $p_{20}$ (orange)---are $v_1$ and $v_{16}$.
    Since $v_{16}$ has the larger index, $v_{20}$ is inserted as its rightmost child.
\end{example}

\subsection{Correctness and Compactness}
\label{subsec:correct-compact}

We now establish that our two-stage workflow correctly generates an SHT.
We also show that the resulting SHT has good compactness.
Correctness and compactness hold for any document and are independent of the SD inference approach.
With respect to robustness, different SD inference approaches apply to different classes of documents. We analyze these document classes in \Cref{sec:robustness-analysis}.

\subsubsection{SHT Correctness}
We can formally show
that the tree $T$ inferred by our two-stage workflow is an SHT:

\begin{theorem}\label{lem:sht}
    The tree $T$ output by Algorithm~\ref{alg:shtgen} is an SHT.
\end{theorem}
\begin{myproof}
    When inserting $v_i$, all nodes whose text spans cover the header phrase of $v_i$ lie on the {\em rightmost path} of $T$ (e.g., in \autoref{fig:shtgen}, all nodes highlighted in orange are on the rightmost path at each insertion). $v_i$ is then inserted as the rightmost child of a node on this rightmost path, and therefore appears {\em last} in the pre-order traversal of $T$. By induction, after every iteration of Algorithm~\ref{alg:shtgen}, the pre-order traversal of $T$ remains consistent with the increasing order of the node indices.
\end{myproof}

\keep{Proof is  in~Appendix~\ref{append:proof_lem_sht}.
In brief, each node, when added, preserves the reading order thanks to condition {\bf 1)} in \Cref{subsec:sht-assembly}.}

\subsubsection{SHT Compactness}
We now show that SHTs inferred by our workflow have good compactness (irrespective of algorithm).

While true compactness $\mathsf{C} = \sum_{v\in V}ts(v)\backslash ts'(v)$ is hard to quantify without the true SHT $T'$, a suitable proxy is the depth of the inferred SHT $T$, defined as the number of nodes on the longest root-to-leaf path.
Deeper SHTs tend to be less compact\keep{: as \sys's SHT depth ratio to the true SHT {(i.e., depth($T$)/depth($T'$))} grows from ${\leq}1$ to $2.2$, average compactness drops from 0.92 to 0.83 (Appendix~\ref{append:correlation-compact-depth}). {We therefore use bounded depth as a proxy for good compactness, without claiming optimality with respect to compactness.}}
\hidec{}

Regardless of the SD inference approach used, the depth of an SHT $T$ inferred by Algorithm~\ref{alg:shtgen} is bounded by the number of clusters (i.e., the number of distinct visual patterns among headers), since the path from any node to the root consists of strictly decreasing semantic depths, terminating at zero\hidec{}.
In contrast, a deep SHT $T_d$ that links all nodes into a single chain  (\autoref{fig:deep_and_wide_shts}-1) has depth equal to the number of tree nodes\hide{}.
Since the number of clusters is usually much smaller than the number of tree nodes, $T$ is more compact than $T_d$, as shown on 600 real-world documents (Appendix~\ref{append:datasets}).  \hide{}

\section{Robustness Analysis and Hierarchy}
\label{sec:robustness-analysis}

In this section, we analyze the robustness of SHTs from different SD inference approaches.
We first identify document classes that ensure robustness for arbitrary semantic depths in \Cref{subsec:sd-robustnes}. These classes satisfy a property that we call {\em semantic depth conformance}, which depends on the inferred semantic depths; different semantic depths therefore induce different document classes where robustness is met.
Next, in \Cref{subsec:comp-infer-sd}, we compare robustness guarantees of local- (\Cref{subsec:infer-sd-local}) and global-first (\Cref{subsec:infer-sd-global}) SD inference. We also identify a document class for which local-first guarantees robustness while global-first does not.
Then, in \Cref{subsec:family-infer-sd}, we generalize local- and global-first by combining them into an infinite family of SD inference approaches.
In \Cref{subsec:hierarchy-robustness}, we summarize all proposed document classes by showing hierarchical relationships among them via a Venn diagram.
\keep{In \Cref{subsec:discussion_relationships_zendb}, we compare \sys with ZenDB~\cite{zendb}.}
\keep{Overall, this section completes our formal study of robust and compact SHT inference (\Cref{subsec:robust-sht}), complementing the correctness and compactness analyses  (\Cref{subsec:correct-compact}). {We establish guarantees of perfect robustness}, and characterize when they hold so that users can select approaches based on
document characteristics.}

\subsection{Semantic Depth Conformance}
\label{subsec:sd-robustnes}

Consider a document $D$, its true SHT $T'$, and node clusters $\mathcal{C}$.
We now establish a property of a document under which the SHT $T$ returned by Algorithm~\ref{alg:shtgen} is guaranteed to be robust for any given assignment of semantic depths $sd$ to clusters.

To ensure robustness of $T$, Theorem~\ref{thm:ts-correct} requires that every parent-child pair $(v_j, v_i)$ in the true SHT $T'$ remains an ancestor-descendant pair in $T$, which can be achieved by Algorithm~\ref{alg:shtgen} only if $sd(C(v_j)) < sd(C(v_i))$,
a cluster-level relationship encoded by the relative ordering of semantic depths, as enforced by condition {\bf 2)} in \Cref{subsec:sht-assembly}.
In other words, for any true parent-child pair in $T'$, the parent’s cluster must have a smaller semantic depth than the child’s cluster, a document property we refer to as {\em semantic depth conformance} and formalize in Theorem~\ref{theo:correctness-general}, so that Algorithm~\ref{alg:shtgen} permits the parent to remain as an ancestor of the child in $T$.
But, perhaps surprisingly, this property is sufficient to also guarantee robustness.

For convenience, we define $sd(v) = sd(C(v))$ for any node $v\in V$.

\begin{theorem} [SD Conformance Equals SHT Robustness.]
    Consider a document $D$, with nodes $V$ and clusters $\mathcal{C}$.
    Let $T$ be the SHT inferred by Algorithm~\ref{alg:shtgen}, and $T'$ be the true SHT.
    Given semantic depths $sd$ of the clusters inferred by an arbitrary SD inference approach,
    $T$ is robust {\em if and only if} $~\forall~v_i, v_j\in V$, whenever $v_j$ is the parent of $v_i$ in $T'$, it holds that $sd(v_j) < sd(v_i).$
\label{theo:correctness-general}
\end{theorem}
\begin{myproof} [$(\Rightarrow)$.]
    If $T$ is robust, $v_j$ is an ancestor of $v_i$ in $T$. Algorithm~\ref{alg:shtgen} ensures that for any node, the semantic depth of its cluster is larger than that of the cluster containing its parent. Applying this property along the ancestor chain from $v_i$ up to $v_j$ in $T$ gives $sd(v_j) < sd(v_i).$

    ($\Leftarrow$).
    We prove by induction that each time Algorithm~\ref{alg:shtgen} inserts a node $v_i$, it places $v_i$ as a descendant of its true parent $v_j$. Consequently, the tree $T$ remains robust throughout the execution of the algorithm.
    The key step is to show that, at the moment $v_i$ is inserted, {\em its true parent $v_j$ is a candidate parent of $v_i$ in Algorithm~\ref{alg:shtgen}}.
    Since $j < i$ and $sd(v_j) < sd(v_i)$, it suffices to show that $ts(v_j)$ covers the header phrase of $v_i$ at the point when $v_i$ is inserted, which is equivalent to showing that $v_j$ lies on the rightmost path of $T$.
    Indeed, all nodes inserted into $T$ after $v_j$ and before $v_i$ are placed as descendants of $v_j$ in $T$: they are descendants of $v_j$ in the true SHT, and remain descendants of $v_j$ in $T$ by our induction hypothesis that $T$ is robust before inserting $v_i$.
    Therefore, when $v_i$ is inserted, $v_j$ remains on the rightmost path of $T$.
\end{myproof}

\keep{We prove Theorem~\ref{theo:correctness-general} in Appendix~\ref{append:proof_correctness_general}.}
Note that when semantic depths are inferred using global-first, SD conformance reduces to requiring $ind(C(v_j)) < ind(C(v_i))$ for any parent-child pair $(v_j, v_i)$ in the true SHT, because the ordering of semantic depths is consistent with the indices of the corresponding clusters.

\subsection{Local vs. Global-First Depth Inference}
\label{subsec:comp-infer-sd}

To compare local- and global-first approaches,
we first distinguish them using two
concrete
examples in \Cref{subsec:diff-sd-example}:
one where local-first
infers
a robust SHT
while global-first fails,
and vice versa.
Then, in \Cref{subsec:diff-sd-class},
we characterize a class of documents
for which local-first guarantees robustness,
whereas global-first does not.

\subsubsection{Illustrative Scenarios}
\label{subsec:diff-sd-example}
\begin{figure*}[t]
\centering
    \begin{subfigure}[b]{0.14\linewidth}
         \centering
         \includegraphics[width=1\textwidth]{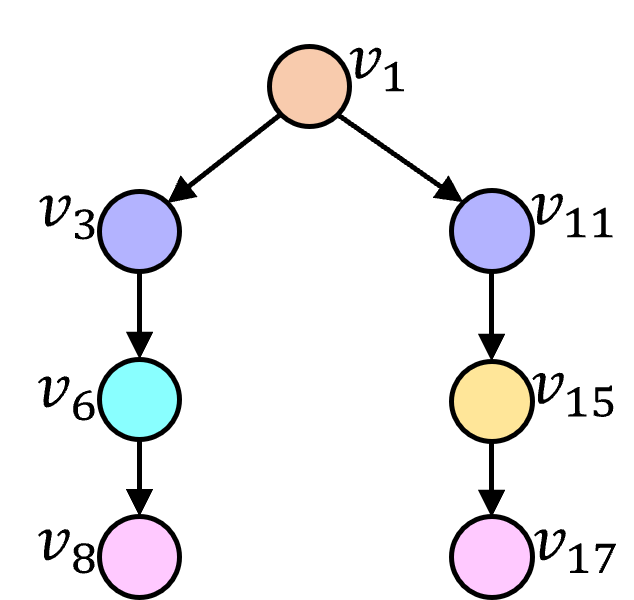}
         \caption{True SHT}
         \label{fig:comp_infer_sd_1_true}
    \end{subfigure}
    \begin{subfigure}[b]{0.14\linewidth}
         \centering
         \includegraphics[width=1\textwidth]{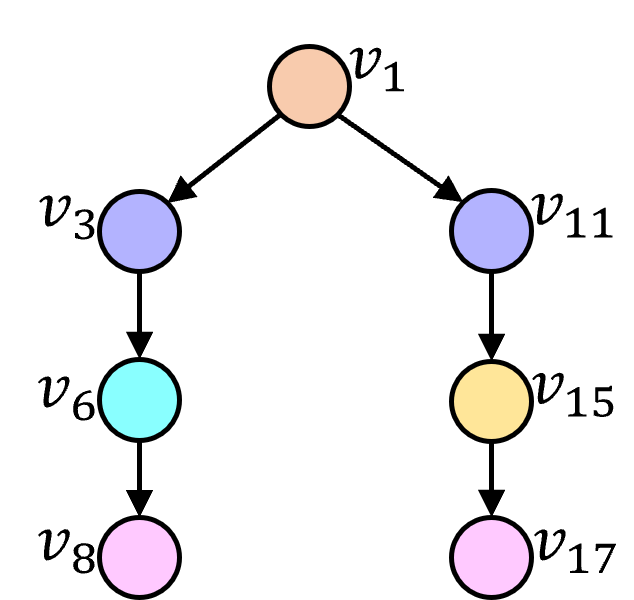}
         \caption{Algorithm~\ref{alg:build_cg_1}}
         \label{fig:comp_infer_sd_1_1}
    \end{subfigure}
    \begin{subfigure}[b]{0.14\linewidth}
         \centering
         \includegraphics[width=1\textwidth]{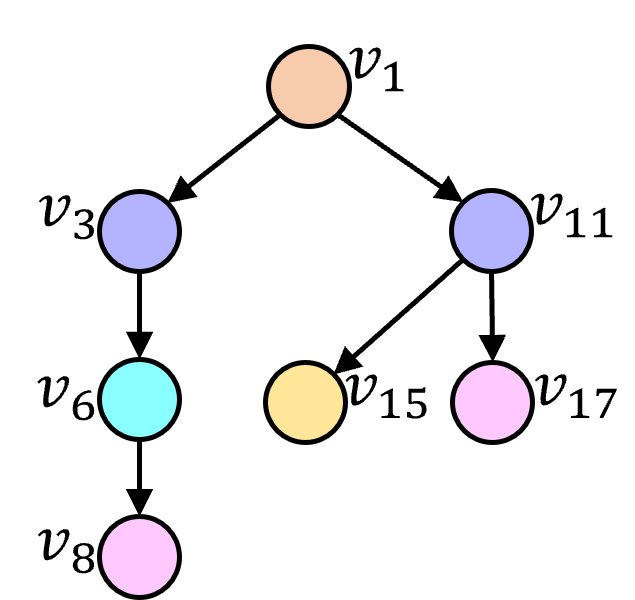}
         \caption{Algorithm~\ref{alg:approxcg-2}}
         \label{fig:comp_infer_sd_1_2}
    \end{subfigure}
    \begin{subfigure}[b]{0.14\linewidth}
         \centering
         \includegraphics[width=1\textwidth]{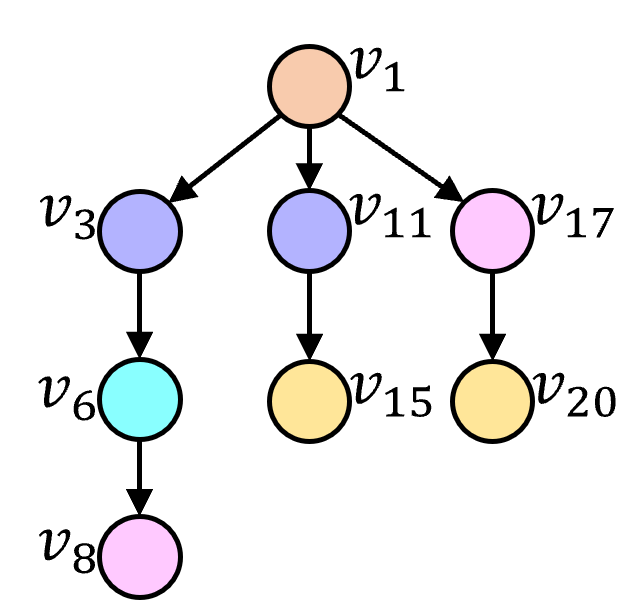}
         \caption{True SHT}
         \label{fig:comp_infer_sd_2_true}
    \end{subfigure}
    \begin{subfigure}[b]{0.14\linewidth}
         \centering
         \includegraphics[width=1\textwidth]{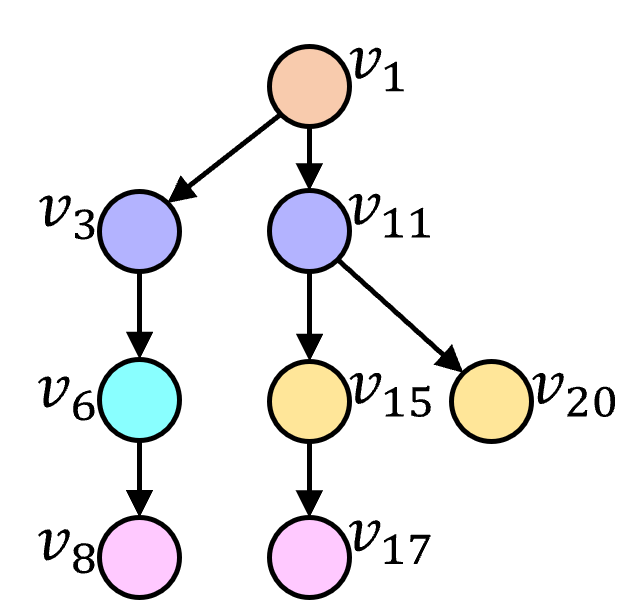}
         \caption{Algorithm~\ref{alg:build_cg_1}}
         \label{fig:comp_infer_sd_2_1}
    \end{subfigure}
    \begin{subfigure}[b]{0.14\linewidth}
         \centering
         \includegraphics[width=1\textwidth]{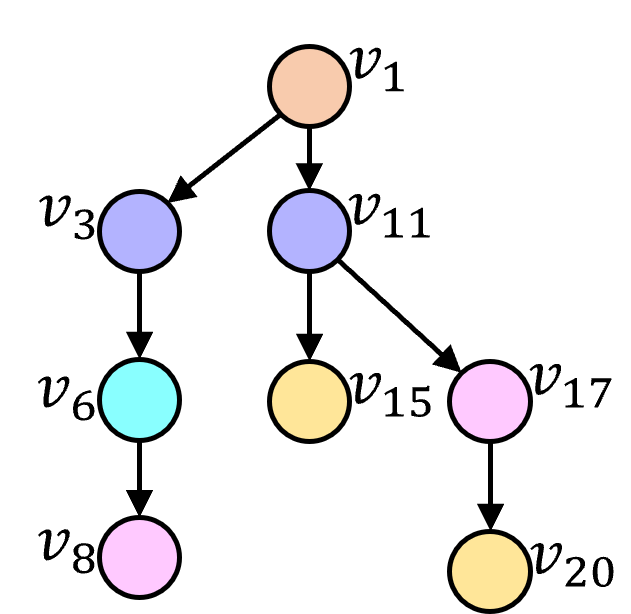}
         \caption{Algorithm~\ref{alg:approxcg-2}}
         \label{fig:comp_infer_sd_2_2}
    \end{subfigure}
    \begin{subfigure}[b]{0.1\linewidth}
         \centering
         \includegraphics[width=0.45\textwidth]{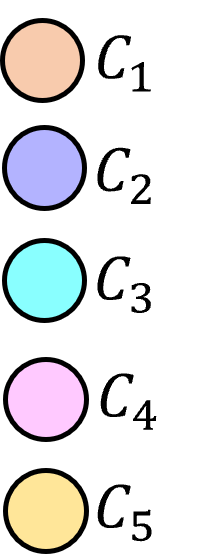}
         \caption{Clusters}
         \label{fig:comp_infer_sd_clusters}
    \end{subfigure}
    \caption{An example in which local-first generates a robust SHT while global-first does not (Figures \ref{fig:comp_infer_sd_1_true}, \ref{fig:comp_infer_sd_1_1}, \ref{fig:comp_infer_sd_1_2}), and vice versa (Figures \ref{fig:comp_infer_sd_2_true}, \ref{fig:comp_infer_sd_2_1}, \ref{fig:comp_infer_sd_2_2}). Colors indicate the cluster membership of each node in \autoref{fig:comp_infer_sd_clusters}.}
    \label{fig:comp_infer_sd}
\end{figure*}

Consider
documents $D_1, D_2$ with true SHTs
shown
in Figures~\ref{fig:comp_infer_sd_1_true} and~\ref{fig:comp_infer_sd_2_true}, respectively.
\keep{The clusters are
indicated
by colors in \autoref{fig:comp_infer_sd_clusters}
(e.g.,
$\small C_5 = \{v_{15}\}$
in $D_1$
and
$\small \{v_{15}, v_{20}\}$
in $D_2$).
For both documents, local-first assigns semantic depths $\small sd_1(C_{1, 2, 3, 4, 5}) = 0, 1, 2, 3, 2$; global-first assigns
$\small sd_2(C_{1..5}) = 0, 1, 2, 3, 4$.}

\keep{When SD conformance holds, the corresponding approach infers a robust SHT: local-first in $D_1$ under $sd_1$
and global-first in $D_2$ under $sd_2$
(Figures \ref{fig:comp_infer_sd_1_1} and \ref{fig:comp_infer_sd_2_2}).
When SD conformance is violated, non-robust SHTs result:
global-first in $D_1$ due to $sd_2(v_{15}) > sd_2(v_{17})$
and local-first in $D_2$ due to $sd_1(v_{17}) > sd_1(v_{20})$
(Figures \ref{fig:comp_infer_sd_1_2} and \ref{fig:comp_infer_sd_2_1}).}

\hidec{}

\subsubsection{A Document Class Favoring Local-First Inference}
\label{subsec:diff-sd-class}

\hidec{}
\keep{We next identify a document class for which local-first guarantees robustness, while global-first does not
(e.g., the document whose true SHT is in  \autoref{fig:comp_infer_sd_1_true}).
We call this class {\em depth-aligned} and provide a detailed explanation, as well as a proof, in Appendix~\ref{append:depth_aligned}.}
\hidec{}

\begin{definition} [Depth-Aligned Documents.]~\label{defn:depth_aligned}
    A document $D$, with nodes $V$, clusters $\mathcal{C}$, and true SHT $T'$ is {\em depth-aligned} if
    1) $\forall~C\in \mathcal{C}$, $v_{ind(C)}$, whenever it is not the root, is the leftmost child of its parent in $T'$, and
    2) for any $v_i, v_j \in V$ such that $C(v_i) = C(v_j)$, $v_i$ and $v_j$ are at the same level in $T'$.
\end{definition}

\hidec{}

\keep{For a depth-aligned document, the semantic depth assigned by local-first matches the true depth (i.e., the number of edges from the root to the node) of every node in its true SHT, guaranteeing SD conformance and thus robustness:}

\begin{theorem}~\label{thm:depth_aligned}
    Consider a document $D$, with nodes $V$, clusters $\mathcal{C}$, and true SHT $T'$.
    Let $sd$ denote the semantic depths inferred by local-first SD inference,
    and let $T$ denote the SHT inferred by Algorithm~\ref{alg:shtgen} from $sd$.
    If $D$ is depth-aligned, then $\forall~v_i\in V$, $sd(v_i)$ equals the depth of $v_i$ in $T'$, and $T$ is robust.
\end{theorem}
\begin{myproof}
    At the moment before inserting node $v_i$ into $T$, assume that all nodes currently in $T$ have semantic depths equal to their true depths in $T'$.
    If $v_i$ is {\em not} the indexing node of its cluster, then by condition 2) of depth alignment, $sd(v_i)$ equals the true depth of the indexing node of $C(v_i)$, which in turn equals the true depth of $v_i$ in $T'$.
    If $v_i$ is the indexing node of its cluster, then by condition 1) of depth alignment, $sd(v_i) = sd(v_j) + 1$, where $v_j$ is the true parent of $v_i$ in $T'$. By our assumption that $sd(v_j)$ equals the true depth of $v_j$ in $T'$, it follows that $sd(v_i)$ also equals its true depth in $T'$.
\end{myproof}

\hidec{}

\subsection{Semantic Depth Inference Family}
\label{subsec:family-infer-sd}

We can further generalize SD inference by combining local- and global-first approaches into {\em an infinite family of approaches}.

Recall that local-first assigns $sd_1(C)$ to be one greater than $sd_1$ of the cluster corresponding to the node immediately prior to its indexing node (Algorithm~\ref{alg:build_cg_1}),
while global-first scans all nodes prior to its indexing node and assigns $sd_2(C)$ as one greater than the largest $sd_2$ among their clusters (Algorithm~\ref{alg:approxcg-2}).

We extend these two approaches by combining $sd_1$ and $sd_2$ linearly as
$sd(C) = \left(K(C)\cdot sd_1 + (1 - K(C)) \cdot sd_2\right) + 1$,
where $K(C)\in [0, 1]$ is a cluster-specific parameter.
Setting $K(C)=1 (0)$ for all clusters recovers the local-first (global-first) approach.
Arbitrary mappings $K: \mathcal{C} \to [0,1]$ define an infinite family of SD inference approaches $\mathcal{K}$, each denoted as $\code{Infer-SD}_K$\hidec{}\keep{; pseudo code can be found in~Appendix~\ref{append:infer_sd_family}.}

\hidec{}

\keep{$\code{Infer-SD}_K$ guarantees SHT robustness whenever the input document satisfies SD conformance (Theorem~\ref{theo:correctness-general}).}
\hidec{}
Because SD conformance depends on the semantic depths $sd$ inferred by $\code{Infer-SD}_K$, different $K$s induce different $sd$s and thus different conformance conditions.
Therefore, the document classes that are necessary and sufficient to guarantee robustness for $\code{Infer-SD}_K$ vary across different parameter mappings $K$.
In \Cref{subsec:hierarchy-robustness}, we identify a document class, termed {\em loosely-formatted} documents, which lies at the intersection of all such document classes, guaranteeing robustness for all approaches in  $\mathcal{K}$ induced by varying the mappings $K$.

\subsection{Hierarchy of Robustness}
\label{subsec:hierarchy-robustness}

In this section, we study two additional document classes: {\em well-formatted documents} proposed by Lin et al.~\cite{zendb}, and what we term {\em loosely-formatted documents}, which encompasses well-formatted documents and ensures SHT robustness for the family of SD inference approaches, $\mathcal{K}$, introduced in \Cref{subsec:family-infer-sd}.
We then summarize the relationships between all document classes considered thus far in the form of a hierarchy in \autoref{fig:hierarchy-robustness} and estimate their frequencies, as well as document classes that satisfy SD conformance for different approaches in $\mathcal{K}$.

\topic{Well-formatted Documents}
In a well-formatted document,
any pair of sibling nodes (i.e., nodes with the same parent) in the true SHT $T'$ belongs to the same cluster.
Additionally,
for any parent-child node pair in $T'$,
this relationship (i.e., which node is the parent and which is the child)
is consistent with the parent-child relationship observed at the first occurrence of their corresponding visual patterns\footnote{This statement is equivalent to ZenDB’s original definition of well-formattedness \keepc{(Definition~1 in~\cite{zendb})}, which we will establish in \hide{}\keepc{Appendix~\ref{append:relation_zendb}}.}.
Concretely, in $T'$, if $v_i$ is the parent of $v_j$, then, at the first occurrence of the visual pattern corresponding to $v_j$ (i.e., at node $v_l$, the indexing node of $C(v_j)$), $v_l$'s parent $v_k$ has a visual pattern that matches $v_i$. Here, $(v_k, v_l)$ is the first parent-child occurrence of the visual patterns corresponding to $(v_i, v_j)$.

For well-formatted documents, any approach in $\mathcal{K}$
produces robust SHTs after applying Algorithm~\ref{alg:shtgen}.
However, these documents constitute a small fraction of real-world documents, e.g., 23.39\% across 600 real-world documents (Appendix~\ref{append:datasets}).
So, instead\hidec{}, we define a broader document class below, called {\em loosely-formatted documents}, and prove that they also guarantee robustness for all approaches in $\mathcal{K}$ in Theorem~\ref{thm:loosely-robustness}.
Since well-formatted documents are a special case of loosely-formatted documents (see Theorem~\ref{thm:well-loosely-relationship} below), their robustness guarantees follow immediately.

\topic{Loosely-formatted Documents}
These documents relax well-formattedness by removing the requirement that siblings belong to the same cluster. And the requirement that each parent-child relationship in the true SHT matches the first {\em parent-child} occurrence of the corresponding visual patterns is replaced by the weaker requirement that it matches the first {\em ancestor-descendant} occurrence of the corresponding visual patterns.
Concretely, in $T'$, if $v_i$ is the parent of $v_j$,
then $v_l$ must have an {\em ancestor} $v_k$ whose visual pattern is the same as that of $v_i$,
where $(v_k, v_l)$ is the first ancestor-descendant occurrence of the visual patterns corresponding to $(v_i, v_j)$.
For example, for the node pair $(v_i, v_j) = (v_1, v_{16})$ in \autoref{fig:true_sht}, their corresponding visual patterns first appear as an ancestor-descendant pair at $(v_k, v_l) = (v_1, v_{10})$.
This ancestor-descendant relationship
correctly indicates the parent-child relationship between
$v_1$ and $v_{16}$.
\begin{definition} [Loosely-Formatted Documents.]~\label{def:loosely-format}
    A document $D$, with nodes $V$, clusters $\mathcal{C}$, and true SHT $T'$ is {\em loosely-formatted} if
    $\forall~v_i, v_j\in V$, when $v_i$ is the parent of $v_j$ in $T'$, there exists a pair of nodes $v_k, v_l \in V$ such that $v_k$ is an ancestor of $v_l$ in $T'$ with $l = ind(C(v_j))$ and $C(v_k) = C(v_i)$.
\end{definition}

As mentioned earlier,
loosely-formatted documents relax
{well-formattedness},
\keep{as stated below and proved in Appendix~\ref{append:proof_well_loosely_relationship}. }\hidec{}
\begin{theorem}~\label{thm:well-loosely-relationship}
    Every well-formatted document is also loosely-formatted, but the converse does not hold.
\end{theorem}
\begin{myproof}
    If $D$ is well-formatted, then $\forall~v_i, v_j\in V$ where $v_i$ is the parent of $v_j$ in $T'$, $v_l$'s parent $v_s$ in $T'$ satisfies $C(v_s) = C(v_i)$, where $l = ind(C(v_j))$. Choosing $v_k = v_s$ satisfies all conditions in Definition~\ref{def:loosely-format}, and thus $D$ is loosely-formatted.
    The converse does not hold because well-formattedness does not allow parent-child relationships to be inferred from ancestor-descendant occurrences.
\end{myproof}

We further establish that applying our two-stage workflow with any SD inference approach from $\mathcal{K}$ on loosely-formatted documents guarantees robustness\keep{; proof can be found in Appendix~\ref{append:proof_loosely_robustness}}.
\begin{theorem} [Loose-Formattedness Implies SHT Robustness.]
    Let $T$ be the SHT inferred by Algorithm~\ref{alg:shtgen} after applying any approach from $\mathcal{K}$ for SD inference on a document $D$.
    If $D$ is loosely-formatted, then $T$ is robust.
    \label{thm:loosely-robustness}
\end{theorem}
\begin{myproof}
    Let the approach from $\mathcal{K}$ be $\code{Infer-SD}_K$ with an arbitrary parameter mapping $K$.
    $D$ is loosely-formatted.
    We prove this theorem by contradiction: suppose that the inferred SHT $T$ is not robust.
    By Theorem~\ref{theo:correctness-general}, there must exist a parent-child pair $(v_i, v_j)$ in the true SHT $T'$ that violates SD conformance, i.e., $sd(v_i) \geq sd(v_j).$
    Among all such violating pairs, choose $(v_i, v_j)$ with the smallest child index $j$.
    $v_{k_1}, v_{k_2}, sd_1$ and $sd_2$ are inferred by Algorithm~\ref{alg:infer-sd-family} while scanning the cluster $C(v_j)$ of $v_j$.
    We first establish two inequalities:\\
    1) $sd(v_i) \leq sd_1$.
    If $v_{k_1} = v_i$, the claim is immediate.
    Otherwise, $v_{k_1}$ is a true descendant of $v_i$ in $T'$.
    By the minimality of $j$, SD conformance holds for $(v_i, v_{k_1})$, implying $sd(v_i) < sd(v_{k_1}) = sd_1$.\\
    2) $sd(v_i) \leq sd_2$.
    This follows directly from $i < j$ and the definition of $sd_2$.\\
    We now consider $v_j$.\\
    \emph{Case (a): $v_j$ is the indexing node of its cluster.}
    By Algorithm~\ref{alg:infer-sd-family},
    $sd(v_j) = K(C(v_j)) \cdot sd_1 + \bigl(1 - K(C(v_j))\bigr) \cdot sd_2 + 1$
    for some $K(C(v_j)) \in [0,1]$.
    Since $sd(v_i) \le sd_1$ and $sd(v_i) \le sd_2$,
    it follows that $sd(v_i) < sd(v_j)$.\\
    \emph{Case (b): $v_j$ is not the indexing node of its cluster.}
    Let $v_l$ be the indexing node of $C(v_j)$, with $l = ind(C(v_j)) < j$.
    Because $D$ is loosely-formatted, there exists a node $v_k$ such that
    $v_k$ is an ancestor of $v_l$ in $T'$ and $C(v_k) = C(v_i)$.
    By the minimality of $j$, SD conformance holds for $(v_k, v_l)$, implying
    $sd(v_k) < sd(v_l)$.
    Since $sd(v_k) = sd(v_i)$ and $sd(v_l) = sd(v_j)$, we again obtain
    $sd(v_i) < sd(v_j)$.\\
    In both cases, we achieve $sd(v_i) < sd(v_j)$, a contradiction to our initial assumption that $sd(v_i) \geq sd(v_j)$.
    Therefore, the assumption that $T$ is not robust must be false, and $T$ is robust.
\end{myproof}

\topic{Hierarchy of Robustness}
We can now summarize the relationships among all document classes introduced using the Venn diagram in \autoref{fig:hierarchy-robustness}.
Well-formatted documents form the smallest class (yellow), and are strictly contained within loosely-formatted documents (Definition~\ref{def:loosely-format}, orange) via Theorem~\ref{thm:well-loosely-relationship}.
For any approach $\code{Infer-SD}_{K}$ from $\mathcal{K}$, loose-formattedness is a sufficient condition for SHT robustness (Theorem~\ref{thm:loosely-robustness}), while SD conformance is both necessary and sufficient (Theorem~\ref{theo:correctness-general}).
Consequently, for each $\code{Infer-SD}_{K}$, loosely-formatted documents are strictly contained within the class of documents that satisfy SD conformance (one such conformance class is highlighted in green).
We also highlight the corresponding conformance classes for local-first (blue) and global-first (purple) approaches, which are two instances of $\mathcal{K}$.
Loosely-formatted documents lie in the intersection of these conformance classes.
Depth-aligned documents (pink) guarantee robustness under the local-first approach, so they lie within the local-first conformance class, with possible intersections with other classes.

\topic{Frequency of Each Class}
Across 600 real-world documents spanning diverse domains from the source datasets of our evaluation (see Appendix~\ref{append:datasets}) with manually annotated true SHTs, we find 60.17\% are loosely-formatted, which is $2.57\times$ the fraction of well-formatted documents (23.39\%).
To further analyze class frequencies at scale,
we additionally collect 3,919 real-world
PDF documents~\cite{safedocs},
each with tables of contents encoded in metadata (i.e., extractable via \code{get\_toc} in {PyMuPDF}~\cite{pymupdf}),
avoiding manual true SHT construction.
Among them, 52.0\% are well formatted, 82.6\% loosely formatted, and 32.5\% depth-aligned; 85.7\% and 83.4\% satisfy SD conformance under the local- and global-first approaches, respectively.
Since local-first dominates, we focus on it in our evaluation\hide{}.

\hide{}
\label{subsec:equiv-zendb}

\keep{\hide{}}

\hidec{}

\hidec{}

\keep{\subsection{{Relationship to ZenDB}}
\label{subsec:discussion_relationships_zendb}
Lin et al.~\cite{zendb} \hide{} infer true SHTs for well-formatted documents, which
is equivalent to global-first \sys, as is proved in Appendix~\ref{append:relation_zendb}.
\sys differs from ZenDB in the following ways:
\begin{itemize}[leftmargin=*, topsep=0pt]
\itemsep0em
    \item A new problem formulation (\Cref{sec:defn}): ZenDB focuses on the true SHT; we formulate robust and compact SHT inference, which is practically useful and more feasible.
    \item A general two-stage framework (Sections~\ref{sec:sht-construction} and \ref{sec:robustness-analysis}): \sys splits SHT inference into two stages, of which ZenDB's algorithm is just one choice for the first stage, and generalizes global and local-first algorithms into an infinite family of approaches.
    \item Broad applicability (\Cref{sec:robustness-analysis}): \sys extends beyond ZenDB's narrow class ($<$25\% of real-world documents {in Appendix~\ref{append:datasets}}) to much broader classes ($>$60\%), with robustness guarantees.
\end{itemize}}
\section{Evaluation Setup}
\label{sec:eval_setup}

We
evaluate \sys\footnote{Code available at \url{https://github.com/Ruiying-Ma/SHED}} on real-world documents, assessing the robustness and compactness of its inferred SHTs and its \hide{} usefulness in agentic document question answering (QA).
Datasets and implementation of \sys are described in \Cref{sec:datasets} and \Cref{sec:sys_impl}, respectively.
We then present SHT evaluation in \Cref{sec:sht_eval}\hide{}\keep{,} agentic document QA in \Cref{sec:doc_qa}\keep{, and scalability analysis in \Cref{sec:scalability}}.

\subsection{Datasets}
\label{sec:datasets}

\keep{\begin{table}[t]
    \centering
    \small
    \caption{Statistics of our evaluation datasets. Avg. doc size: average word count per document in each dataset.\blank{} \#Docs (\#Qs): number of documents (questions) in each dataset.}
    \label{tab:dataset}
    \setlength{\tabcolsep}{4pt}
    \renewcommand{\arraystretch}{0.92}
    \begin{tabular}{@{}ccccc@{}}
    \toprule
    {\bf Statistics} & {\bf \ds{Civic}} & {\bf \ds{Contracts}} & {\bf \ds{Finance}} & {\bf \ds{Papers}}\\
    \midrule
    Avg. doc size & 7156 & 7863 & 178,905 & 17,352 \\
    \#Docs (\#Qs) & 107 & 248 & 100 & 500 \\

    \bottomrule
    \end{tabular}
\end{table}

\topic{Dataset creation}
We collected real-world PDF documents spanning
diverse lengths, domains,
and document templates, with QA sets
from existing document QA benchmarks:
civic agenda reports from~\cite{zendb, civic} (\ds{Civic}),
legal contracts from~\cite{contract-nli} (\ds{Contracts}),
financial filings from~\cite{financebench-paper} (\ds{Finance}),
and scientific papers from~\cite{qasper-paper} (\ds{Papers}).
Since these original datasets contain either
simple documents (e.g., \textasciitilde2K tokens
with shallow, single-level true SHTs)
or simple natural-language questions
(e.g., answerable from a single
sentence in the queried document),
they are insufficient for stress-testing our techniques
for SHT inference and document QA.
We therefore extend our datasets by composing
documents.
Specifically, for each question $q$
associated with a document $D$,
we replace $D$ with a composite document $D'$, constructed by concatenating $D$ with $k$ randomly selected documents
from the same dataset (e.g., $k=3$ for \ds{Civic}),
each preceded by a title page bearing its name.
Answering $q$ now
requires identifying $D$,
which appears as a section in $D'$, making the
task more challenging.
To construct the true SHTs for each composite document $D'$,
we first {\em manually} construct the SHTs
of its constituent documents, and concatenate their roots.
We use the modified datasets
for evaluation.

Dataset statistics are in \autoref{tab:dataset} (details in Appendix~\ref{append:datasets}), which reports average document size, the number of documents (\#Docs), and questions (\#Qs) per dataset. Here, \#Docs equals \#Qs, as each question is associated with a randomly chosen document.
\ds{Finance} has long documents, followed by \ds{Papers}; \ds{Civic} and \ds{Contracts} are shorter. \ds{Civic} \hide{}{has} complex nested structures, while \ds{Contract} \hide{}{has} dense layouts (i.e., short spans between headers).
}%

\keep{To evaluate how \sys performs on large-scale data, we additionally collect 3 datasets with complex hierarchical structure from different domains (law~\cite{cfr}, telecom~\cite{etsi}, environment~\cite{ferc}) with synthesized questions and manually labeled ground-truth answers, as we will discuss in \Cref{sec:scalability}.}

\hidec{}

\subsection{\sys Implementation}
\label{sec:sys_impl}

For header identification,
we use small pretrained models that can be run on CPUs from {Huridocs}~\cite{vgt-code}.
Following {Huridocs}' recommendations, we use a smaller model~\cite{lightgbm} for large-scale datasets
(\ds{Finance} and \ds{Papers}), trading slightly lower accuracy for substantially faster execution than the larger model~\cite{vgt-code, vgt-paper}.
For each header phrase, we extract its visual pattern using {PyMuPDF}~\cite{pymupdf}.
The header phrases are treated as nodes $V$ and clustered by their visual patterns, yielding clusters $\mathcal{C}$.
Finally, we use local-first SD inference, and construct SHTs $T$ using Algorithm~\ref{alg:shtgen}.
We adopt local- over global-first since more documents satisfy SD conformance under local-first\hide{}\keep{~(}as shown in prior statistics \hide{}\keep{in~}\Cref{subsec:hierarchy-robustness}\hide{}\keep{), for which local-first infers more robust and compact SHTs (Appendix~\ref{append:global_first_sht_infer})}.
We implemented \sys using \textasciitilde1.3k lines of Python.

\section{Evaluation of SHT Inference}
\label{sec:sht_eval}

We evaluate robustness and compactness of SHTs inferred by \sys and five baselines on documents from four datasets.
{\em Overall, \sys achieves average F-1 scores 13\%--68\% higher than non-LLM baselines and 9\%--15\% higher than significantly more expensive LLM-based approaches under two metrics, demonstrating its advantage in simultaneously ensuring robustness and optimizing compactness.}
We describe the baselines and metrics in \Cref{sec:setup_sht_eval},
and present evaluation results in \Cref{sec:sht_eval_results}.

\subsection{SHT Inference Experimental Setup}
\label{sec:setup_sht_eval}

We compare
\sys against five baselines in terms of SHT robustness and compactness
using two metrics.

\subsubsection{Baselines}
We have five baselines spanning rule-based, custom model-based, and LLM-based approaches:

\topic{Deep}
This approach uses the same nodes $V$ and clusters $\mathcal{C}$ as \sys, and constructs an SHT by linking the header nodes into a list in index order (e.g., \autoref{fig:deep_and_wide_shts}-1),
guaranteeing robustness at the cost of compactness.

\topic{Wide}
Using the same nodes $V$ and clusters $\mathcal{C}$, this approach selects the first node as the root and the remaining nodes as its children (e.g., \autoref{fig:deep_and_wide_shts}-2),
with low robustness and high compactness.

\topic{GROBID}
To infer an SHT for a composite document, this approach infers an SHT for each individual document using a model trained on scientific papers~\cite{grobid}, adds its title as a root node, and merges these individual document SHTs under a single virtual root. Since the model is not applied directly to composite documents, we treat GROBID's evaluation results as an optimistic upper bound.

\topic{LLM-text}
This approach extracts plain text from $D$ using PyMuPDF and prompts an LLM (GPT-5.4) to infer an SHT from the text in one shot (see prompt in
\keep{Appendix~\ref{append:prompts-sht}}\hidec{}).

\topic{LLM-vision}
This approach prompts an LLM (\hide{}GPT-5.4) with page images from $D$ to infer SHTs.
Pages are grouped sequentially into batches of 50 pages due to LLM input limits (prompts in~\keep{Appendix~\ref{append:prompts-sht}}\hidec{}); SHTs inferred per batch are merged under a virtual root.
\keep{This approach is state-of-the-art among \hide{} those jointly modeling PDF layout and semantics, inferring better structure than  others (, e.g., SmolDocling-256M-preview~\cite{nassar2025smoldocling}; Appendix~\ref{append:smoldocling})}.

\begin{table*}[t]
    \centering
    \caption{
Recall, precision, and F-1 scores for identified headers relative to the true headers, averaged per dataset and then across datasets (stratified average {\bf Avg.}), as well as the total cost (USD) of SHT inference per dataset ({\bf Tot.}: summed across datasets). \sys, Deep, and Wide use the same set of headers extracted with Huridocs; we report \sys only as the other two are identical on header inference and cost. \best{Green}: highest recall, precision, F-1 scores, and lowest cost per dataset}
    \label{tab:header_identification}
    \footnotesize
    \setlength{\tabcolsep}{0.7pt}
    \renewcommand{\arraystretch}{0.92}
    \begin{tabular}{@{}lrrrrr|rrrrr|rrrrr|rrrrr@{}}
    \toprule
    \multirow{2}{*}{\bf Approach}
    & \multicolumn{5}{c}{\bf Recall}
    & \multicolumn{5}{c}{\bf Precision}
    & \multicolumn{5}{c}{\bf F-1}
    & \multicolumn{5}{c}{\bf Total Cost (USD)}\\
    \cmidrule(lr){2-6}
    \cmidrule(lr){7-11}
    \cmidrule(lr){12-16}
    \cmidrule(lr){17-21}
    {} & {\bf \ds{\footnotesize Civic}} & {\bf \ds{\footnotesize Contracts}} & {\bf\ds{\footnotesize Finance}} & {\bf\ds{\footnotesize Papers}} & {\bf Avg.}
       & {\bf \ds{\footnotesize Civic}} & {\bf \ds{\footnotesize Contracts}} & {\bf\ds{\footnotesize Finance}} & {\bf\ds{\footnotesize Papers}} & {\bf Avg.}
       & {\bf \ds{\footnotesize Civic}} & {\bf \ds{\footnotesize Contracts}} & {\bf\ds{\footnotesize Finance}} & {\bf\ds{\footnotesize Papers}} & {\bf Avg.}
       & {\bf \ds{\footnotesize Civic}} & {\bf \ds{\footnotesize Contracts}} & {\bf\ds{\footnotesize Finance}} & {\bf\ds{\footnotesize Papers}} & {\bf Tot.}\\
    \midrule
    GROBID
    & 0.14 & 0.21 & 0.61 & 0.85 & 0.45
    & 0.73 & 0.61 & 0.58 & 0.85 & 0.69
    & 0.24 & 0.30 & 0.58 & 0.84 & 0.49
    & \best{0} & \best{0} & \best{0} & \best{0} & \best{0}\\

    LLM-text
    & 0.91 & 0.46 & 0.56 & 0.86 & 0.70
    & 0.84 & 0.80 & 0.52 & 0.88 & 0.76
    & 0.87 & 0.56 & 0.48 & 0.87 & 0.70
    & \hide{}\keep{5.7} & \hide{}\keep{10.7} & \hide{}\keep{104.0} & \hide{}\keep{36.3} & \hide{}\keep{156.6} \\

    LLM-vision
    & 0.91 & 0.37 & 0.72 & 0.85 & 0.71
    & 0.82 & 0.73 & 0.48 & 0.85 & 0.72
    & 0.85 & 0.47 & 0.55 & 0.85 & 0.68
    & \hide{}\keep{12.0} & \hide{}\keep{19.9} & \hide{}\keep{88.6} & \hide{}\keep{43.4} & \hide{}\keep{164.0} \\

    \sys
    & \best{0.93} & \best{0.85} & \best{0.95} & \best{0.92} & \best{0.91}
    & \best{0.89} & \best{0.93} & \best{0.60} & \best{0.90} & \best{0.83}
    & \best{0.91} & \best{0.88} & \best{0.72} & \best{0.91} & \best{0.86}
    & \best{0} & \best{0} & \best{0} & \best{0} & \best{0}\\
    \bottomrule
    \end{tabular}
\end{table*}

\hide{}{When analyzing}  correctness, robustness,
and compactness, we assume that the
node set $V$ used by \sys is identical to that of the true SHT $V'$.
This assumption need not hold
due to errors in header identification,
which we will evaluate and discuss in \Cref{sec:sht_eval_results}.

\subsubsection{Evaluation Metrics}
We design two metrics to evaluate
robustness and compactness: {\em top-down}, denoted $M_{\downarrow}$, and {\em bottom-up}, denoted $M_{\uparrow}$.
Given an inferred SHT $T$ (with header nodes only) for a document with true SHT $T'$,
$M_{\downarrow}$ compares, for each node $v$ in $T'$, the text spans $ts(v)$ and $ts'(v)$ inferred from $T$ and $T'$, respectively,
motivated by the definition of robustness (Definition~\ref{def:robust-sht}), which requires $ts'(v) \subseteq ts(v)$.
In contrast, $M_{\uparrow}$ operates in reverse: it compares %
the lists of ancestor headers of $v$ inferred from $T$ and $T'$, denoted $H(v)$ and $H'(v)$, respectively,
motivated by Theorem~\ref{thm:ts-correct}, which requires $H'(v) \subseteq H(v)$.

Let the node sets of $T$ and $T'$ be $V$ and $V'$, respectively. As mentioned earlier, $V$ may differ from $V'$.
We now define the top-down and bottom-up metrics for $T$.

\topic{\keep{Text-span overlap}}
\keep{Both metrics rely on the text-span overlap $ts(v)\cap ts'(v)$, computed as follows: we normalize $ts(v)$ and $ts'(v)$ (by canonicalizing characters, removing punctuation, lowercasing,
and normalizing whitespace), split each span into whitespace-separated tokens (i.e., a multiset of words), and take the intersection of the two multisets.
Text spans include all non-image text (e.g., tables, captions, page headers and footers) but exclude section headers (this has negligible impact since headers are far shorter than the text spans).
Long or empty sections are treated like other sections.}

\keep{\topic{Top-down ($M_{\downarrow}$)}
Consider a node $v\in V'$.
Suppose $v\in V$;
we consider the other case later.
$M_{\downarrow}$ compares $ts(v)$ (i.e., a set of words) with $ts'(v)$
in terms of recall ${\small \text{REC}(v)=\frac{|ts(v) \cap ts'(v)|}{|ts'(v)|}}$, precision \\ ${\small \text{PREC}(v)=\frac{|ts(v) \cap ts'(v)|}{|ts(v)|}}$, and F-1 score $\small {\text{F-1}(v)=\frac{2|ts(v) \cap ts'(v)|}{|ts(v)| + |ts'(v)|}}$,
which respectively measure robustness, compactness, and their trade-off.
The overall top-down robustness and compactness of $T$ are computed by averaging over all nodes in $V'$: $M_{\downarrow}(T) = \text{avg}_{v\in V'} f(v)$, where $f\in \{\text{REC, PREC, F-1}\}$.
We average over $V'$ to ensure that SHTs inferred from the same document by different methods, potentially with different node sets, are evaluated on a common basis.

If $v\notin V$,
we select the ``closest'' ancestor node $u \in V$ whose text span $ts(u)$ contains  $v$'s header, such that no child of $u$ has a text span that contains $v$'s header, to compute $f(v)$. If multiple such nodes $u$ exist (e.g., $v$ is referenced by multiple sections),
we choose the one that maximizes $f(v)$.

\topic{Bottom-up ($M_{\uparrow}$)}
$M_{\uparrow}(T)$ is computed analogously to $M_{\downarrow}(T)$,
except that  $ts(v)$ (resp.~$ts'(v)$) of each node $v\in V'$ is replaced by $H(v)$ ($H'(v)$),
where $H(v)$ ($H'(v)$) is defined as the list of header phrases corresponding to the nodes in the path from the root to $v$ in $T$ ($T'$), ordered by phrase indices.
}%

\hidec{}

\subsection{SHT Inference Results}
\label{sec:sht_eval_results}

We compare \sys with baselines in terms of header identification accuracy, SHT inference cost, and SHT robustness and compactness.

\begin{table*}[t]
    \centering
    \caption{$M_{\downarrow}$ and $M_{\uparrow}$ results, averaged per dataset and evaluated with $f\in \{\text{REC}, \text{PREC}, \text{F-1}\}$ for recall, precision, and F-1 score, representing robustness, compactness, and their trade-off, respectively. {\bf Avg.}: stratified averages across datasets. \best{Green}: highest per dataset per metric; costs for each scheme can be found in \Cref{tab:header_identification}.}
    \label{tab:sht_eval_results}
    \footnotesize
    \setlength{\tabcolsep}{1.8pt}
    \renewcommand{\arraystretch}{0.92}
    \begin{tabular}{@{}lrrrrr|rrrrr|rrrrr@{}}
    \toprule
    \multirow{2}{*}{\bf Approach}
    & \multicolumn{5}{c}{\bf Recall (Robustness)}
    & \multicolumn{5}{c}{\bf Precision (Compactness)}
    & \multicolumn{5}{c}{\bf F-1 (Trade-off)}\\
    \cmidrule(lr){2-6}
    \cmidrule(lr){7-11}
    \cmidrule(lr){12-16}
    {} & {\bf \ds{\footnotesize Civic}} & {\bf \ds{\footnotesize Contracts}} & {\bf\ds{\footnotesize Finance}} & {\bf\ds{\footnotesize Papers}} & {\bf Avg.}
       & {\bf \ds{\footnotesize Civic}} & {\bf \ds{\footnotesize Contracts}} & {\bf\ds{\footnotesize Finance}} & {\bf\ds{\footnotesize Papers}} & {\bf Avg.}
       & {\bf \ds{\footnotesize Civic}} & {\bf \ds{\footnotesize Contracts}} & {\bf\ds{\footnotesize Finance}} & {\bf\ds{\footnotesize Papers}} & {\bf Avg.}\\
    \midrule

    \multicolumn{16}{l}{\cellcolor{gray!10} \em Top-down: {$M_{\downarrow}(T) = avg_{v\in T'}f(v)$, where $f(v)$ calculates recall, precision, and F-1 score of $ts(v)$ relative to $ts'(v)$.}} \\

    Deep
    & 0.94 & \best{0.99} & \best{1.00} & \best{0.97} & \best{0.98}
    & 0.21 & 0.15 & 0.09 & 0.20 & 0.16
    & 0.29 & 0.20 & 0.13 & 0.27 & 0.22 \\

    Wide
    & 0.70 & 0.81 & 0.81 & 0.78 & 0.78
    & 0.93 & \best{0.90} & \best{0.97} & \best{0.97} & \best{0.94}
    & 0.72 & 0.74 & 0.81 & 0.79 & 0.77 \\

    GROBID
    & 0.48 & 0.84 & 0.74 & 0.76 & 0.71
    & 0.25 & 0.45 & 0.75 & 0.83 & 0.57
    & 0.26 & 0.47 & 0.66 & 0.76 & 0.54 \\

    LLM-text
    & 0.96 & 0.93 & 0.76 & 0.96 & 0.90
    & \best{0.98} & 0.67 & 0.61 & 0.96 & 0.81
    & 0.96 & 0.67 & 0.57 & 0.95 & 0.79 \\

    LLM-vision
    & \best{0.97} & 0.89 & 0.94 & 0.95 & 0.94
    & 0.97 & 0.61 & 0.77 & 0.94 & 0.82
    & \best{0.97} & 0.60 & 0.73 & 0.93 & 0.81 \\

    \sys
    & 0.90 & 0.97 & 0.94 & 0.96 & 0.94
    & 0.91 & 0.88 & 0.95 & 0.95 & 0.92
    & 0.89 & \best{0.87} & \best{0.90} & \best{0.94} & \best{0.90} \\

    \multicolumn{16}{l}{\cellcolor{gray!10} \em Bottom-up: {$M_{\uparrow}(T) = avg_{v\in T'}f(v)$, where $f(v)$ calculates recall, precision, and F-1 score of $H(v)$ relative to $H'(v)$.}} \\

    Deep
    & \best{0.94} & \best{0.91} & \best{0.98} & \best{0.85} & \best{0.92}
    & 0.28 & 0.26 & 0.15 & 0.33 & 0.26
    & 0.39 & 0.28 & 0.20 & 0.41 & 0.32 \\

    Wide
    & 0.32 & 0.41 & 0.26 & 0.25 & 0.31
    & \best{0.94} & \best{0.94} & \best{0.98} & \best{0.95} & \best{0.95}
    & 0.46 & 0.53 & 0.38 & 0.34 & 0.43 \\

    GROBID
    & 0.20 & 0.48 & 0.37 & 0.73 & 0.45
    & 0.38 & 0.73 & 0.57 & 0.91 & 0.65
    & 0.25 & 0.52 & 0.40 & 0.78 & 0.49 \\

    LLM-text
    & 0.76 & 0.60 & 0.40 & 0.79 & 0.64
    & 0.71 & 0.74 & 0.60 & 0.94 & 0.75
    & 0.72 & 0.62 & 0.45 & 0.83 & 0.66 \\

    LLM-vision
    & 0.92 & 0.47 & 0.50 & 0.77 & 0.67
    & 0.92 & 0.79 & 0.80 & \best{0.95} & 0.87
    & 0.91 & 0.55 & 0.57 & 0.82 & 0.71 \\

    \sys
    & 0.92 & 0.82 & 0.69 & 0.82 & 0.81
    & 0.93 & 0.92 & 0.71 & 0.93 & 0.87
    & \best{0.93} & \best{0.83} & \best{0.65} & \best{0.84} & \best{0.81} \\

    \bottomrule
    \end{tabular}
\end{table*}

\topic{Header identification}
As in \autoref{tab:header_identification},
\sys (which uses Huridocs) consistently achieves high accuracy,
whereas GROBID performs poorly on datasets \ds{Civic}, \ds{Contracts}, and \ds{Finance} that are outside its training domain.
LLM-based approaches perform poorly on \ds{Contracts} and \ds{Finance}. For \ds{Contracts}, which has a dense layout (i.e., short spans between headers), LLMs tend to ignore lower-level headers and focus only on top-level ones, leading to low recall. For \ds{Finance}, which contains long documents, LLMs lose track of headers midway~\cite{lost-in-the-middle}; this issue also occurs, to a lesser extent, in other datasets.
Moreover, 0.9\% (3.7\%) of headers identified by LLM-text (LLM-vision) are hallucinations that do not appear in the document.

\topic{Cost}
\keep{\autoref{tab:header_identification} reports the {\em end-to-end} cost of inferring SHTs
from a PDF, including preprocessing (i.e., layout parsing and visual pattern extraction) and tree construction.}
Only LLM-text and LLM-vision incur LLM costs, in total U.S. \hide{}\keep{\$157} and \hide{}\keep{\$164} across all datasets, respectively, using GPT-5.4.
\sys and other baselines incur no LLM cost\keep{, because they use open-source pre-trained models and Python programs for these steps.}

\topic{Robustness, compactness, and trade-off}
As shown in \autoref{tab:sht_eval_results},
across almost all datasets
for both $M_{\downarrow}$ and $M_{\uparrow}$,
Deep (Wide) has the highest values of recall
(precision, resply.) among all approaches.
Both approaches are closely
followed by \sys with minor
reductions in recall and precision,
whereas the remaining baselines
fail to achieve comparable performance.
However, both Deep and Wide exhibit low F-1 scores for both $M_{\downarrow}$ and $M_{\uparrow}$:
deep trees have poor compactness (and therefore precision),
while wide trees have poor robustness (i.e., recall).
GROBID has higher $M_{\downarrow}$ and $M_{\uparrow}$ values of recall, precision, and F-1 scores on \ds{Papers} than on the other datasets,
as \ds{Papers} closely aligns with its training domain,
showing GROBID's limited generalizability.
LLM-vision outperforms LLM-text in average F-1 score in both $M_{\downarrow}$ and $M_{\uparrow}$, \hide{}due to the additional visual information \hide{}that helps LLMs infer structures more accurately.

\topic{\keep{True header ablation}}
\keep{To isolate structure inference from header identification errors, we evaluate \sys with the true headers $V'$: this improves average recall, precision, and F-1 by only 3\%--6\% under both $M_{\downarrow}$ and $M_{\uparrow}$ (see Appendix~\ref{append:true_header_abl} for details). Subsequent evaluations therefore use the extracted headers, which perform comparably and reflect the realistic setting where $V'$ is unavailable.}

\section{Evaluation of Agentic Question Answering}
\label{sec:doc_qa}

We evaluate agentic document QA under four strategies, two with SHT assistance and two without.
{\em Overall, our SHT-based agent
achieves the highest average accuracy
and lowest cost, improving accuracy by 7\% and 19\%
over two non-SHT strategies,
respectively, while reducing cost by \hide{}\keep{$9\times$ and $10\times$} compared to strategies that process documents (and SHTs) in context.}
\keep{We also evaluate various agentic retrieval baselines, and {\em our SHT-based agent achieves 3\%--23\% higher accuracy while being 1.3--2$\times$ cheaper than others.}}
We further compare \sys against
SHT inference baselines using our SHT-based agent.  {\em \sys achieves accuracy and cost comparable to true SHTs, while outperforming other SHT inference approaches by 6\%--19\% in average accuracy, while reducing total cost by up to \hide{}\keep{$9\times$}.} We describe the four strategies and metrics in \Cref{sec:setup_doc_qa},  present the strategy evaluation results in \Cref{sec:eval_qa_sht},
\keep{the retrieval ablation results in \Cref{subsec:retrieval-ablation},}
and the SHT ablation results in \Cref{sec:eval_qa_sht_ablation}.

\subsection{Agentic QA Evaluation Setup}
\label{sec:setup_doc_qa}

\subsubsection{Evaluated Strategies}
\label{subsec:strategies_doc_qa}
We evaluate four strategies, spanning in-context processing~\cite{rag-or-long-context} and tool-based retrieval~\cite{bigeard2025finance}, with and without SHT assistance. We use GPT-5.4 as the backbone model.

For in-context strategies, a single LLM
API call places all
information (documents with/without SHTs) in the model context.
For tool-based retrieval, we evaluate ReAct-style agents equipped with different retrieval tools~\cite{yao2022react}.
Here, agents have no direct access to
the full document text;
instead, at each iteration an
agent reasons over the current {\em context},
including prior tool calls and their results, and generates a response via an LLM API call (with up to three retries for non-200 status codes), which may contain one or more tool calls.
The runtime executes these tool calls
and appends their outputs (i.e., retrieved content)
to the context for the next iteration.
We implement the agents using LangChain~\cite{langchain},
capped at 100 iterations.

We evaluate the following four strategies, each of
which includes the question as part of the prompt,
but differ in terms of how much document context
and which tools are provided:

\topic{Vanilla-in-context}
This strategy places the full serialized document text (obtained via \code{get\_text} from PyMuPDF) into the model context,
without using SHTs.

\topic{Vanilla-grep agent}
This strategy
gives the agent a \code{grep} tool
that takes a regex pattern and returns matching document chunks,
each containing
one match with 1000 characters of context.
Results are {\em streamed} across calls:
for a pattern, the first call returns the earliest match in reading order and subsequent calls
successive matches.
If the pattern changes,
the search restarts from the beginning;
if a previous pattern is reused
after other patterns,
the tool resumes from its last returned position. The agent decides when to stop searching a pattern.
\hide{}

\topic{SHT-augmented-context}
This strategy augments model context
with the {\em true} SHT corresponding to the document
prepended to the document text,
listing section headers one per line in reading order.
Each line begins with a {\em numeric prefix} indicating the section's hierarchical level, followed by a separator | and the section header (e.g., ``{\em 3.1 | Consolidated Statements of Income}'').
\blank{}
\hide{}

\topic{SHT-based agent}
This strategy provides an agent
with the {\em true} SHT for the document in-context
(via the user prompt, formatted as in the previous strategy).
The agent reasons over the SHT and
uses a \code{read\_section}
tool that takes a section's numeric
prefix\hide{}
and returns the corresponding section content
(i.e., text span).
\blank{}

\subsubsection{Evaluation Metrics}
We use task-specific metrics to evaluate the above strategies.
For \ds{Civic} questions that ask for a {\em set} of civic projects,
we use F-1 score.
For \ds{Contracts} questions with a single
true answer drawn from a fixed-size set, we use 0-1 accuracy.
For \ds{Finance} and \ds{Papers}, which involve free-form or open-ended answers,  we use LLM-as-judge~\cite{llm-as-a-judge} with GPT-4o-mini, following established practice~\cite{pdf-qa, grace, finsage, finbench-llmjudge} (prompts in Appendix~\ref{append:prompts-llm-as-a-judge}).
\keep{We report the total cost of end-to-end SHT inference and QA.}

\subsection{Usefulness of SHTs}
\label{sec:eval_qa_sht}

We report average accuracy
and total cost in \autoref{tab:eval_doc_qa_sht_usefulness_acc_cost}.

\topic{Accuracy}
Overall, SHTs improve accuracy over non-SHT strategies\hide{}: SHT-augmented-context
outperforms vanilla-in-context by 7\%,
and the SHT-based agent by 19\% over vanilla-grep.\hide{}
{Using SHTs} does not provide any
gains for \ds{Contracts} and \ds{Papers}
{(SHT-augmented-context vs. vanilla-in-context)},
while\hide{}
{the SHT-based agent improves accuracy by 30\% on \ds{Contracts} and 23\% on \ds{Papers} over vanilla-grep,}
\blank{}
likely because full document \hide{} provides enough context that enables LLMs to infer hierarchical structure. \hide{}
With SHT assistance, \hide{} the SHT-based agent has 6\% lower accuracy than SHT-augmented-context on \ds{Civic} because the questions require cross-section aggregation; the agent misses correct answers \hide{}\keepc{by checking only relevant sections in isolation}. In contrast, the agent achieves 5\% higher accuracy on \ds{Contracts} because the questions require reasoning over multiple sections but without cross-section dependencies,\hide{}
\keepc{making SHT-augmented-context more prone to context rot~\cite{context-rot} when all sections are processed in one shot.}
\blank{}

\topic{Cost}
\keepc{The SHT-based agent is overall \hide{}\keep{$10\times$ and $9\times$} cheaper than SHT-augmented-context and vanilla-in-context, respectively, especially on the large-scale \ds{Finance} setting (\hide{}\keep{$30\times$ and $29\times$}), because the agent loads only small relevant sections.
Compared with vanilla-grep\keepc{,} \hide{} the SHT-based agent is overall \hide{}\keep{$1.3\times$} cheaper, because SHTs provide additional semantics that enable more precise retrieval than keyword search (\code{grep}).
On \ds{Finance}, the SHT-based agent has higher cost \hide{}\keepc{than} vanilla-grep, because \hide{}\keepc{section-based} retrieval (\code{read\_section}) \keepc{often}  returns more content than \hide{}\keepc{1k-character chunks} (\code{grep})\hide{}\keepc{, especially in}\hide{} long documents, while also providing \hide{}\keepc{8\% higher accuracy}.
}
\hide{}
\blank{}

\begin{table}[t]
\centering
\caption{Average accuracy and total cost (USD, inside parenthesis) of four strategies by dataset. The last column reports stratified average accuracy ({\bf Avg.}) and total cost ({\bf Tot.}) across four datasets. \best{Green}: highest accuracy (lowest cost) per dataset. \blank{}}
\label{tab:eval_doc_qa_sht_usefulness_acc_cost}
\footnotesize
\setlength{\tabcolsep}{1.8pt}
\renewcommand{\arraystretch}{0.92}
\begin{tabular}{@{}
l@{\hspace{0.6pt}}
>{\raggedleft\arraybackslash}p{0.06\linewidth}@{\hspace{-0.9pt}}
>{\raggedleft\arraybackslash}p{0.06\linewidth}@{\hspace{0pt}}
>{\raggedleft\arraybackslash}p{0.09\linewidth}@{\hspace{-7pt}}
>{\raggedleft\arraybackslash}p{0.09\linewidth}@{\hspace{0pt}}
>{\raggedleft\arraybackslash}p{0.06\linewidth}@{\hspace{-2pt}}
>{\raggedleft\arraybackslash}p{0.09\linewidth}@{\hspace{0pt}}
>{\raggedleft\arraybackslash}p{0.07\linewidth}@{\hspace{-1pt}}
>{\raggedleft\arraybackslash}p{0.07\linewidth}@{\hspace{0pt}}
>{\raggedleft\arraybackslash}p{0.1\linewidth}@{\hspace{-5pt}}
>{\raggedleft\arraybackslash}p{0.1\linewidth}
@{}}
\toprule
{\bf Strategy}
& \multicolumn{2}{r}{\bf \ds{\footnotesize Civic}}
& \multicolumn{2}{r}{\bf \ds{\footnotesize Contracts}}
& \multicolumn{2}{r}{{\bf \ds{\footnotesize Finance}}}
& \multicolumn{2}{r}{{\bf \ds{\footnotesize Papers}}}
& \multicolumn{2}{r}{{\bf Avg. (Tot.)}} \\
\midrule

{Vanilla-in-context}
& 0.59 & (\hide{}\keep{3.1}) & 0.71 & (\hide{}\keep{6.7}) & 0.64 & (\hide{}\keep{112.5}) & \best{0.77} & (\hide{}\keep{32.6}) & 0.67 & (\hide{}\keep{154.9}) \\

Vanilla-grep agent
& 0.57 & (\hide{}\keep{3.4}) & 0.45 & (\hide{}\keep{\best{4.9}}) & 0.63 & (\hide{}\keep{\best{2.3}}) & 0.53 & (\hide{}\keep{11.9}) & 0.55 & (\hide{}\keep{22.5}) \\

{SHT-aug-context}
& \best{0.78} & (\hide{}\keep{4.1}) & 0.70 & (\hide{}\keep{9.0}) & \best{0.71} & (\hide{}\keep{118.2}) & \best{0.77} & (\hide{}\keep{34.1}) & \best{0.74} & (\hide{}\keep{165.3}) \\

SHT-based agent
& 0.72 & (\hide{}\keep{\best{2.1}}) & \best{0.75} & (\hide{}\keep{5.1}) & \best{0.71} & (\hide{}\keep{3.9}) & 0.76 & (\hide{}\keep{\best{6.3}}) & \best{0.74} & (\hide{}\keep{\best{17.4}}) \\

\bottomrule
\end{tabular}
\end{table}
\vspace{2mm}

\keep{\subsection{Agentic Retrieval Ablation}
\label{subsec:retrieval-ablation}
The SHT-based agent outperforms vanilla-grep, benefiting from both structure and a different retrieval tool (i.e., \code{read\_section}).
To isolate the two, we evaluate {\em SHT-grep} that uses \ttt{grep} with SHTs and two embedding-based agents, {\em vanilla-embed} and {\em SHT-embed}.
}

\topic{\keep{Setup}}
\keep{The SHT-grep and SHT-embed agents
have true SHT provided in the user prompt (as in SHT-based agent). For both embedding baselines, we split each document into 1,000-character chunks (as in vanilla-grep) and
embed them with the text-embedding-3-small model; the retrieval tool
takes a natural-language query and returns chunks by descending
cosine similarity, one-by-one.}

\topic{\keep{Results}}
\keep{\autoref{tab:retrieval_ablation} reports accuracy and cost of all agents.
SHTs improve accuracy for both retrieval tools: SHT-grep (SHT-embed) outperforms vanilla-grep (vanilla-embed) by 2--25\% (6--35\%).
The SHT-based agent is the best-performing one---3\%--23\% more accurate and 1.3--2$\times$ cheaper than all agentic baselines.
}

\begin{table}[t]
\centering
\caption{\keep{Average accuracy and total cost in parenthesis of all agents: {baselines (vanilla-grep from Tab.~\ref{tab:eval_doc_qa_sht_usefulness_acc_cost}, and the 3 other approaches from Sec.~\ref{subsec:retrieval-ablation})} and SHT-based (Tab.~\ref{tab:eval_doc_qa_sht_usefulness_acc_cost}). \best{Green}: highest acc. (lowest cost).}
}
\label{tab:retrieval_ablation}
\scriptsize
\setlength{\tabcolsep}{3pt}
\renewcommand{\arraystretch}{0.92}
\keep{%
\begin{tabular}{@{}l r@{\,}r r@{\,}r r@{\,}r r@{\,}r r@{\,}r@{}}
\toprule
{\bf Agent}
& \multicolumn{2}{r}{\bf \ds{\scriptsize Civic}}
& \multicolumn{2}{r}{\bf \ds{\scriptsize Contracts}}
& \multicolumn{2}{r}{{\bf \ds{\scriptsize Finance}}}
& \multicolumn{2}{r}{{\bf \ds{\scriptsize Papers}}}
& \multicolumn{2}{r}{{\bf Avg. Acc. (Tot. Cost)}} \\
\midrule
{Vanilla-embed}
& 0.38 & ({4.5}) & 0.39 & ({8.5}) & 0.64 & ({2.6}) & 0.64 & ({9.1}) & 0.51 & ({24.7}) \\
Vanilla-grep
& 0.57 & ({3.4}) & 0.45 & ({\best{4.9}}) & 0.63 & (\best{2.3}) & 0.53 & ({11.9}) & 0.55 & ({22.5}) \\
{SHT-grep}
& 0.62 & (5.5) & 0.70 & (12.5) & {0.65} & (5.6) & 0.61 & (11.3) & 0.65 & (34.8) \\
{SHT-embed}
& \best{0.75} & (4.9) & {0.74} & (15.2) & 0.64 & (5.8) & {0.69} & (9.2) & {0.71} & (35.0) \\
SHT-based
& 0.72 & (\hide{}{\best{2.1}}) & \best{0.75} & (\hide{}{{5.1}}) & \best{0.71} & (\hide{}{3.9}) & \best{0.76} & (\hide{}\keep{\best{6.3}}) & \best{0.74} & (\hide{}\keep{\best{17.4}}) \\
\bottomrule
\end{tabular}}
\end{table}

\subsection{SHT Ablation}
\label{sec:eval_qa_sht_ablation}

\begin{figure}
    \centering
    \begin{minipage}{0.52\linewidth}
        \includegraphics[width=\linewidth]{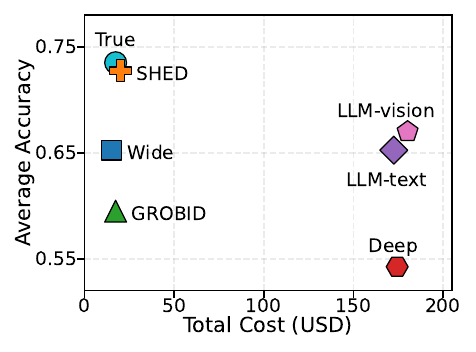}
    \end{minipage}%
    \hfill
    \begin{minipage}{0.44\linewidth}
        \caption{\keepc{Total cost (USD) vs. average accuracy of SHT-based agents using different SHTs across datasets.}}
        \label{fig:sht_ablation_acc_cost}
    \end{minipage}
    \vspace{-2mm}
\end{figure}

\begin{table}[t]
\centering
\caption{\keepc{Average accuracy and total cost (USD, inside parenthesis) of SHT-based agents using different SHTs by dataset. The last column reports stratified average accuracy ({\bf Avg.}) and total cost ({\bf Tot.}) across four datasets. \best{Green}: highest accuracy (lowest cost) per dataset.}}
\vspace{1mm}
\label{tab:sht_ablation_acc_cost}
\footnotesize
\setlength{\tabcolsep}{0.2pt}
\renewcommand{\arraystretch}{0.94}
\begin{tabular}{@{}
l@{\hspace{0.6pt}}
>{\raggedleft\arraybackslash}p{0.06\linewidth}@{\hspace{0pt}}
>{\raggedleft\arraybackslash}p{0.06\linewidth}@{\hspace{0pt}}
>{\raggedleft\arraybackslash}p{0.09\linewidth}@{\hspace{-5pt}}
>{\raggedleft\arraybackslash}p{0.09\linewidth}@{\hspace{0pt}}
>{\raggedleft\arraybackslash}p{0.08\linewidth}@{\hspace{-2pt}}
>{\raggedleft\arraybackslash}p{0.09\linewidth}@{\hspace{0pt}}
>{\raggedleft\arraybackslash}p{0.07\linewidth}@{\hspace{-0.5pt}}
>{\raggedleft\arraybackslash}p{0.07\linewidth}@{\hspace{0pt}}
>{\raggedleft\arraybackslash}p{0.08\linewidth}@{\hspace{-4.8pt}}
>{\raggedleft\arraybackslash}p{0.1\linewidth}
@{}}
\toprule
& \multicolumn{2}{r}{\bf \ds{\footnotesize Civic}}
& \multicolumn{2}{r}{\bf \ds{\footnotesize Contracts}}
& \multicolumn{2}{r}{{\bf \ds{\footnotesize Finance}}}
& \multicolumn{2}{r}{{\bf \ds{\footnotesize Papers}}}
& \multicolumn{2}{r}{{\bf Avg. (Tot.)}} \\
\midrule

Deep
& 0.59 & (\hide{}\keep{12.2}) & 0.63 & (\hide{}\keep{27.1}) & 0.20 & (\hide{}\keep{95.5}) & 0.75 & (\hide{}\keep{39.7}) & 0.54 & (\hide{}\keep{174.5}) \\

Wide
& 0.57 & (\hide{}\keep{2.0}) & 0.61 & (\hide{}\keep{\best{4.9}}) & 0.67 & (\hide{}\keep{\best{3.0}}) & \best{0.76} & (\hide{}\keep{\best{5.3}}) & 0.65 & (\hide{}\keep{\best{15.1}}) \\

GROBID
& 0.31 & (\hide{}\keep{\best{1.3}}) & 0.70 & (\hide{}\keep{5.8}) & 0.63 & (\hide{}\keep{5.0}) & 0.74 & (\hide{}\keep{5.3}) & 0.60 & (\hide{}\keep{17.4}) \\

LLM-text
& 0.68 & (\hide{}\keep{7.9}) & 0.60 & (\hide{}\keep{15.9}) & 0.59 & (\hide{}\keep{108.4}) & 0.74 & (\hide{}\keep{40.4}) & 0.65 & (\hide{}\keep{172.6}) \\

LLM-vision
& \best{0.75} & (\hide{}\keep{14.0}) & 0.50 & (\hide{}\keep{25.0}) & 0.69 & (\hide{}\keep{93.3}) & 0.74 & (\hide{}\keep{47.9}) & 0.67 & (\hide{}\keep{180.3}) \\

\sys
& 0.67 & (\hide{}\keep{2.2}) & \best{0.74} & (\hide{}\keep{6.4}) & \best{0.75} & (\hide{}\keep{5.4}) & 0.75 & (\hide{}\keep{6.2}) & \best{0.73} & (\hide{}\keep{20.1}) \\

\bottomrule
\end{tabular}
\end{table}

So far, we have shown that SHTs improve QA performance.\hide{}\blank{}
We now evaluate different SHT inference approaches (\Cref{sec:setup_sht_eval}) on document QA. We focus on the SHT-based agent for this comparison, as it achieves the best overall performance (accuracy and cost) among \hide{}{all} strategies.
Overall, we find that \sys achieves the highest average accuracy and lowest total cost among all SHT inference approaches, performing comparably to using true SHTs (\autoref{fig:sht_ablation_acc_cost}). Per-dataset results are in Table~\ref{tab:sht_ablation_acc_cost}.

\topic{Accuracy}
\sys outperforms {Deep} by 19\%, {GROBID} by 13\%, {Wide} and LLM-text by 8\%, and LLM-vision by 6\% in average accuracy.
\hide{}
Deep performs poorly  due to low compactness \keepc{causing context rot}, especially on \ds{Finance},
where the agent frequently exceeds the model context limit.
Wide achieves low accuracy due to poor robustness:
when the agent retrieves a section,
it only retrieves its prefix, and
may miss the content of
subsections appended as siblings
of that section in the tree\hide{}.
GROBID performs worse on out-of-domain datasets
than on \ds{Papers}, indicating limited generalizability.\blank{}\hide{}
LLM-based approaches lack guarantees on robustness or compactness, especially on \ds{Contracts}
with dense layout
and on large-scale \ds{Finance} documents.
All approaches have similar accuracy on \ds{Papers}, as it is trivial to infer structure only from section headers based on domain knowledge carried by the model.

\topic{Cost}
\sys has cost comparable to Wide and GROBID:
Wide is the cheapest due to strong compactness,
as each header contains only a small text span
(ending at the next header),
\keepc{\hide{}}\keepc{at the cost of lower robustness, leading to 8\% lower accuracy than \sys};\hide{}
\keepc{GROBID is cheap because its SHTs contain inaccurate headers (low F-1 scores in \autoref{tab:header_identification}\hide{}),  resulting in 13\% lower accuracy than \sys.}
\sys achieves {\em \keepc{both} low cost \keepc{and high accuracy}} through \hide{}\keepc{effective header identification and} \keepc{strong robustness-}compactness \keepc{\hide{}}\keepc{guarantees}.\blank{}
Deep has poor compactness, with each header’s text span
extending to the end of the document,
resulting in \hide{}\keep{$4\times$--$18\times$} higher cost than \sys \keepc{and up to 55\% accuracy drop across datasets due to context rot or \hide{}\keepc{exceeding the context limit}}.\blank{}
For LLM-text and LLM-vision, SHT inference dominates across datasets (\hide{}\keep{67\%--96\%} for LLM-text and 80\%--95\% for LLM-vision)\keepc{, making them \hide{}\keep{both about $9\times$} more expensive than \sys, while still yielding 8\% and 6\% lower average accuracy due to lack of robustness and compactness guarantees}.
\keep{\section{Scalability Analysis}
\label{sec:scalability}
To show \sys's scalability, \keep{we estimate the SHT inference latency of \sys (\Cref{subsec:sht_infer_latency}),} and repeat \hide{}{SHT inference and QA evaluations from} Sections~\ref{sec:sht_eval} \& \ref{sec:doc_qa}\hide{} on large-scale data,  varying composition sizes (\Cref{subsec:doc_composition})
and evaluating 3 new \hide{}datasets
\hide{}{of} individually complex documents\hide{}
(\Cref{subsec:scalability_complex_docs}).}

\keep{\subsection{SHT Inference Latency}
\label{subsec:sht_infer_latency}
We evaluate \sys's SHT inference latency on 80 sampled documents (20 per dataset; \autoref{tab:dataset}) and extrapolate to all 955 documents.
\sys is estimated to sequentially process all documents (\textasciitilde54K pages, 30M tokens) in 5 hours, showing good scalability; using parallelization would further reduce latency. {GROBID takes 4 hours but infers poor SHTs, LLM-based approaches take 2--5 hours but are costly, and SmolDocling-256M-preview is slow, taking \textasciitilde11.4 days.}
Details of the testbed, results\hide{}, and discussion are in Appendix~\ref{append:scalability_sht_infer}.
}

\keep{\subsection{Document Composition}
\label{subsec:doc_composition}
\hide{}{So far we} used composite documents, each concatenating a queried original with $k$ \hide{}{others, i.e., composition size $k+1$ (}$k$ varies by dataset\hide{}{; e.g.,} $k=3$ for \ds{Civic}; details in Appendix~\ref{append:datasets})\hide{}.
We now vary \hide{}{composition size from 1 (single originals) to }$2(k+1)$\hide{}{,}
\hide{}{mirroring growing} document size and structural complexity.
}
\keep{{Across all sizes, SHT-based agent with \sys is accurate
(0.71---0.73 on avgerage) at low cost (\$12--\$28 total, up to $28\times$ cheaper than baselines), on par with true SHTs
(0.71---0.74, \$11--\$25).
Its gains widen with size: it is up to  7\% to 19\% to 27\% more accurate than the baselines at sizes 1, $k+1$, and $2(k+1)$; its cost rises only 1.4$\times$ as size doubles.}
We provide detailed results and analysis in Appendix~\ref{append:doc_composition}.
}

\keep{\subsection{Complex Documents with Synthetic Q\&A}
\label{subsec:scalability_complex_docs}
To further show \sys's usefulness, we collected 3 new datasets with 119 complex documents and 142 synthetic Q\&A pairs. Specifically, we collected \ds{CFR}~\cite{cfr} (30 federal regulation codes), \ds{ETSI}~\cite{etsi} (49 radio-equipment standards), and \ds{FERC}~\cite{ferc} (40 environmental assessments), averaging 415K, 30K, and 49K tokens per document.
For each dataset, we designed realistic query templates, \hide{}instantiated 1--2 queries per document, \hide{}{and} manually labeled ground-truth answers.}
\keep{We evaluate SHT-based agent with \sys against Wide and two
non-SHT strategies, vanilla-in-context and vanilla-grep. Due to budget limits, we omit SHT-aug-context,
which is as accurate as SHT-based agent with \sys but far costlier (\Cref{sec:eval_qa_sht}), and other SHT baselines,
which are less accurate or costlier than Wide (\Cref{sec:eval_qa_sht_ablation}).
Overall, \sys has 7\%--18\% higher QA average accuracy at up to $15\times$ lower cost than baselines. Details can be found in Appendix~\ref{append:scalability_new_datasets}.}

\section{Related Work}
\label{sec:relatedwork}

We survey related work on document processing and analytics.

\topic{Systems for Document Analytics}
Recent systems target document processing \hide{}for analytical tasks.
Some systems focus on table extraction for predominantly tabular or structured documents~\cite{text2table, twix, evaporate},
which could complement our work.
Other systems provide declarative interfaces for general-purpose document processing~\cite{pz, lotus, docetl, caesura, thalamusDB, quest, agenticdata}
by treating documents as linear text sequences
and ignoring their inherent structure; our approach can be integrated to enhance their capabilities.
ZenDB~\cite{zendb} is the only document analytics
system that considers structure,
focusing on the setting where many
documents in a collection share the same template.
\hide{}
Our approach broadens its applicability with a new formulation of structure extraction and a general two-stage framework.

\topic{Document Structure Inference}
Significant prior work has considered document structure inference.
Some approaches target documents with explicit structures, such as HTML~\cite{structure-of-html}, which don't translate
to PDF and Word documents.
For the latter class of documents,
existing approaches fall into three classes.
{\em Rule-based approaches} rely on heuristics~\cite{pypdf, pdfplumber, pymupdf, pdfminer, pdfium},
which often fail on heterogeneous layouts
and provide no guarantees, being largely outperformed by LLM-based approaches.
Our approach instead provides \hide{} robustness guarantees sufficient for real-world applications.
{\em Learning-based approaches}~\cite{ml-sht-1, ml-sht-2, ml-sht-3, nougat} train models to infer structures but generalize poorly.
We choose GROBID~\cite{grobid}, trained primarily on scientific papers\hide{} {due to its wide adoption.}\hide{}
Our approach differs by leveraging learning-based approaches only for header identification, a simpler\hide{}
task~\cite{publaynet-paper, doclaynet-paper, docbank}, where models achieve high accuracy~\cite{vgt-paper, vgt-code, lightgbm}.
Finally, \emph{LLM-based approaches} can infer structures from plain text~\cite{structrag, bookrag} or page images~\cite{contextgem, llmaparse}.
As we see in our experiments,
these approaches are costly, hallucination-prone~\cite{llm-hallucination-1, llm-qa-7, lost-in-the-middle}, and yield poor robustness and compactness.

\topic{Agentic Document Processing}
Agents process documents across domains
such as coding~\cite{claude-code, codex}, finance~\cite{bigeard2025finance}, and scientific research~\cite{starace2025paperbench},
treating documents as bags of words by directly processing the full document~\cite{financebench-paper} or
via retrieval for relevant portions~\cite{sun2025docagent, agentic-rag-survey}.
They \hide{}ignore the hierarchical structure
\hide{}in documents\hide{} {that} is crucial for accuracy and cost,
\hide{}{as shown} by augmenting agents with \sys.
Related settings include agent skills~\cite{agent-skills} and knowledge management systems (e.g., LLM wikis~\cite{llm-wiki}), where agents operate over hierarchically organized document collections\hide{} rather than a single document. This hierarchical organization enables agents to progressively load relevant components, improving \hide{}{accuracy} and reducing cost~\cite{progressive-disclosure}.
For instance, agent skills require users to manually define and maintain task-specific hierarchies;
\sys can potentially be used here to recover hierarchical structures \hide{}{without manual specification}.

\blank{}

\vspace{-4mm}
\topic{Retrieval-Augmented Generation (RAG)}
\keepc{RAG retrieves the most relevant context for generation, compensating for LLM limitations on long documents.
Most techniques treat a document as a collection of chunks~\cite{rag, rag-2, agentic-rag-survey}, a recursive tree of summaries~\cite{raptor}, or a knowledge graph~\cite{graph-rag, hippo-rag},
ignoring
inherent hierarchical structure as in \sys. Our technique could be used by such approaches to improve RAG.
Other RAG techniques leverage hierarchical structures~\cite{bookrag, page-index, structrag}
but infer them \hide{}\keepc{using LLMs at significantly higher cost} and lacking guarantees.}

\hidec{}

\section{Conclusion}
We introduced \sys,
a two-stage workflow for SHT inference
that guarantees robustness \hide{}\keepc{with good compactness}.
We proposed local- or global-first SD inference,
and generalized them into an infinite family of approaches.
We further theoretically characterized
a taxonomy of document classes
for which \sys generates robust SHTs for each approach.
Empirically, \sys outperforms non-LLM
and expensive LLM-based
SHT inference baselines by 13\%--68\% and 9\%--15\% in F-1 scores, respectively, \hide{}at no cost.
Applying \sys to agentic document QA improves accuracy by \hide{}\keep{3\%--23\%} over baselines, reducing cost by up to \hide{}\keep{$10\times$}. \hide{}

\clearpage
\bibliographystyle{ACM-Reference-Format}
\bibliography{main}

\clearpage
\appendix
\section{Technical Details}

\subsection{Proof for \autoref{thm:ts-correct}}
\label{append:proof-ts-correct}

\begin{proof}
$(\Leftarrow)$
Let $T'$ be the true SHT. Any parent-child pair in $T'$ remains an ancestor-descendant pair in SHT $T$.
Suppose $v_{k}$ and $v_{k'}$ are the next non-descendant\keepc{s} of $v_i$ in $T$ and $T'$, respectively.
For $v_{k'-1}$, either $v_{k'-1} = v_i$, or it is the last descendant of $v_i$ in $T'$.
For the former case, $k'-1 = i < k$.
For the latter case, since $v_{i}$ is \hide{}\keepc{an ancestor} of $v_{k'-1}$ in $T'$, it is also an ancestor of $v_{k'-1}$ in $T$. Therefore, $v_{k'-1}$ is a descendant of $v_i$ in $T$, indicating $k'-1 < k$.
Thus $ts'(v_i) = [p_i, \ldots, p_{k'-1}] \subseteq [p_i, \ldots, p_{k-1}] = ts(v_i)$. $T$ is robust.

$(\Rightarrow)$
For each parent-child pair $(v_i, v_j)$ in the true SHT $T'$, $p_j\in ts'(v_i)$. Since $T$ is robust, we have $ts'(v_i)\subseteq ts(v_i)$. Therefore, $p_j\in ts(v_i)$, indicating that $v_i$ is an ancestor of $v_j$ in $T$.
\end{proof}

\subsection{Proof for \autoref{lem:sht}}
\label{append:proof_lem_sht}
\begin{proof}
    When inserting $v_i$, all nodes whose text spans cover the header phrase of $v_i$ lie on the {\em rightmost path} of $T$ (e.g., in \autoref{fig:shtgen}, all nodes highlighted in orange are on the rightmost path at each insertion). $v_i$ is then inserted as the rightmost child of a node on this rightmost path, and therefore appears {\em last} in the pre-order traversal of $T$. By induction, after every iteration of Algorithm~\ref{alg:shtgen}, the pre-order traversal of $T$ remains consistent with the increasing order of the node indices.
\end{proof}

\subsection{Proof for \autoref{theo:correctness-general}}
\label{append:proof_correctness_general}
\begin{proof}
    $(\Rightarrow)$.
    If $T$ is robust, $v_j$ is an ancestor of $v_i$ in $T$. Algorithm~\ref{alg:shtgen} ensures that for any node, the semantic depth of its cluster is larger than that of the cluster containing its parent. Applying this property along the ancestor chain from $v_i$ up to $v_j$ in $T$ gives $sd(v_j) < sd(v_i).$

    ($\Leftarrow$).
    We prove by induction that each time Algorithm~\ref{alg:shtgen} inserts a node $v_i$, it places $v_i$ as a descendant of its true parent $v_j$. Consequently, the tree $T$ remains robust throughout the execution of the algorithm.
    The key step is to show that, at the moment $v_i$ is inserted, {\em its true parent $v_j$ is a candidate parent of $v_i$ in Algorithm~\ref{alg:shtgen}}.
    Since $j < i$ and $sd(v_j) < sd(v_i)$, it suffices to show that $ts(v_j)$ covers the header phrase of $v_i$ \hide{}when $v_i$ is inserted, which is equivalent to showing that $v_j$ lies on the rightmost path of $T$.
    Indeed, all nodes inserted into $T$ after $v_j$ and before $v_i$ are placed as descendants of $v_j$ in $T$: they are descendants of $v_j$ in the true SHT, and remain descendants of $v_j$ in $T$ by our induction hypothesis that $T$ is robust before inserting $v_i$.
    Therefore, when $v_i$ is inserted, $v_j$ remains on the rightmost path of $T$.
\end{proof}

\subsection{Depth-Aligned Documents}
\label{append:depth_aligned}

In this class of documents, the indexing node of each cluster must be the {\em leftmost child} of its parent in the true SHT (i.e., only case \circled{1} is allowed to occur in \autoref{fig:intuition_infer_sd_1}, while case \circled{2} is disallowed).
In other words, when a new header formatting is introduced, it must be the first such header in a parent section.
Under this condition, the cluster corresponding to the child indexing node is assigned a semantic depth exactly one greater than that of its true parent's cluster by Algorithm~\ref{alg:build_cg_1}.

These documents further require clusters to be {\em level-unique}, meaning that all nodes within the same cluster are at the same level in the true SHT.
As a result, the semantic depth assigned by local-first matches the true depth of every node (i.e., the number of edges from the root to the node) in the true SHT, guaranteeing SD conformance and thus robustness. We refer to such documents as {\em depth-aligned documents} (Definition~\autoref{defn:depth_aligned}).

For example, \autoref{fig:comp_infer_sd_1_true} is the true SHT of a depth-aligned document.
Therefore, the SHT inferred by local-first is robust, as shown in \autoref{fig:comp_infer_sd_1_1}.

We prove \autoref{thm:depth_aligned} below:
\begin{proof}
    At the moment before inserting node $v_i$ into $T$, assume that all nodes currently in $T$ have semantic depths equal to their true depths in $T'$.
    If $v_i$ is {\em not} the indexing node of its cluster, then by condition 2) of depth alignment, $sd(v_i)$ equals the true depth of the indexing node of $C(v_i)$, which in turn equals the true depth of $v_i$ in $T'$.
    If $v_i$ is the indexing node of its cluster, then by condition 1) of depth alignment, $sd(v_i) = sd(v_j) + 1$, where $v_j$ is the true parent of $v_i$ in $T'$. By our assumption that $sd(v_j)$ equals the true depth of $v_j$ in $T'$, it follows that $sd(v_i)$ also equals its true depth in $T'$.
\end{proof}

Although depth alignment ensures that local-first produces robust SHTs, global-first may fail on such documents.
\autoref{fig:comp_infer_sd_1_2} shows a non-robust SHT produced by global-first for a depth-aligned document whose true SHT is shown in \autoref{fig:comp_infer_sd_1_true}.

\subsection{Proof for \autoref{thm:well-loosely-relationship}}
\label{append:proof_well_loosely_relationship}

\begin{proof}
    If $D$ is well-formatted, then $\forall~v_i, v_j\in V$ where $v_i$ is the parent of $v_j$ in $T'$, $v_l$'s parent $v_s$ in $T'$ satisfies $C(v_s) = C(v_i)$, where $l = ind(C(v_j))$. Choosing $v_k = v_s$ satisfies all conditions in Definition~\ref{def:loosely-format}, and thus $D$ is loosely-formatted.
    The converse does not hold because well-formattedness does not allow parent-child relationships to be inferred from ancestor-descendant occurrences.
\end{proof}

\subsection{Proof for \autoref{thm:loosely-robustness}}
\label{append:proof_loosely_robustness}
\begin{proof}
    Let the approach from $\mathcal{K}$ be $\code{Infer-SD}_K$ with an arbitrary parameter mapping $K$.
    \keepc{Assume }$D$ is loosely-formatted.
    We prove this theorem by contradiction: suppose that the inferred SHT $T$ is not robust.
    By Theorem~\ref{theo:correctness-general}, there must exist a parent-child pair $(v_i, v_j)$ in the true SHT $T'$ that violates SD conformance, i.e., $sd(v_i) \geq sd(v_j).$
    Among all such violating pairs, choose $(v_i, v_j)$ with the smallest child index $j$.
    $v_{k_1}, v_{k_2}, sd_1$ and $sd_2$ are inferred by Algorithm~\ref{alg:infer-sd-family} while scanning the cluster $C(v_j)$ of $v_j$.
    We first establish two inequalities:\\
    1) $sd(v_i) \leq sd_1$.
    If $v_{k_1} = v_i$, the claim is immediate.
    Otherwise, $v_{k_1}$ is a true descendant of $v_i$ in $T'$.
    By the minimality of $j$, SD conformance holds for $(v_i, v_{k_1})$, implying $sd(v_i) < sd(v_{k_1}) = sd_1$.\\
    2) $sd(v_i) \leq sd_2$.
    This follows directly from $i < j$ and the definition of $sd_2$.\\
    We now consider $v_j$.\\
    \emph{Case (a): $v_j$ is the indexing node of its cluster.}
    By Algorithm~\ref{alg:infer-sd-family},
    $sd(v_j) = K(C(v_j)) \cdot sd_1 + \bigl(1 - K(C(v_j))\bigr) \cdot sd_2 + 1$
    for some $K(C(v_j)) \in [0,1]$.
    Since $sd(v_i) \le sd_1$ and $sd(v_i) \le sd_2$,
    it follows that $sd(v_i) < sd(v_j)$.\\
    \emph{Case (b): $v_j$ is not the indexing node of its cluster.}
    Let $v_l$ be the indexing node of $C(v_j)$, with $l = ind(C(v_j)) < j$.
    Because $D$ is loosely-formatted, there exists a node $v_k$ such that
    $v_k$ is an ancestor of $v_l$ in $T'$ and $C(v_k) = C(v_i)$.
    By the minimality of $j$, SD conformance holds for $(v_k, v_l)$, implying
    $sd(v_k) < sd(v_l)$.
    Since $sd(v_k) = sd(v_i)$ and $sd(v_l) = sd(v_j)$, we again obtain
    $sd(v_i) < sd(v_j)$.\\
    In both cases, we \hide{}\keepc{obtain} $sd(v_i) < sd(v_j)$, a contradiction to our initial assumption that $sd(v_i) \geq sd(v_j)$.
    Therefore, the assumption that $T$ is not robust must be false, and $T$ is robust.
\end{proof}

\section{SD Inference Family}
\label{append:infer_sd_family}

We present the pseudo-code for our infinite family of SHT inference approaches, extending local and global-first, in Algorithm~\ref{alg:infer-sd-family}. $\code{Infer-SD}_K$ guarantees SHT robustness whenever the input document satisfies SD conformance, i.e., in the true SHT of $D$, the inferred semantic depth of each parent's cluster is smaller than that of its child's cluster (Theorem~\ref{theo:correctness-general}).

\begin{algorithm}[t]
    \caption{$\code{Infer-SD}_{K}(V, \mathcal{C})$}
    \begin{algorithmic}[1]
        \State $sd(C_1)\leftarrow 0$
        \For{$C\in \mathcal{C}\backslash \{C_1\}$ in ascending order of $ind(C)$}
            \State $v_{k_1}\leftarrow \arg\max_{v_i\in V, i < ind(C)}: i$
            \State $sd_1 \leftarrow sd(C(v_{k_1}))$
            \Comment{inferred by the local-first approach}
            \State $v_{k_2}\leftarrow \arg\max_{v_i\in V, i < ind(C)}: sd(C(v_i))$
            \State $sd_2 \leftarrow sd(C(v_{k_2}))$
            \Comment{inferred by the global-first approach}
            \State $sd(C)\leftarrow K(C) \cdot sd_1 + (1 - K(C)) \cdot sd_2 + 1$
        \EndFor
        \State \Return $sd$
    \end{algorithmic}
    \label{alg:infer-sd-family}
\end{algorithm}
\section{Relationship to ZenDB}
\label{append:relation_zendb}

Lin et al.~\cite{zendb} propose an algorithm for SHT construction (Algorithm~1 in \cite{zendb}).
In \Cref{append:equiv-zendb-sht},
we prove that this algorithm is equivalent to our two-stage workflow with global-first SD inference: for each node $v_i$, both approaches select as \hide{}\keepc{$v_i$'s} parent the node $v_j$ with the largest index $j < i$ such that $ind(C(v_j)) < ind(C(v_i))$.
Moreover, in \Cref{subsec:hierarchy-robustness},
we \hide{}\keepc{state} one requirement of well-formattedness as consistency between each parent-child pair and the first occurrences of their corresponding visual patterns.
We show in \Cref{append:equiv-zendb-well-format} that
this requirement is equivalent to Lin et al.’s original well-formattedness definition (Definition~1 in \cite{zendb}),
which requires consistency among the ancestor visual-pattern lists of all nodes within the same cluster.
Lin et al.
\cite{zendb} \hide{}\keepc{do} not go beyond well-formattedness and \hide{}\keepc{do} not consider properties such as robustness or compactness.

\subsection{SHT Construction}
\label{append:equiv-zendb-sht}
ZenDB proposes Algorithm~1 in~\cite{zendb} to construct SHTs.
We show that this algorithm is equivalent to our two-stage workflow, with global-first SD inference (Algorithm~\ref{alg:approxcg-2}) followed by SHT assembly (Algorithm~\ref{alg:shtgen}).
\begin{theorem}
    Our two-stage workflow with global-first SD inference is equivalent to Algorithm~1 of ZenDB.
\end{theorem}
\begin{proof}
\vspace{-10pt}
For simplicity, we use $ind(v_i)$ to denote $ind(C(v_i))$.

{\em 1) Removing the constraint on text span in Line~8 of Algorithm~1 of ZenDB.}
In Algorithm~1 of ZenDB,
when inserting a node $v_i$ (Line~6),
the algorithm chooses the node $v_j$ with the largest index $j <i$ such that $ind(v_j) < ind(v_i)$ as the parent of $v_i$,
where $ind(v_j) < ind(v_i)$ is implied by Line~7, and the maximality of $j$ is implied by Line~8.
At this point in Algorithm~1, $ts(v_j) = [p_j, \ldots, p_{s-1}]$, where $p_s$ corresponds to the next non-descendant $v_s$ of $v_j$ in $T$.
$v_s$ is thus inserted before $v_i$, implying $ind(v_s) < ind(v_i)$.
The maximality of $j$ further implies $s > i$.
Therefore, $p_i \in ts(v_j)$ holds automatically,
and the constraint $ind(v_i)\in ts(v_j)$ in Line~8 of Algorithm~1 can be removed without changing the algorithm’s behavior.

After removing this constraint on text span, the nested loops (Line 4-10) in Algorithm~1 of ZenDB can be reduced to the following single loop, where we additionally replace \keepc{the }constraint $v_j\in V - V_C$ (Line~7) with $v_j\in T$:
\begin{algorithm}[h]
    \begin{algorithmic}[1]
        \For{$v_i\in V$ in ascending order of $i$}
            \State $v_k \leftarrow \arg\max_{v_j\in T, ind(v_j) < ind(v_i)}: j$
            \State Add node $v_i$ to $T$ as the rightmost child of $v_k$.
        \EndFor
    \end{algorithmic}
\end{algorithm}

{\em 2) Equivalence to our two-stage workflow with global-first SD inference.}
Algorithm~\ref{alg:shtgen} (Line 3) selects $v_k$ as the parent of $v_i$, where $k$ is the largest index $j$ such that
$v_j\in V$, $p_i\in ts(v_j)$, and $sd(C(v_j)) < sd(C(v_i))$.
Under global-first SD inference, the ordering of semantic depths is consistent with the ordering of cluster indices.
Therefore, $sd(C(v_j)) < sd(C(v_i))$ is equivalent to $ind(v_j) < ind(v_i)$.
As argued in {\em 1)}, the condition $p_i \in ts(v_j)$ is also redundant for Algorithm~\ref{alg:shtgen} under global-first.
Line~3 of Algorithm~\ref{alg:shtgen} under global-first can thus be reduced to the following line:
\begin{algorithm}[h]
    \begin{algorithmic}[1]
        \State $v_k \leftarrow \arg\max_{v_j\in T, ind(v_j) < ind(v_i)}: j$
    \end{algorithmic}
\end{algorithm}

This line is identical to Line~2 of the single loop in {\em 1)}.
Therefore, our two-stage workflow under global-first is equivalent to Algorithm~1 of ZenDB.
\end{proof}

\subsection{Well-Formattedness}
\label{append:equiv-zendb-well-format}
ZenDB defines a document $D$, given its true SHT $T'$ and clusters $\mathcal{C}$, to be well-formatted via Definition~1 in \cite{zendb} if it satisfies the following two conditions:\\
1) Any two nodes that are siblings (i.e., have the same parent) in $T'$ must belong to the same cluster.\\
2) Any two nodes belonging to the same cluster must share identical {\em visual prefixes}, denoted $vispre$, defined as the sequences of clusters along the paths from the root to the nodes in $T'$.

In \Cref{subsec:hierarchy-robustness}, we rephrase the second condition as a consistency requirement between any parent-child pair and the first occurrence\keepc{s} of their corresponding visual patterns.
We now show that this rephrasing (i.e., condition (2) in Theorem~\ref{thm:well-format-equiv}) \hide{}\keepc{does not} change ZenDB's original definition (i.e., condition (1) in Theorem~\ref{thm:well-format-equiv}):
\begin{theorem}~\label{thm:well-format-equiv}
    Consider a document $D$ with nodes $V$, clusters $\mathcal{C}$, and true SHT $T'$.
    The following two conditions are equivalent:\\
    (1) $\forall~v_i, v_j\in V$, if $C(v_i) = C(v_j)$, then $vispre(v_i) = vispre(v_j)$, where $vispre(v) = [C(v_1), \ldots, C(v)]$ is the sequence of clusters along the path from the root $v_1$ to $v$ in $T'$.\\
    (2) $\forall~v_i, v_j\in V$, if $v_i$ is the parent of $v_j$ in $T'$, then there exists a parent-child pair $(v_k, v_l)$ in $T'$ such that $l = ind(C(v_j))$ and $C(v_k) = C(v_i)$.
\end{theorem}
\begin{proof} [$(1) \Rightarrow (2)$.]
    If (1) holds for $D$, $vispre(v_j) = vispre(v_l)$ implies that the parents of $v_j$ and $v_l$ in $T'$ belong to the same cluster, establishing (2).

    $(2) \Rightarrow (1)$.
    If (2) holds for $D$,
    then for $v_i$ and $v_j$ in the same cluster $C$,
    the true parent of $v_i$ belongs to the same cluster as the true parent of $v_{ind(C)}$, and the same holds for $v_j$.
    Thus, $v_i$, $v_j$, and $v_{ind(C)}$ all have parents in the same cluster.
    Applying the same argument recursively to their parents, we conclude that $v_i$ and $v_j$ share identical clusters at every level along their paths to the root. Hence, $vispre(v_i) = vispre(v_j)$.
\end{proof}

\section{Datasets}
\label{append:datasets}

\begin{table*}[t]
    \centering
    \caption{Datasets are collected from existing sources and cited accordingly. {\bf Src. \#Docs}: number of documents in the source. {\bf Src. \#Qs}: number of questions in the source. For \ds{Civic}, metadata are provided, allowing flexible question curation; we therefore report --. {\bf \#Const. Docs ($k$)}: number of source documents constituting each composite document (i.e., $k+1$) used in our evaluation (\Cref{sec:datasets}). {\bf Src. Query Example}: an example question on the original source (single document). {\bf Modified Query Example}: its modified version for our agentic document QA evaluation. \diff{Red}: modified text.}
    \label{tab:append_datasets}
    \footnotesize
    \setlength{\tabcolsep}{2.6pt}
    \renewcommand{\arraystretch}{1.2}
    \begin{tabular}{@{}
    >{\raggedright\arraybackslash}p{0.09\linewidth}
    >{\raggedright\arraybackslash}p{0.035\linewidth}
    >{\raggedright\arraybackslash}p{0.035\linewidth}
    >{\raggedright\arraybackslash}p{0.06\linewidth}
    >{\raggedright\arraybackslash}p{0.3\linewidth}
    >{\raggedright\arraybackslash}p{0.4\linewidth}
    @{}}
    \toprule
    {\bf Datasets} & {\bf Src. \#Docs} & {\bf Src. \#Qs} & {\bf \#Const. Docs} & {\bf Src. Query Example} & {\bf Modified Query Example}\\
    \midrule
    \ds{Civic}~\cite{zendb, civic} & 19 & -- & 4 ($k=3$) & {\em Return a list of project names for all projects whose status matches the status of project `Westward Beach Road Shoulder Repairs (CalOES Project)'} & {\em \diff{According to the report for the meeting on January 26, 2022:} Return a list of project names for all projects whose status matches the status of project `Westward Beach Road Shoulder Repairs (CalOES Project)'}\\
    \ds{Contracts}~\cite{contract-nli} & 73 & 1241 & 5 (\hide{}\keep{$k=4$}) & {\em Determine the relationship between contract `064-19 Non Disclosure Agreement 2019' and a hypothesis $\mathcal{H}$ (one of `Entailment', `Contradiction', or `NotMentioned')} & {\em \diff{Return a list of contract names whose relationship to $\mathcal{H}$ is the same as that of the contract `064-19 Non Disclosure Agreement 2019'.}} \\
    \ds{Finance}~\cite{financebench-paper} & 84 & 150 & 2 ($k=1$) & {\em What is the FY2018 capital expenditure amount (in USD millions) for 3M?} & Same as the source query.\\
    \ds{Papers}~\cite{qasper-paper} & 416 & 1451 & \hide{}\keep{3 ($k=2$)} & {\em How big is the ANTISCAM dataset?} & {\em \diff{According to the paper `End-to-End Trainable Non-Collaborative Dialog System':} How big is the ANTISCAM dataset?} \\
    \bottomrule
    \end{tabular}
\end{table*}

We \hide{}\keepc{compose} real-world documents and sample questions from existing datasets spanning diverse lengths, domains, and templates into our evaluation datasets (\Cref{sec:datasets}). We describe each dataset in detail (\autoref{tab:append_datasets}).

\topic{\ds{Civic}}
The source contains Malibu City Council agenda reports~\cite{civic}, each using a complex, nested hierarchy to describe civic projects and their attributes (e.g., status, updates, start and completion times). We use each report's meeting date as its name, placed on its prepended title page.

\topic{\ds{Contracts}}
The source contains non-disclosure agreements (NDAs) with hierarchical structures indicated by list items. We use the raw PDF file names as contract names, placed on title pages.

\topic{\ds{Finance}}
The source contains financial filings (10-Ks, 10-Qs, 8-Ks, and earnings reports) from U.S. publicly traded companies between 2015 and 2023. We use a hash string as each document's name and place it on the corresponding title page; thus, to answer a (source) question (which specifies the company) over a composite document, one must first identify the correct constituent document (with hashed name) by inspecting partial content before retrieving relevant text.
Moreover, since some (source) questions rely on tables, we \hide{}\keepc{mitigate} LLMs' difficulty in processing serialized tables by reformatting them into Markdown~\cite{tab-llm, tab-llm-blog}.
Specifically,
each table row is represented as a Markdown section marked with \#\#,
and each cell value $v$, together with its corresponding column header $C$, is represented as a list item ``- $C: v$'' under this section.

\topic{\ds{Papers}}
The source contains 416 natural language processing (NLP) research papers in its test split. We extract paper titles from metadata and place them on title pages.

\keep{\section{Additional Evaluations and Results}
{In this section, we provide detailed results and analysis of evaluations mentioned in the main paper. }
}

\keep{\subsection{True Header Ablation}
\label{append:true_header_abl}
In \Cref{sec:sht_eval_results}, we summarized the results of SHT inference with \sys using extracted headers $V$ versus true headers $V'$. Here, we provide detailed results and analysis.}

\keep{Results in \autoref{tab:sht_eval_results_true_header} show that with true headers $V'$, \sys's SHTs reach nearly perfect robustness (e.g., 96\%--99\% $M_{\downarrow}$ recall across all datasets), with the remaining 1\%--4\% arising from documents that violate SD conformance (Section~\ref{subsec:sd-robustnes}).
$V'$ improves almost all metrics across all datasets {over extracted headers $V$}, most notably on \ds{Contracts} (2\%--12\%) and \ds{Finance} (0\%--20\%), whose header identification accuracy is relatively low (\autoref{tab:header_identification}).
However, on average, the gap between $V'$ and $V$ is moderate (3\%--6\% across all metrics), so \sys with $V$ still works well (because the text span of a missing true header $v\in V'\backslash V$ is usually contained in that of the preceding header, typically a sibling or parent).
Therefore, we use $V$ for \sys throughout our evaluation, which performs comparably to $V'$ and reflects the realistic setting where $V'$ is unavailable.
\providecommand{\ppd}[1]{{\scriptsize(#1)}}
\begin{table}[t]
    \centering
    \caption{\keep{Robustness (Recall), compactness (Precision), and their F-1 of SHTs
    inferred by \sys from true headers. Parentheses indicate the percentage difference from \sys with
    extracted headers.}}
    \label{tab:sht_eval_results_true_header}
    \footnotesize
    \setlength{\tabcolsep}{1.8pt}
    \renewcommand{\arraystretch}{0.92}
    \keep{%
    \begin{tabular}{@{}l
      r@{\hspace{1pt}}l r@{\hspace{1pt}}l r@{\hspace{1pt}}l |
      r@{\hspace{1pt}}l r@{\hspace{1pt}}l r@{\hspace{1pt}}l@{}}
    \toprule
    \multirow{2}{*}{\bf Dataset}
    & \multicolumn{6}{c}{\bf Top-down: $M_{\downarrow}$}
    & \multicolumn{6}{c}{\bf Bottom-up: $M_{\uparrow}$} \\
    \cmidrule(lr){2-7} \cmidrule(lr){8-13}
    {} & \multicolumn{2}{c}{\bf Recall} & \multicolumn{2}{c}{\bf Precision} & \multicolumn{2}{c}{\bf F-1}
       & \multicolumn{2}{c}{\bf Recall} & \multicolumn{2}{c}{\bf Precision} & \multicolumn{2}{c}{\bf F-1} \\
    \midrule
    {\ds{\footnotesize Civic}}
      & 0.96 & \ppd{+6} & 0.96 & \ppd{+5}  & 0.95 & \ppd{+6}
      & 0.94 & \ppd{+2} & 0.93 & \ppd{+0}  & 0.93 & \ppd{+0}  \\
    {\ds{\footnotesize Contracts}}
      & 0.99 & \ppd{+2}  & 0.99 & \ppd{+11} & 0.98 & \ppd{+11}
      & 0.94 & \ppd{+12} & 0.97 & \ppd{+5}  & 0.94 & \ppd{+11} \\
    {\ds{\footnotesize Finance}}
      & 0.96 & \ppd{+2} & 0.95 & \ppd{+0}  & 0.92 & \ppd{+2}
      & 0.71 & \ppd{+2} & 0.91 & \ppd{+20} & 0.76 & \ppd{+11} \\
    {\ds{\footnotesize Papers}}
      & 0.96 & \ppd{+0} & 0.96 & \ppd{+1} & 0.95 & \ppd{+1}
      & 0.86 & \ppd{+4} & 0.94 & \ppd{+1} & 0.88 & \ppd{+4} \\
    \cmidrule(lr){1-13}
    {\bf Avg.}
      & 0.97 & \ppd{+3} & 0.96 & \ppd{+4} & 0.95 & \ppd{+5}
      & 0.86 & \ppd{+5} & 0.94 & \ppd{+6} & 0.88 & \ppd{+6} \\
    \bottomrule
    \end{tabular}}
\end{table}
}

\keep{\subsection{Compactness-Depth Correlation}
\label{append:correlation-compact-depth}}

\keep{To show that {depth is a reasonable proxy for compactness}, we evaluate the correlation between the two over all 955 documents (\autoref{tab:dataset} in \Cref{sec:datasets}), using SHTs inferred by \sys from both extracted and true headers (\Cref{tab:sht_eval_results} in \Cref{sec:sht_eval_results}). As shown in \autoref{fig:depth_compact_corr}, we group SHTs by their inferred depth relative to the true depth into three bins, balancing the number of instances and depth ratio range. The average {bottom-up} compactness with extracted (true) headers is 0.92 (0.95), 0.91 (0.94), and 0.83 (0.91) for bins $(0.5,1]$, $(1,1.2]$, and $(1.2,2.2]$, respectively, indicating that deeper SHTs tend to be less compact.%
}

\begin{figure}[t]
    \centering
    \begin{minipage}{0.58\linewidth}
        \includegraphics[width=\linewidth]{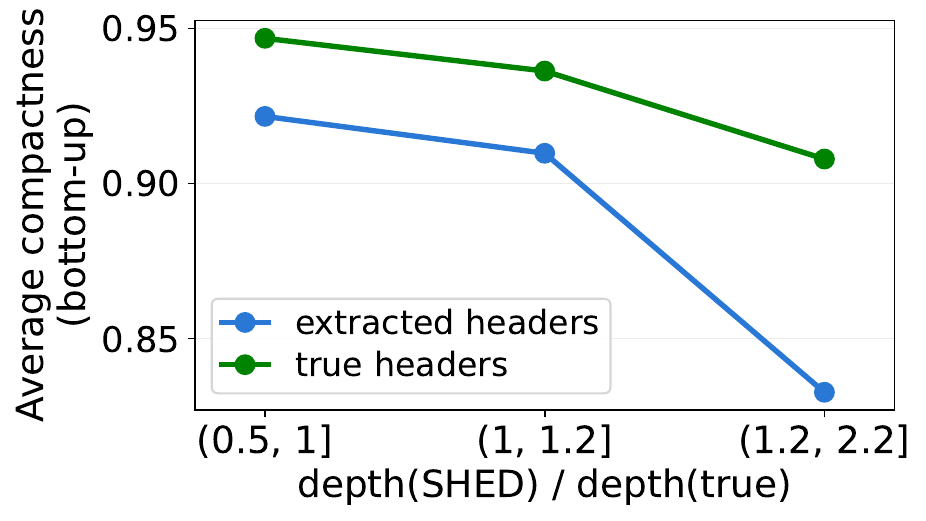}
    \end{minipage}%
    \hfill
    \begin{minipage}{0.38\linewidth}
    \caption{\keep{Compactness vs.\ depth relative to true SHT. A point shows average compactness of SHTs in a depth ratio bin, from extracted (blue) or true (green) headers.
    }}
    \label{fig:depth_compact_corr}
    \end{minipage}
\end{figure}

\keep{\subsection{Scalability Analysis: Document Composition}
\label{append:doc_composition}}

\keep{In \Cref{subsec:scalability_complex_docs}, we described results of \sys and baselines on composite documents of varying sizes, from 1 to $k$ to $2(k+1)$. Here, we provide detailed results and analysis. For single originals (size 1), we use the source queries rather than \hide{}{those }adapted for \hide{}{composites} (see \autoref{tab:append_datasets}).
}

\topic{\keep{Robustness and Compactness}}
\keep{We report the results ($M_{\downarrow}$/ $M_{\uparrow}$ recall, precision, and F-1) of composition sizes 1 and $2(k+1)$ in \autoref{tab:sht_eval_results_comp1} and \autoref{tab:sht_eval_results_comp2k} (the $k+1$ results are in \autoref{tab:sht_eval_results} in \Cref{sec:sht_eval}). The SHTs inferred by \sys stay robust and compact as composition size grows: recall, precision and F-1 change by at most 2\%. They also beat baselines by a widening margin: their $M_{\downarrow}$ ($M_{\uparrow}$) F-1 are 7\%--56\% (9\%--39\%), 9\%--68\% (10\%--49\%), and 13\%--71\% (15\%--55\%) higher than baselines at sizes 1, $k+1$, and $2(k+1)$.}

\topic{\keep{QA accuracy \& cost}}
\keep{We show the QA accuracy and total cost across composition sizes in \autoref{fig:scalability_acc_cost}.
Across all composition sizes, SHT-based agent with \sys is accurate
(0.71---0.73 on average) at low cost (\$12--\$28 total, up to $28\times$ cheaper than baselines), on par with true SHTs
(0.71---0.74, \$11--\$25).
Its gains widen with size: it is up to  7\% to 19\% to 27\% more accurate than the baselines at sizes 1, $k+1$, and $2(k+1)$; its cost rises only $1.4\times$ as \hide{}size doubles.
Specifically, at \hide{}size 1\hide{} the vanilla-grep agent is 2\% \hide{}{below} vanilla-in-context but nearly 4$\times$ cheaper.
{With the SHT-based agent,}
Deep drops fastest (\hide{}6\%, 19\%, \hide{}27\% lower than \sys) because its upper-level text spans grow with \hide{}size, \hide{}{giving} poor compactness and  \hide{}{overly context}.
{LLM-text and LLM-vision \hide{}{match }\sys at size 1 (within 1\%)\hide{} but cost $6.5\times$ and $14\times$ more\hide{}{; their} accuracy drops \hide{}7\% and 6\% at $k+1$\hide{} and \hide{}another 5\% and 6\% at $2(k+1)$, as inferring an SHT from the full concatenated document \hide{}{gets} harder with longer contexts.}
{GROBID \hide{}{improves} on larger compositions due to artificial \hide{}merging: it infers an SHT \hide{}{per document} and merges them under a virtual root (\Cref{sec:setup_sht_eval})\hide{}{;}} \hide{}{i}t stays uniformly low, yielding barely usable structure at any \hide{}size.
\sys's gain over Wide does not widen \hide{}{much} (4\%, 8\%, \hide{}6\%), likely because Wide's per-node text spans stay small and its flat SHT (only \textasciitilde0.5\% of GPT-5.4's 1M window even at $2(k+1)$) still \hide{}{lets the agent} infer a partial hierarchy.
}

\begin{table*}[t]
    \centering
    \caption{\keep{$M_{\downarrow}$ and $M_{\uparrow}$ results at composition size $1$ (a single original document), in the same format as \autoref{tab:sht_eval_results}.}}
    \label{tab:sht_eval_results_comp1}
    \footnotesize
    \setlength{\tabcolsep}{1.8pt}
    \renewcommand{\arraystretch}{0.92}
    \keep{%
    \begin{tabular}{@{}lrrrrr|rrrrr|rrrrr@{}}
    \toprule
    \multirow{2}{*}{\bf Approach}
    & \multicolumn{5}{c}{\bf Recall (Robustness)}
    & \multicolumn{5}{c}{\bf Precision (Compactness)}
    & \multicolumn{5}{c}{\bf F-1 (Trade-off)}\\
    \cmidrule(lr){2-6}
    \cmidrule(lr){7-11}
    \cmidrule(lr){12-16}
    {} & {\bf \ds{\footnotesize Civic}} & {\bf \ds{\footnotesize Contracts}} & {\bf\ds{\footnotesize Finance}} & {\bf\ds{\footnotesize Papers}} & {\bf Avg.}
       & {\bf \ds{\footnotesize Civic}} & {\bf \ds{\footnotesize Contracts}} & {\bf\ds{\footnotesize Finance}} & {\bf\ds{\footnotesize Papers}} & {\bf Avg.}
       & {\bf \ds{\footnotesize Civic}} & {\bf \ds{\footnotesize Contracts}} & {\bf\ds{\footnotesize Finance}} & {\bf\ds{\footnotesize Papers}} & {\bf Avg.}\\
    \midrule
\multicolumn{16}{l}{\cellcolor{gray!10} \em Top-down: {$M_{\downarrow}(T) = avg_{v\in T'}f(v)$, where $f(v)$ calculates recall, precision, and F-1 score of $ts(v)$ relative to $ts'(v)$.}} \\
    Deep
    & 0.93 & \best{0.99} & \best{0.99} & \best{0.97} & \best{0.97}
    & 0.28 & 0.29 & 0.15 & 0.35 & 0.27
    & 0.37 & 0.37 & 0.19 & 0.44 & 0.34 \\
    Wide
    & 0.67 & 0.81 & 0.80 & 0.78 & 0.76
    & 0.92 & \best{0.94} & \best{0.95} & \best{0.97} & \best{0.95}
    & 0.69 & 0.79 & 0.79 & 0.79 & 0.76 \\
    GROBID
    & 0.35 & 0.79 & 0.71 & 0.78 & 0.66
    & 0.17 & 0.34 & 0.71 & 0.83 & 0.51
    & 0.18 & 0.38 & 0.62 & 0.77 & 0.49 \\
    LLM-text
    & \best{0.96} & 0.75 & 0.92 & \best{0.97} & 0.90
    & \best{0.97} & 0.60 & 0.89 & \best{0.97} & 0.86
    & \best{0.96} & 0.58 & 0.86 & \best{0.96} & 0.84 \\
    LLM-vision
    & 0.93 & 0.77 & 0.95 & 0.94 & 0.90
    & 0.93 & 0.61 & 0.72 & 0.96 & 0.80
    & 0.93 & 0.61 & 0.68 & 0.95 & 0.79 \\
    \sys
    & 0.89 & 0.97 & 0.93 & \best{0.97} & 0.94
    & 0.89 & 0.92 & 0.93 & 0.95 & 0.93
    & 0.88 & \best{0.92} & \best{0.88} & 0.94 & \best{0.91} \\
    \midrule
\multicolumn{16}{l}{\cellcolor{gray!10} \em Bottom-up: {$M_{\uparrow}(T) = avg_{v\in T'}f(v)$, where $f(v)$ calculates recall, precision, and F-1 score of $H(v)$ relative to $H'(v)$.}} \\
    Deep
    & \best{0.94} & \best{0.89} & \best{0.96} & \best{0.87} & \best{0.91}
    & 0.36 & 0.56 & 0.24 & 0.64 & 0.45
    & 0.48 & 0.56 & 0.31 & 0.68 & 0.51 \\
    Wide
    & 0.31 & 0.43 & 0.26 & 0.26 & 0.31
    & \best{0.94} & \best{0.95} & \best{0.97} & \best{0.96} & \best{0.96}
    & 0.46 & 0.54 & 0.38 & 0.34 & 0.43 \\
    GROBID
    & 0.16 & 0.46 & 0.35 & 0.76 & 0.43
    & 0.30 & 0.71 & 0.55 & 0.92 & 0.62
    & 0.20 & 0.52 & 0.38 & 0.81 & 0.48 \\
    LLM-text
    & 0.90 & 0.58 & 0.72 & 0.82 & 0.76
    & 0.73 & 0.62 & 0.79 & 0.88 & 0.76
    & 0.80 & 0.57 & \best{0.73} & 0.82 & 0.73 \\
    LLM-vision
    & 0.91 & 0.34 & 0.47 & 0.81 & 0.63
    & 0.76 & 0.73 & 0.77 & \best{0.96} & 0.81
    & 0.82 & 0.43 & 0.55 & 0.85 & 0.66 \\
    \sys
    & 0.92 & 0.84 & 0.69 & 0.84 & 0.83
    & 0.93 & 0.92 & 0.69 & 0.93 & 0.87
    & \best{0.93} & \best{0.85} & 0.65 & \best{0.86} & \best{0.82} \\
    \bottomrule
    \end{tabular}}
\end{table*}

\begin{table*}[t]
    \centering
    \caption{\keep{$M_{\downarrow}$ and $M_{\uparrow}$ results at composition size $2(k+1)$, in the same format as \autoref{tab:sht_eval_results}.}}
    \label{tab:sht_eval_results_comp2k}
    \footnotesize
    \setlength{\tabcolsep}{1.8pt}
    \renewcommand{\arraystretch}{0.92}
    \keep{%
    \begin{tabular}{@{}lrrrrr|rrrrr|rrrrr@{}}
    \toprule
    \multirow{2}{*}{\bf Approach}
    & \multicolumn{5}{c}{\bf Recall (Robustness)}
    & \multicolumn{5}{c}{\bf Precision (Compactness)}
    & \multicolumn{5}{c}{\bf F-1 (Trade-off)}\\
    \cmidrule(lr){2-6}
    \cmidrule(lr){7-11}
    \cmidrule(lr){12-16}
    {} & {\bf \ds{\footnotesize Civic}} & {\bf \ds{\footnotesize Contracts}} & {\bf\ds{\footnotesize Finance}} & {\bf\ds{\footnotesize Papers}} & {\bf Avg.}
       & {\bf \ds{\footnotesize Civic}} & {\bf \ds{\footnotesize Contracts}} & {\bf\ds{\footnotesize Finance}} & {\bf\ds{\footnotesize Papers}} & {\bf Avg.}
       & {\bf \ds{\footnotesize Civic}} & {\bf \ds{\footnotesize Contracts}} & {\bf\ds{\footnotesize Finance}} & {\bf\ds{\footnotesize Papers}} & {\bf Avg.}\\
    \midrule
\multicolumn{16}{l}{\cellcolor{gray!10} \em Top-down: {$M_{\downarrow}(T) = avg_{v\in T'}f(v)$, where $f(v)$ calculates recall, precision, and F-1 score of $ts(v)$ relative to $ts'(v)$.}} \\
    Deep
    & 0.94 & \best{0.99} & \best{1.00} & \best{0.97} & \best{0.97}
    & 0.20 & 0.12 & 0.08 & 0.14 & 0.13
    & 0.28 & 0.16 & 0.10 & 0.20 & 0.19 \\
    Wide
    & 0.72 & 0.82 & 0.82 & 0.78 & 0.78
    & 0.93 & \best{0.90} & \best{0.97} & \best{0.97} & \best{0.94}
    & 0.74 & 0.73 & 0.82 & 0.79 & 0.77 \\
    GROBID
    & 0.55 & 0.84 & 0.76 & 0.76 & 0.73
    & 0.30 & 0.48 & 0.76 & 0.82 & 0.59
    & 0.30 & 0.48 & 0.67 & 0.75 & 0.55 \\
    LLM-text
    & 0.97 & 0.94 & 0.79 & 0.96 & 0.91
    & \best{0.97} & 0.66 & 0.55 & 0.95 & 0.79
    & \best{0.96} & 0.67 & 0.51 & \best{0.95} & 0.77 \\
    LLM-vision
    & \best{0.98} & 0.87 & 0.95 & 0.91 & 0.93
    & 0.94 & 0.55 & 0.76 & 0.84 & 0.77
    & 0.94 & 0.53 & 0.72 & 0.83 & 0.76 \\
    \sys
    & 0.90 & 0.97 & 0.94 & 0.95 & 0.94
    & 0.91 & 0.88 & 0.96 & 0.95 & 0.92
    & 0.90 & \best{0.86} & \best{0.91} & 0.93 & \best{0.90} \\
    \midrule
\multicolumn{16}{l}{\cellcolor{gray!10} \em Bottom-up: {$M_{\uparrow}(T) = avg_{v\in T'}f(v)$, where $f(v)$ calculates recall, precision, and F-1 score of $H(v)$ relative to $H'(v)$.}} \\
    Deep
    & \best{0.94} & \best{0.91} & \best{0.99} & \best{0.85} & \best{0.92}
    & 0.27 & 0.20 & 0.11 & 0.22 & 0.20
    & 0.37 & 0.22 & 0.15 & 0.29 & 0.26 \\
    Wide
    & 0.32 & 0.41 & 0.27 & 0.25 & 0.31
    & \best{0.94} & \best{0.94} & \best{0.98} & \best{0.94} & \best{0.95}
    & 0.46 & 0.53 & 0.39 & 0.34 & 0.43 \\
    GROBID
    & 0.22 & 0.48 & 0.39 & 0.73 & 0.45
    & 0.42 & 0.73 & 0.59 & 0.91 & 0.66
    & 0.28 & 0.52 & 0.41 & 0.78 & 0.50 \\
    LLM-text
    & 0.76 & 0.61 & 0.38 & 0.78 & 0.63
    & 0.72 & 0.76 & 0.57 & 0.93 & 0.74
    & 0.72 & 0.64 & 0.42 & 0.81 & 0.65 \\
    LLM-vision
    & 0.86 & 0.41 & 0.47 & 0.68 & 0.60
    & \best{0.94} & 0.73 & 0.77 & 0.85 & 0.82
    & 0.88 & 0.48 & 0.54 & 0.72 & 0.66 \\
    \sys
    & 0.93 & 0.81 & 0.70 & 0.81 & 0.81
    & 0.93 & 0.92 & 0.68 & 0.92 & 0.87
    & \best{0.93} & \best{0.83} & \best{0.64} & \best{0.84} & \best{0.81} \\
    \bottomrule
    \end{tabular}}
\end{table*}

\begin{figure*}[t]
\centering
    \begin{subfigure}[b]{0.23\linewidth}
         \centering
         \includegraphics[width=1\textwidth]{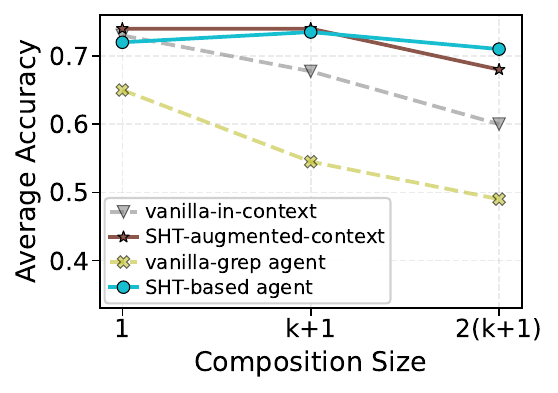}
         \caption{\keep{Accuracy: usefulness of SHTs.}}
         \label{fig:scalability_e2e}
    \end{subfigure}
    \begin{subfigure}[b]{0.23\linewidth}
         \centering
         \includegraphics[width=1\textwidth]{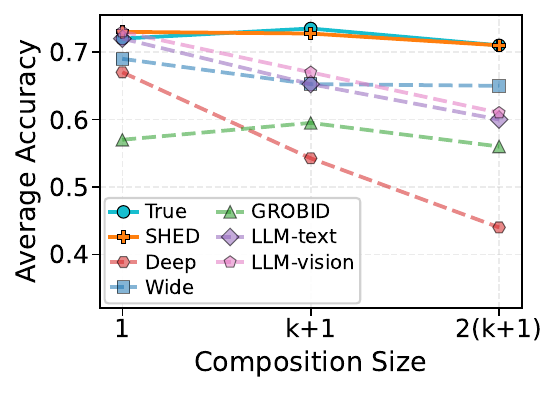}
         \caption{\keep{Accuracy: SHT ablation.}}
         \label{fig:scalability_sht}
    \end{subfigure}
    \begin{subfigure}[b]{0.23\linewidth}
         \centering
         \includegraphics[width=1\textwidth]{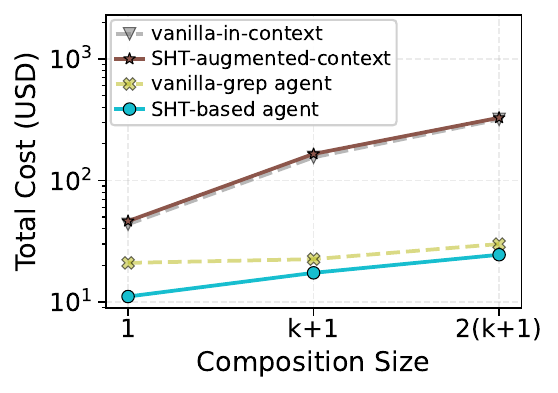}
         \caption{\keep{Cost: usefulness of SHTs.}}
         \label{fig:scalability_e2e_cost}
    \end{subfigure}
    \begin{subfigure}[b]{0.23\linewidth}
         \centering
         \includegraphics[width=1\textwidth]{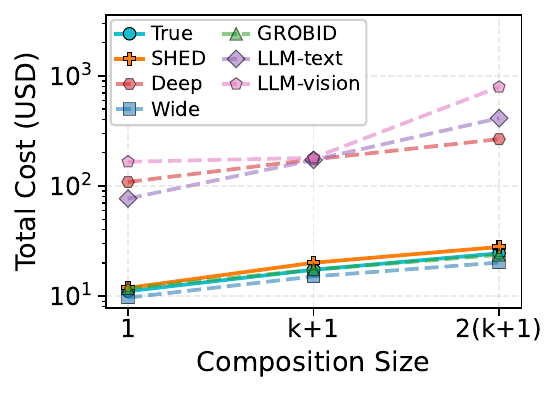}
         \caption{\keep{Cost: SHT ablation.}}
         \label{fig:scalability_sht_cost}
    \end{subfigure}
    \caption{\keep{Average accuracy and total cost (USD) as the composition size (i.e., the number of concatenated documents) grows from the original (1) to twice the size used in \Cref{sec:doc_qa} ($2(k+1)$),
    where $k$ varies by dataset (\autoref{tab:append_datasets} in Appendix~\ref{append:datasets}). The panels
    replicate the experiments in Sections~\ref{sec:eval_qa_sht} (Figures~\ref{fig:scalability_e2e}, \ref{fig:scalability_e2e_cost})
    and~\ref{sec:eval_qa_sht_ablation} (Figures~\ref{fig:scalability_sht}, \ref{fig:scalability_sht_cost}).}}
    \label{fig:scalability_acc_cost}
\end{figure*}

\keep{\subsection{Scalability Analysis: Complex Documents with Synthetic Q\&A}
\label{append:scalability_new_datasets}}

\keep{In \Cref{subsec:scalability_complex_docs}, we briefly described the evaluation results of \sys and baselines on three new datasets (\ds{CFR}, \ds{ETSI}, and \ds{FERC}) that contain individually complex documents without concatentation.
Here, we describe the details of the three new datasets, including their document collection, their Q\&A synthesis (summarized in \autoref{tab:append_complex_datasets}), as well as the QA evaluation metric and detailed results and analysis, as summarized in \autoref{tab:append_complex_datasets}.}

\topic{\keep{Document collection}}
\keep{We use dataset-specific criteria to select documents with complex structures: for \ds{CFR}, we select volumes with at least 6 agency sections each; for \ds{ETSI}, we select documents with at least 40 pages. Documents in \ds{FERC} are from 2020--2025.
{Each document has a table of contents embedded in its PDF, which we extract via PyMuPDF's \ttt{get\_toc} as the true SHT.}}

\topic{\keep{Q\&A synthesis}}
\keep{For each document, we design query templates that require document's structure, and then instantiate each template with entities in the document to generate queries.
For example, documents in \ds{CFR} have a three-level structure: each document contains multiple agencies as sections, each agency contains multiple regimes as subsections, and each regime contains multiple regulations as subsubsections (agency $\rightarrow$ regime $\rightarrow$ regulation).
Therefore, we design a query template that asks for all agencies whose regulations in a given regime $R$ relate to a given hypothesis $H$ the same way as a given agency $X$.
For each document, we select a regime $R$, a regulation $H$ as the hypothesis, and an agency $X$, to generate a query for that document.
Queries for \ds{ETSI} and \ds{FERC} are generated similarly, as shown in \autoref{tab:append_complex_datasets}.
}

\topic{\keep{Evaluation metric}}
\keep{Each query requests a set of entities (e.g., agency names, as shown in \autoref{tab:append_complex_datasets}). To evaluate answer accuracy, we split the generated answer on line breaks to extract entities, normalize them (canonicalizing characters, lowercasing, and removing punctuation and whitespace), and compute the F-1 score against the ground-truth entity set.}

\topic{\keep{Results}}
\keep{
\autoref{tab:scalability_complex_docs_sht} shows that \sys outperforms Wide by 8\%--15\% (16\%--26\%) in recall and 1\%--13\% (12\%--20\%) in F-1 across datasets under $M_{\downarrow}$ ($M_{\uparrow}$).
\autoref{tab:scalability_complex_docs_acc_cost} shows that \sys's QA accuracy is
7\%--18\% higher on average than {the 3 baselines\hide{}{,}}
\hide{}{at nearly the cheapest cost (Wide's) and $15\times$ below} the costliest (vanilla-in-context).
\hide{}{On} \ds{CFR}, \sys still achieves good QA accuracy despite low header\hide{}{-}identification recall: the missed headers are mostly regulations (the lowest level, \textasciitilde1.6K per document), whose content the agent recovers \hide{}{from the parent section's (regime's) text span}.
\ds{FERC} shows low accuracy across strategies due to query \hide{}{complexity}: \hide{}{queries require} finding all environmental-effect sections and classifying experts' sentiments \hide{}{from} all corresponding analyses (subsection{s}).
}

\begin{table*}[t]
    \centering
    \caption{\keep{Three new datasets, collected from inherently complex real-world documents. Avg. Doc Size: average word count per document in each dataset. \#Docs (\#Qs): number of documents (questions) in each dataset. {Structure: document's hierarchical structure, represented as section $\rightarrow$ subsection $\rightarrow$ subsubsection.} Query template: natural-language template instantiated with entities in documents.}}
    \label{tab:append_complex_datasets}
    \footnotesize
    \setlength{\tabcolsep}{2.6pt}
    \renewcommand{\arraystretch}{1.2}
    \keep{
    \begin{tabular}{@{}
    >{\raggedright\arraybackslash}p{0.09\linewidth}
    >{\raggedright\arraybackslash}p{0.04\linewidth}
    >{\raggedright\arraybackslash}p{0.1\linewidth}
    >{\raggedright\arraybackslash}p{0.04\linewidth}
    >{\raggedright\arraybackslash}p{0.15\linewidth}
    >{\raggedright\arraybackslash}p{0.4\linewidth}
    @{}}
    \toprule
    {\bf Datasets} & {\bf \#Docs} & {\bf Avg. Doc Size} & {\bf \#Qs} & {\bf Structure} &  {\bf Query Template}\\
    \midrule
    \ds{CFR}~\cite{cfr}  & 30 & 415,115 & 40 & agency $\rightarrow$ regime $\rightarrow$ regulation & {\em List all agencies whose regulations in regime $R$ (e.g., FOIA, Privacy) relate to hypothesis $H$ the same way (Entailment / Contradiction / NotMentioned) as agency $X$.}\\
    \ds{ETSI}~\cite{etsi} & 49 &  30,485 & 49 & equipment category $\rightarrow$ characteristic $\rightarrow$ attribute & {\em List all characteristics, spanning $\geq\!2$ equipment categories (transmitter / receiver / duplex), whose attribute $A$ (Definition / Method-of-measurement / Limits) satisfies predicate $P$.}\\
    \ds{FERC}~\cite{ferc} & 40 &  49,126 & 53 &  resource type $\rightarrow$ environmental effect $\rightarrow$ staff analysis & {\em List all environmental effects whose resource type (e.g., aquatic, terrestrial, T\&E species, recreation, cultural) differs from that of effect $X$ but whose staff-analysis sentiment (positive / negative / neutral) is the same as $X$.}\\
    \bottomrule
    \end{tabular}}
\end{table*}

\begin{table}[t]
    \centering
    \caption{\keep{Robustness (Recall), compactness (Precision), and their F-1 of SHTs inferred by Wide and \sys on the three new datasets in \Cref{subsec:scalability_complex_docs}. \best{Green}: better of the two approaches per dataset per metric.}}
    \label{tab:scalability_complex_docs_sht}
    \footnotesize
    \setlength{\tabcolsep}{0.8pt}
    \renewcommand{\arraystretch}{0.92}
    \keep{%
    \begin{tabular}{@{}l rrrr | rrrr | rrrr@{}}
    \toprule
    \multirow{2}{*}{\bf Approach}
    & \multicolumn{4}{c}{\bf Recall}
    & \multicolumn{4}{c}{\bf Precision}
    & \multicolumn{4}{c}{\bf F-1} \\
    \cmidrule(lr){2-5} \cmidrule(lr){6-9} \cmidrule(lr){10-13}
    {} & {\bf \ds{\footnotesize CFR}} & {\bf \ds{\footnotesize ETSI}} & {\bf \ds{\footnotesize FERC}} & {\bf Avg.}
       & {\bf \ds{\footnotesize CFR}} & {\bf \ds{\footnotesize ETSI}} & {\bf \ds{\footnotesize FERC}} & {\bf Avg.}
       & {\bf \ds{\footnotesize CFR}} & {\bf \ds{\footnotesize ETSI}} & {\bf \ds{\footnotesize FERC}} & {\bf Avg.} \\
    \midrule
    \multicolumn{13}{l}{\cellcolor{gray!10}\em Top-down: $M_{\downarrow}(T)$} \\
    Wide
      & 0.71 & 0.77 & 0.68 & 0.72
      & \best{0.47} & \best{0.96} & \best{0.93} & \best{0.78}
      & 0.42 & 0.76 & 0.69 & 0.62 \\
    \sys
      & \best{0.79} & \best{0.89} & \best{0.83} & \best{0.84}
      & 0.43 & 0.95 & 0.91 & 0.76
      & \best{0.43} & \best{0.89} & \best{0.79} & \best{0.70} \\
    \multicolumn{13}{l}{\cellcolor{gray!10}\em Bottom-up: $M_{\uparrow}(T)$} \\
    Wide
      & 0.19 & 0.19 & 0.43 & 0.27
      & \best{0.75} & \best{0.98} & \best{0.91} & \best{0.88}
      & 0.28 & 0.29 & 0.54 & 0.37 \\
    \sys
      & \best{0.35} & \best{0.45} & \best{0.63} & \best{0.48}
      & 0.58 & 0.63 & 0.90 & 0.70
      & \best{0.40} & \best{0.49} & \best{0.71} & \best{0.54} \\
    \bottomrule
    \end{tabular}}
\end{table}

\begin{table}[t]
\centering
\caption{\keep{Average QA accuracy and total cost (USD, inside parenthesis) on the three new datasets in \Cref{subsec:scalability_complex_docs}. \best{Green}: highest accuracy (lowest cost) per dataset.}}
\label{tab:scalability_complex_docs_acc_cost}
\footnotesize
\setlength{\tabcolsep}{1.8pt}
\renewcommand{\arraystretch}{0.92}
\keep{
\begin{tabular}{@{}
l@{\hspace{2pt}}
>{\raggedleft\arraybackslash}p{0.065\linewidth}@{\hspace{-1pt}}
>{\raggedleft\arraybackslash}p{0.105\linewidth}@{\hspace{4pt}}
>{\raggedleft\arraybackslash}p{0.065\linewidth}@{\hspace{-1pt}}
>{\raggedleft\arraybackslash}p{0.075\linewidth}@{\hspace{4pt}}
>{\raggedleft\arraybackslash}p{0.065\linewidth}@{\hspace{-1pt}}
>{\raggedleft\arraybackslash}p{0.085\linewidth}@{\hspace{4pt}}
>{\raggedleft\arraybackslash}p{0.065\linewidth}@{\hspace{-1pt}}
>{\raggedleft\arraybackslash}p{0.105\linewidth}
@{}}
\toprule
{\bf Strategy}
& \multicolumn{2}{r}{\bf \ds{\footnotesize CFR}}
& \multicolumn{2}{r}{\bf \ds{\footnotesize ETSI}}
& \multicolumn{2}{r}{\bf \ds{\footnotesize FERC}}
& \multicolumn{2}{r}{{\bf Avg. (Tot.)}} \\
\midrule
{Vanilla-in-context}
& 0.65 & (111.1) & 0.87 & (4.8) & 0.21 & (12.3) & 0.58 & (128.3) \\
Vanilla-grep agent
& 0.67 & (4.8) & 0.66 & (6.5) & 0.23 & (18.0) & 0.52 & (29.2) \\
SHT-based agent (Wide)
& 0.71 & (\best{3.6}) & 0.83 & (\best{1.1}) & 0.34 & (3.0) & 0.63 & (\best{7.7}) \\
SHT-based agent (\sys)
& \best{0.78} & (4.3) & \best{0.90} & (1.7) & \best{0.42} & (\best{2.4}) & \best{0.70} & (8.5) \\
\bottomrule
\end{tabular}}
\end{table}

\keep{\subsection{Scalability Analysis: SHT Inference Latency}
\label{append:scalability_sht_infer}
We test SHT inference latency for the approaches in \Cref{sec:sht_eval},
plus
SmolDocling-256M-preview
(Appendix~\ref{append:smoldocling}, a new learning-based baseline that jointly models the layout and semantics of a PDF),
since running the costlier methods (e.g., LLM-vision) on large corpora exceeds our budget---scaling to 1M documents like those we evaluate would cost
$>$250K USD---we instead evaluate
on 80 composite documents (20 sampled from each of the four datasets in \autoref{tab:dataset}) and report the estimated latency over all 955 documents in \autoref{tab:sht_inference_latency}.
Overall,
\sys needs 5 hours to process millions of tokens, showing good scalability; employing parallelism would further improve scalability.}

\topic{\keep{Testbed}}
\keep{{We run the large model~\cite{vgt-code} used by \sys (for header identification on \ds{Civic} and \ds{Contracts}; Sec.~\ref{sec:sys_impl}) and SmolDocling-256M-preview (Appendix~\ref{append:smoldocling})
on a GPU machine with one NVIDIA L4 GPU with 24 GiB, 4 vCPUs, and 16 GiB RAM.
We run
all other experiments that do not involve local neural-network inference on a CPU machine with one 2.20 GHz Intel Xeon (8 cores, 16 threads)
and 15 GiB RAM, including
\sys's smaller model~\cite{lightgbm} for header
identification on \ds{Finance} and \ds{Papers}, visual pattern extraction and SHT assembly, and GROBID, LLM-text,
LLM-vision.}}

\topic{\keep{Results}}
\keep{Each approach processes documents sequentially.
\autoref{tab:sht_inference_latency} shows that GROBID, LLM-text, and LLM-vision infer SHTs for all 955 documents in 4, 2, and 5 hours.
However, GROBID produces low-quality SHTs (poor robustness and compactness; \autoref{tab:sht_eval_results});
LLM-text and LLM-vision do not scale in cost (the LLM-based ones would run \textasciitilde\$0.2M per 1M documents).
\sys takes about 5 hours to infer SHTs for all 955 documents (in total nearly 54K pages and 30M tokens) sequentially, showing good scalability. Specifically, 90\%--94\% of time is spent on header identification by Huridocs~\cite{vgt-code,lightgbm}, the rest visual-pattern extraction and SHT construction.
Deep and Wide should be similar to \sys in latency, since they only skip visual-pattern extraction (which accounts for up to 10\% of \sys's latency), but their SHTs are far less robust or compact.
SmolDocling-256M-preview would take \textasciitilde11.4 days, far slower than \sys, since it runs more complex models to build SHTs end-to-end rather than just for header identification.}

\begin{table}[t]
    \centering
    \caption{\keep{SHT inference latency (average seconds per document), measured on 20 sampled composite documents per dataset. {\bf Est. Tot.}: total latency (hours) for all 955 documents in \autoref{tab:dataset}, as extrapolated estimates via the samples.
    }}
    \label{tab:sht_inference_latency}
    \footnotesize
    \setlength{\tabcolsep}{1.8pt}
    \renewcommand{\arraystretch}{0.92}
    \keep{%
    \begin{tabular}{@{}l rrrr | r@{}}
    \toprule
    {\bf Approach}
    & {\bf \ds{\footnotesize Civic}} & {\bf \ds{\footnotesize Contracts}}
    & {\bf \ds{\footnotesize Finance}} & {\bf \ds{\footnotesize Papers}}
    & {\bf Est. Tot. (h)} \\
    \midrule
    GROBID       &     14 &    15 &      52 &     8 &       4 \\
    LLM-text     &      8 &     5 &      33 &     3 &       2 \\
    LLM-vision   &     13 &     9 &      93 &    11 &       5 \\
    \sys         &    22 &   18 &      48 &    13 &      5 \\
    SmolDocling-256M-preview  & 229 & 308 & 5{,}074 & 747 & 273 \\
    \bottomrule
    \end{tabular}}
\end{table}

\keep{\subsection{An Advanced SHT Inference Baseline}
\label{append:smoldocling}
For approaches that jointly model PDF layout and semantics to infer structure, we've already evaluated LLM-vision (GPT-5.4) in Sections~\ref{sec:sht_eval} and \ref{sec:doc_qa}, which is state-of-the-art---e.g., on DocVQA, GPT-4o scores 92.8\%~\cite{docvqa_leaderboard}, {beating fine-tuned} LayoutLMv3 (83.4\%)~\cite{huang2022layoutlmv3} by over 9\%; our baseline uses the \hide{}stronger GPT-5.4.
To broaden our baselines with this line of approaches,
we further evaluate SmolDocling-256M-preview~\cite{nassar2025smoldocling}, which is a learning-based approach that generates a Markdown to encode hierarchical structure.
Due to resource constraints, we evaluate only the 80 documents from Appendix~\ref{append:scalability_sht_infer}.
}

\keep{\autoref{tab:sht_eval_results_smoldocling} reports the top-down and bottom-up recall, precision, and F-1 of its SHTs, which average up to 56\% below \sys.
Possible causes are: (1) its inferred SHTs are mostly flat (78 of the 80 have only 1-level non-root nodes), and (2)
its header identification is less accurate than \sys's (recall/precision/F-1 is 0.82/0.84/0.83 on \ds{Civic}, 0.20/0.74/0.31 on \ds{Contracts}, 0.85/0.46/0.58 on \ds{Finance}, and 0.82/0.82/0.82 on \ds{Papers}).}

\providecommand{\ppd}[1]{{\scriptsize(#1)}}
\begin{table}[t]
    \centering
    \caption{\keep{Robustness (Recall), compactness (Precision), and their F-1 of SHTs
    inferred by SmolDocling-256M-preview, on 80 sampled documents (20 per dataset).
    Parentheses indicate the percentage difference from \sys (\autoref{tab:sht_eval_results}).}}
    \label{tab:sht_eval_results_smoldocling}
    \footnotesize
    \setlength{\tabcolsep}{1.8pt}
    \renewcommand{\arraystretch}{0.92}
    \keep{%
    \begin{tabular}{@{}l
      r@{\hspace{1pt}}l r@{\hspace{1pt}}l r@{\hspace{1pt}}l |
      r@{\hspace{1pt}}l r@{\hspace{1pt}}l r@{\hspace{1pt}}l@{}}
    \toprule
    \multirow{2}{*}{\bf Dataset}
    & \multicolumn{6}{c}{\bf Top-down: $M_{\downarrow}$}
    & \multicolumn{6}{c}{\bf Bottom-up: $M_{\uparrow}$} \\
    \cmidrule(lr){2-7} \cmidrule(lr){8-13}
    {} & \multicolumn{2}{c}{\bf Recall} & \multicolumn{2}{c}{\bf Precision} & \multicolumn{2}{c}{\bf F-1}
       & \multicolumn{2}{c}{\bf Recall} & \multicolumn{2}{c}{\bf Precision} & \multicolumn{2}{c}{\bf F-1} \\
    \midrule
    {\ds{\footnotesize Civic}}
      & 0.66 & \ppd{$-$24} & 0.88 & \ppd{$-$3}  & 0.69 & \ppd{$-$20}
      & 0.26 & \ppd{$-$66} & 0.86 & \ppd{$-$7}  & 0.40 & \ppd{$-$53} \\
    {\ds{\footnotesize Contracts}}
      & 0.76 & \ppd{$-$21} & 0.43 & \ppd{$-$45} & 0.40 & \ppd{$-$47}
      & 0.24 & \ppd{$-$58} & 0.74 & \ppd{$-$18} & 0.33 & \ppd{$-$50} \\
    {\ds{\footnotesize Finance}}
      & 0.70 & \ppd{$-$24} & 0.89 & \ppd{$-$6}  & 0.70 & \ppd{$-$20}
      & 0.22 & \ppd{$-$47} & 0.88 & \ppd{+17}   & 0.33 & \ppd{$-$32} \\
    {\ds{\footnotesize Papers}}
      & 0.63 & \ppd{$-$33} & 0.83 & \ppd{$-$12} & 0.66 & \ppd{$-$28}
      & 0.26 & \ppd{$-$56} & 0.87 & \ppd{$-$6}  & 0.35 & \ppd{$-$49} \\
    \cmidrule(lr){1-13}
    {\bf Avg.}
      & 0.69 & \ppd{$-$25} & 0.76 & \ppd{$-$16} & 0.61 & \ppd{$-$29}
      & 0.25 & \ppd{$-$56} & 0.84 & \ppd{$-$3}  & 0.35 & \ppd{$-$46} \\
    \bottomrule
    \end{tabular}}
\end{table}

\keep{\subsection{Sensitivity of \sys}
\label{append:global_first_sht_infer}}

\keep{To assess \sys's sensitivity to the SD inference approach (\Cref{sec:sht-construction}), we compare the robustness and compactness of SHTs inferred by global-first and local-first (our previous implementation; \Cref{sec:sys_impl}) on the four datasets in \Cref{tab:dataset} (\autoref{tab:sht_eval_results_global_first}).
{To further justify our choice of local-first, we evaluated both approaches on 40 real-world documents from SAFEDOCS~\cite{safedocs} in \Cref{subsec:hierarchy-robustness} where local-first guarantees robustness while global-first does not (\autoref{tab:sht_eval_results_global_first_local_first_comp}).}
}

\topic{\keep{Results}}
\keep{\autoref{tab:sht_eval_results_global_first} shows that global- and local-first achieve nearly identical robustness and compactness across datasets on $M_{\uparrow}$ and $M_{\downarrow}$ (within 1\% difference in average recall, precision, and F-1), indicating \sys's stability across the two approaches over the four datasets. This is partly because over 60\% of documents are loosely-formatted, for which \sys guarantees robustness under {\em any} SD inference approach in \Cref{subsec:family-infer-sd} (Theorem~\ref{thm:loosely-robustness}).
{\autoref{tab:sht_eval_results_global_first_local_first_comp} shows that, over 40 real-world documents other than the ones in the four datasets, local-first guarantees robustness (recall=1) while global-first does not, and local-first achieves 4\% (5\%) higher average F-1 than global-first on $M_{\downarrow}$ ($M_{\uparrow}$). This further justifies our choice of local-first for \sys in \Cref{sec:sys_impl}.}}

\providecommand{\ppd}[1]{{\scriptsize(#1)}}
\begin{table}[t]
    \centering
    \caption{\keep{Robustness (Recall), compactness (Precision), and their trade-off (F-1) of SHTs inferred by \sys with the global-first SD inference. Parentheses: difference from local-first (\autoref{tab:sht_eval_results}), in percentage points. {\bf Avg.}: stratified averages across datasets.}}
    \label{tab:sht_eval_results_global_first}
    \footnotesize
    \setlength{\tabcolsep}{1.8pt}
    \renewcommand{\arraystretch}{0.92}
    \keep{%
    \begin{tabular}{@{}l
      r@{\hspace{1pt}}l r@{\hspace{1pt}}l r@{\hspace{1pt}}l |
      r@{\hspace{1pt}}l r@{\hspace{1pt}}l r@{\hspace{1pt}}l@{}}
    \toprule
    \multirow{2}{*}{\bf Dataset}
    & \multicolumn{6}{c}{\bf Top-down: $M_{\downarrow}$}
    & \multicolumn{6}{c}{\bf Bottom-up: $M_{\uparrow}$} \\
    \cmidrule(lr){2-7} \cmidrule(lr){8-13}
    {} & \multicolumn{2}{c}{\bf Recall} & \multicolumn{2}{c}{\bf Precision} & \multicolumn{2}{c}{\bf F-1}
       & \multicolumn{2}{c}{\bf Recall} & \multicolumn{2}{c}{\bf Precision} & \multicolumn{2}{c}{\bf F-1} \\
    \midrule
    {\ds{\footnotesize Civic}}
      & 0.90 & \ppd{+0} & 0.91 & \ppd{+0} & 0.89 & \ppd{+0}
      & 0.92 & \ppd{+0} & 0.93 & \ppd{+0} & 0.93 & \ppd{+0} \\
    {\ds{\footnotesize Contracts}}
      & 0.97 & \ppd{+0} & 0.88 & \ppd{+0} & 0.87 & \ppd{+0}
      & 0.82 & \ppd{+0} & 0.92 & \ppd{+0} & 0.84 & \ppd{+1} \\
    {\ds{\footnotesize Finance}}
      & 0.94 & \ppd{+0} & 0.95 & \ppd{+0}   & 0.90 & \ppd{+0}
      & 0.69 & \ppd{+0} & 0.70 & \ppd{$-$1} & 0.65 & \ppd{+0} \\
    {\ds{\footnotesize Papers}}
      & 0.96 & \ppd{+0} & 0.95 & \ppd{+0} & 0.94 & \ppd{+0}
      & 0.82 & \ppd{+0} & 0.93 & \ppd{+0} & 0.84 & \ppd{+0} \\
    \cmidrule(lr){1-13}
    {\bf Avg.}
      & 0.94 & \ppd{+0} & 0.92 & \ppd{+0} & 0.90 & \ppd{+0}
      & 0.81 & \ppd{+0} & 0.87 & \ppd{+0} & 0.82 & \ppd{+1} \\
    \bottomrule
    \end{tabular}}
\end{table}

\begin{table}[t]
    \centering
    \caption{\keep{Robustness (Recall), compactness (Precision), and their F-1 of SHTs inferred by \sys with local- and global-first SD inference over 40 real-world documents~\cite{safedocs}.}}
    \label{tab:sht_eval_results_global_first_local_first_comp}
    \footnotesize
    \setlength{\tabcolsep}{1.8pt}
    \renewcommand{\arraystretch}{0.92}
    \keep{%
    \begin{tabular}{@{}l rrr | rrr@{}}
    \toprule
    \multirow{2}{*}{\bf Approach}
    & \multicolumn{3}{c}{\bf Top-down: $M_{\downarrow}$}
    & \multicolumn{3}{c}{\bf Bottom-up: $M_{\uparrow}$} \\
    \cmidrule(lr){2-4} \cmidrule(lr){5-7}
    {} & {\bf Recall} & {\bf Precision} & {\bf F-1}
       & {\bf Recall} & {\bf Precision} & {\bf F-1} \\
    \midrule
    global-first & 0.97 & 0.94 & 0.92 & 0.94 & 0.81 & 0.84 \\
    local-first  & 1.00 & 0.96 & 0.96 & 1.00 & 0.84 & 0.89 \\
    \bottomrule
    \end{tabular}}
\end{table}

\section{Prompts}

\subsection{LLM-based SHT Inference}
\label{append:prompts-sht}

LLM-text use the following prompt to infer SHT given serialized document text:
\begin{lstlisting}[language=text,style=plaintextstyle]
======================== System Prompt ========================
You are given a document. Your task is to extract and organize its Table of Contents (ToC).

Instructions:

1. Identify **all headings at every level** in the document (e.g., title, sections, subsections, sub-subsections, and any deeper levels).
2. **Use the exact wording of each heading as it appears** - do NOT paraphrase, shorten, or modify.
3. Preserve the original capitalization, punctuation, and numbering.
4. Infer the hierarchy levels based on formatting, numbering, or context.
5. Maintain the original order of appearance.
6. Do NOT invent, merge, or omit any headings.
7. Exclude any text that is not a heading.

Formatting requirements:

* Use the following Markdown format to represent hierarchy:

  ```markdown
  # Title  
  ## Section Header  
  ### Subsection Header  
  #### Sub-subsection Header  
  ```

* The number of `#` symbols indicates the hierarchy level.

* Include exactly one space between the `#` symbols and the header text.

* Extend to deeper levels as needed (e.g., `#####`, `######`, etc.).

* Do NOT skip levels (e.g., do not jump from `#` to `###`).

* Output only the Table of Contents in Markdown format.

* Do NOT include any explanations or extra text.

========================= User Prompt =========================
DOCUMENT:
{{ doc_txt }}

TOC IN MARKDOWN:
\end{lstlisting}

The prompt for LLM-vision is largely the same, with two modifications:
\begin{enumerate}[leftmargin=*, nosep]
\item the first sentence of the system prompt is changed to: ``\code{You are given images of a document's pages. Your task is to extract and organize a Table of Contents (ToC) of the document from the images}'';
\item the placeholder \code{\{\{ doc\_txt \}\}} is replaced with a batch of page images in reading order.
\end{enumerate}

\subsection{LLM-as-judge}
\label{append:prompts-llm-as-a-judge}

Below is the prompt for LLM-as-judge on \ds{Finance}:
\begin{lstlisting}[language=text,style=plaintextstyle]
======================== System Prompt ========================
You are an LLM-as-judge. Grade a candidate answer against a ground-truth reference answer for a finance question.

INPUTS (verbatim):
- QUESTION
- REFERENCE_ANSWER
- CANDIDATE_ANSWER

HARD RULES
- Use ONLY QUESTION, REFERENCE_ANSWER, CANDIDATE_ANSWER. No outside knowledge.
- Output ONE JSON object only (no extra text).
- Scoring must follow the rubric below exactly.

OUTPUT JSON SCHEMA (strict)
{
  "score": number,                        // final in [0, 1]
  "question_type": "numeric" | "boolean" | "list" | "explanatory" | "mixed",
  "extracted_candidate_answer": {
    "text_span": string,
    "normalized_value": number | null,
    "normalized_units": "usd_millions" | "usd_billions" | "percent" | "ratio" | "days" | "text" | "none" | null
  },
  "reference_answer": {
    "text": string,
    "normalized_value": number | null,
    "normalized_units": "usd_millions" | "usd_billions" | "percent" | "ratio" | "days" | "text" | "none" | null
  },
  "subscores": {
    "polarity_or_core": number,           // [0,1]
    "numeric_accuracy": number,           // [0,1]
    "units_and_scale": number,            // [0,1]
    "completeness": number,               // [0,1]
    "no_contradictions": number           // [0,1]
  },
  "reasons": [string],
  "confidence": number                   // in [0,1]
}

QUESTION TYPE CLASSIFICATION
- "numeric": asks for a single numeric value/ratio/percent/days/currency amount.
- "boolean": asks Yes/No (or equivalent) with or without brief justification.
- "list": asks to list items (e.g., securities, acquisitions, geographies).
- "explanatory": asks for drivers/why/what caused, narrative required.
- "mixed": combines boolean + explanation or numeric + explanation.

NUMERIC NORMALIZATION (use for numeric parts)
- Extract the first unambiguous numeric target in CANDIDATE_ANSWER that appears to answer the question.
- Accept: "$1577.00", "1,577", "0.66", "65.4%", "1.9 %".
- Remove commas.
- Percent handling:
  - If percent sign present => units "percent", numeric is the percent number (e.g., "1.9%" -> 1.9).
  - If QUESTION requests percent but candidate omits "%" and also omits the word "percent", treat unit as ambiguous and score units accordingly (do NOT assume percent).
- Scale words MUST be honored:
  - million/mn/m => scale of millions; billion/bn/b => scale of billions.
  - If QUESTION requests USD millions (or billions), candidate must present the answer in that unit; merely giving an equivalent in another unit is NOT acceptable unless candidate explicitly converts AND presents the requested unit as the final answer.
- Rounding:
  - If QUESTION says "Round to two decimal places": candidate must match within +-0.005.
  - If QUESTION says "Round to one decimal place": match within +-0.05.
  - If QUESTION not specify rounding: require exact match to reference after normalization; if reference shows 1 decimal allow +-0.05, if 2 decimals allow +-0.01.

CONTRADICTIONS (global)
- If candidate contains internal contradictions (multiple conflicting finals, "Yes" then "No", different numbers) then set "no_contradictions"=0; otherwise 1.

SCORING RUBRIC (deterministic)
Compute subscores, then weighted sum. Clamp final to [0,1].

A) For question_type="numeric"
- polarity_or_core = 1 (not applicable)
- units_and_scale:
  - 1 if units/scale match QUESTION requirements (and no ambiguity).
  - 0 if wrong units/scale OR ambiguous when units are required.
- numeric_accuracy:
  - If units_and_scale=0 => numeric_accuracy = 0 (do not reward correct magnitude in wrong unit).
  - Else compute accuracy:
    - If within tolerance => 1
    - Else if relative error <= 1% => 0.5
    - Else => 0
- completeness = 1 if exactly one clear final numeric answer is provided; else 0.5 if extra non-conflicting numbers; else 0.
- Final score weights:
  - numeric_accuracy 0.60
  - units_and_scale 0.25
  - completeness 0.10
  - no_contradictions 0.05

B) For question_type="boolean"
- Extract polarity from candidate and reference (Yes/No). If reference polarity not explicit, infer from first sentence.
- polarity_or_core:
  - 1 if polarity matches; 0 if mismatches; 0 if candidate gives no polarity.
- completeness:
  - If REFERENCE_ANSWER includes justification metrics/facts, require candidate to mention the same key items (see "COMPLETENESS RULES" below).
  - Score 1 if all key items present; 0.5 if partially; 0 if none or contradictory.
- units_and_scale, numeric_accuracy:
  - If justification includes numbers in reference, treat missing/wrong numbers as completeness penalties (not numeric_accuracy) unless QUESTION explicitly asks to compute a number.
  - Set units_and_scale=1 and numeric_accuracy=1 unless QUESTION contains an explicit numeric computation request.
- Final weights:
  - polarity_or_core 0.55
  - completeness 0.35
  - no_contradictions 0.10

C) For question_type="list"
- polarity_or_core = 1 (not applicable)
- completeness:
  - Let R = set of required items in REFERENCE_ANSWER.
  - Candidate must include all items in R (paraphrase ok; but must preserve identifying details like rates/due years/symbols when present).
  - Score:
    - 1 if all present and no contradictions.
    - 0.5 if some missing but at least half present.
    - 0 if fewer than half present OR includes contradictory items.
- units_and_scale=1, numeric_accuracy=1 unless the list items are numeric-identifiers that must match (then treat mismatch as completeness reduction).
- Final weights:
  - completeness 0.85
  - no_contradictions 0.15

D) For question_type="explanatory" or "mixed"
- polarity_or_core:
  - If a Yes/No or "metric not relevant" stance exists in REFERENCE_ANSWER, candidate must match it.
  - Score 1 match, 0 mismatch, 0.5 if unclear but not contradictory.
- completeness:
  - Identify "key points" in REFERENCE_ANSWER: each bullet/driver/explicit named factor counts as one key point.
  - Score = (# key points covered by candidate) / (total key points), capped at 1.
  - A key point is "covered" if candidate expresses the same idea, not necessarily same wording.
- units_and_scale, numeric_accuracy:
  - Only apply if the QUESTION explicitly asks for a computed number. Otherwise set both to 1.
- Final weights:
  - polarity_or_core 0.40
  - completeness 0.50
  - no_contradictions 0.10
  - (If numeric asked as well, renormalize weights to include numeric_accuracy 0.35, units_and_scale 0.15, completeness 0.35, polarity_or_core 0.10, no_contradictions 0.05.)

COMPLETENESS RULES (key items)
- If REFERENCE_ANSWER includes explicit metrics (e.g., CAPEX/Revenue 5.1%, Fixed assets/Total Assets 20%, ROA 12.4%), each metric is a separate key point.
- If REFERENCE_ANSWER is a list, each listed item is a key point.
- If REFERENCE_ANSWER says "metric not relevant", that statement itself is a mandatory key point plus the provided rationale is another key point.

REASONS
- Provide brief, specific reasons for any subscore < 1 (e.g., "Unit mismatch (USD billions requested)", "Missing 2 of 3 key metrics", "Polarity mismatch: reference No, candidate Yes").

========================= User Prompt =========================
QUESTION:
{{ query }}

REFERENCE_ANSWER:
{{ gt_answer }}

CANDIDATE_ANSWER:
{{ llm_answer }}

OUTPUT JSON:
\end{lstlisting}

Below is the prompt for LLM-as-judge on \ds{Papers}:
\begin{lstlisting}[language=text,style=plaintextstyle]
======================== System Prompt ========================
You are an automatic grader. Judge whether a MODEL_ANSWER correctly answers a QUESTION, using ONLY the provided set of GROUND_TRUTH_ANSWERS (multiple acceptable references). Do not use outside knowledge.

You must:
1) Compare MODEL_ANSWER to EACH ground-truth answer independently.
2) Assign a score for each comparison using the rubric below.
3) Select the MAXIMUM score across references as FINAL_SCORE (best plausible match).
4) Output a strict JSON object with: final_score, matched_reference_index, per_reference_scores, rationale.

Inputs:
- QUESTION: {question}
- GROUND_TRUTH_ANSWERS: (indexed list) {ground_truth_answers}
- MODEL_ANSWER: {model_answer}

General Rules (apply to all questions):
- Use ONLY the information in GROUND_TRUTH_ANSWERS to decide correctness.
- Ignore differences in capitalization, punctuation, and minor rephrasing.
- Treat synonymous wording as equivalent if the meaning is the same.
- If the question asks for numbers, counts, dataset sizes, metrics, names, or lists: the required items must be present and correct. Small rounding differences are acceptable ONLY if clearly implied by the reference (e.g., 13,014 vs 13 000 tweets) and not contradictory.
- If MODEL_ANSWER includes extra assertions not supported by ANY ground truth, penalize (see rubric).
- If MODEL_ANSWER is empty, refuses, or says it cannot answer while ground truth is answerable, score 0.
- If ANY ground truth answer is exactly "Unanswerable" (case-insensitive) meaning the question is unanswerable from the source, then the only fully correct responses are ones that clearly state it is unanswerable/unknown/not provided. Any attempt to answer with specific facts must be scored 0.

Rubric (score each MODEL_ANSWER vs one reference; then take max):
5 = Fully correct: matches the reference meaning and includes all key required facts; no unsupported extra claims.
4 = Mostly correct: minor omission OR minor imprecision that does not change the meaning; no major unsupported claims.
3 = Partially correct: captures some key facts but misses others, OR includes a minor contradiction, OR adds some unsupported details.
2 = Slightly correct: vaguely related but lacks key facts or is overly incomplete; may include unsupported claims.
1 = Incorrect but on-topic: attempts to answer but key facts are wrong or contradict the reference.
0 = Completely incorrect / unrelated / refusal; OR (when GT is Unanswerable) provides specific factual answer.

How to identify "key required facts":
- For size questions: the exact quantities and units (e.g., dialogs vs sentences) are key.
- For "what metrics/baselines/models/datasets": the set of items is key; missing 1-2 items can reduce score depending on importance.
- For definition/how questions: the described procedure/criteria is key; partial descriptions score lower.

Output format (JSON ONLY, no extra text):
{
  "final_score": <0-5 integer>,
  "matched_reference_index": <integer index of the reference that gave the max score>,
  "per_reference_scores": [<0-5 int>, ...],
  "rationale": "<brief, specific justification citing which key facts matched/missed>"
}

========================= User Prompt =========================
QUESTION:
{{ query }}

GROUND_TRUTH_ANSWERS:
{{ gt_answers }}

MODEL_ANSWER:
{{ llm_answer }}

OUTPUT JSON:
\end{lstlisting}

\end{document}